\documentclass[orivec,envcountsame,a4paper]{llncs}
\usepackage{amsmath}
\usepackage{amssymb,stmaryrd,amsmath,color,proof,extarrows}
\usepackage{sidecap}

\let\originallhook\lhook
\usepackage{MnSymbol}
\newcommand{\lhook}{\mathrel{\raise.018ex\hbox{$\originallhook$}}}
\usepackage{mathtools}

\usepackage{amsthm}

\let\llncssubparagraph\subparagraph
\let\subparagraph\paragraph
\usepackage[compact]{titlesec}
\let\subparagraph\llncssubparagraph

\newcommand\ord[1]{\mathsf{ord}(#1)}

\usepackage{listings}
\usepackage{ifthen}
\usepackage[all]{xy}
\usepackage{alltt}
\usepackage{times}
\usepackage{graphicx}
\usepackage{soul}
\usepackage[bookmarks,bookmarksopen,bookmarksdepth=2]{hyperref}
\usepackage{etex}

\usepackage{microtype}

\newcommand\xcall[2]{\call{\hbox{{\it #1}}(#2)}}
\newcommand\xret[2]{\ret{\hbox{{\it #1}}(#2)}}
\newcommand\Lenc[1][L]{#1\textrm{enc}}

\newcommand\lto[1][]{\longrightarrow_{#1}}

\newcommand\redL{\longrightarrow_{\textrm{lib}}}
\newcommand\xr[2][]{\xlongrightarrow{#2}_{#1}}
\newcommand\tid{t_{\mathsf{id}}}

\newcommand\tred[2]{\xrightarrow{#2}_{#1}}
\newcommand\pred[2]{\xLongrightarrow{#2}_{}}

\newcommand{\capprox}{\,\raisebox{-.5ex}{$\stackrel{\textstyle\sqsubset}{\scriptstyle{\sim}}$}\,}
\newcommand\rett[1]{\mathop{\langle#1\rangle}}
\newcommand\letin[3][]{{\sf let}^{#1}\ #2\ {\sf in}\ #3}
\newcommand\link[2]{{\sf link}\ #1\ {\sf in}\ #2}
\newcommand\while[2]{{\sf while}\ #1\ #2}
\newcommand\ifthe[3]{{\sf if}\ #1\ {\sf then}\ #2\ {\sf else}\ #3}
\newcommand\Meths{{\sf Meths}}
\newcommand\Refs{{\sf Refs}}

\newcommand\Vars{{\sf Vars}}

\newcommand\Z{\mathbb{Z}}
\renewcommand\k{\mathcal{K}}
\renewcommand\l{\mathcal{L}}
\newcommand\EE{\mathcal{E}}
\newcommand\FF{\mathcal{F}}
\renewcommand\SS{\mathcal{S}}
\newcommand\PP{\mathcal{P}}
\newcommand\RR{\mathcal{R}}
\newcommand\CC{\mathcal{C}}
\renewcommand\AA{\mathcal{A}}
\newcommand\call[2][]{\mathsf{call}_{#1}\,#2}

\newcommand\ret[2][]{\mathsf{ret}_{#1}\,#2}

\newcommand\dom{\mathsf{dom}}
\newcommand\sem[1]{\llbracket #1 \rrbracket}
\newcommand\encsem[1]{\llbracket #1 \rrbracket_\textrm{enc}}

\renewcommand\tint{\mathsf{int}}
\newcommand\unit{\mathsf{unit}}
\newcommand\proj[2]{#1\upharpoonright #2}
\newcommand\comp{\mathop{\text{\bf ;}}}

\newcommand\insp{\mathit{\phi}}

\newcommand\clg[1]{\mathcal{#1}}
\newcommand\genapp{\capprox}
\newcommand\genlin{\sqsubseteq}
\newcommand\encapp{\capprox_{\textrm{\rm enc}}}
\newcommand\enclin{\sqsubseteq_{\textrm{\rm enc}}}
\newcommand\Rho{\mathrm{P}}
\newcommand\R{\mathcal{R}}

\newcommand\hseq[1]{\mathsf{H}_{\textrm{pre}}^{#1}}
\newcommand\natnum{\mathbb{N}}

\newcommand\sat[2]{\triangleleft_{#1 #2}}
\newcommand\hist{\clg{H}}
\newcommand\ctx{{C}}

\newcommand\btype{\Theta}

\usepackage{wrapfig}
\usepackage[flushright]{paralist}
\defaultleftmargin{1.5em}{}{}{}
\usepackage{tikz}
\usetikzlibrary{automata,positioning,arrows}
\tikzset{automaton/.style={node distance=2cm,on grid}}
\tikzset{every state/.style={draw=none,minimum size=5pt,inner sep=1pt}}
\tikzset{transition/.style={->,>=stealth',shorten >=1pt}}
\tikzset{every initial by arrow/.style={transition}}
\tikzset{initial text={}}
\tikzset{
    >=stealth',
    punkt/.style={
           rectangle,
           rounded corners,
           draw=black, thick,
           text width=6.5em,
           minimum height=2em,
           text centered},
    pil/.style={
           ->,
           shorten <=2pt,
           shorten >=2pt,},
    qwe/.style={
           -,
           shorten <=2pt,
           shorten >=2pt,}           
}

\newcommand{\cutout}[1]{}

\newif\ifsubmitversion
\submitversiontrue

\ifsubmitversion
\newcommand\colortext[2]{#2}
\newcommand{\sidenote}[1]{}
\else
\newcommand\colortext[2]{{\color{#1}#2}}
\newcommand{\sidenote}[1]{\marginpar{\parbox{35mm}{\raggedright\small #1}}}
\fi
\newcommand\nt[1]{\colortext{blue}{#1}}
\newcommand\am[1]{\colortext{red}{#1}}
\newcommand\ntnote[1]{\sidenote{\nt{#1}}}
\newcommand\amnote[1]{\sidenote{\am{#1}}}

\newcommand\rarr\rightarrow
\newcommand\abra[1]{\langle #1 \rangle}
\newcommand\boldemph[1]{\textbf{\em #1}}

\let\olddefinition\definition
\renewcommand{\definition}{\olddefinition\normalfont}

\cutout{
\newtheorem{theorem}{Theorem}
\newtheorem{lemma}[theorem]{Lemma}
\newtheorem{corollary}[theorem]{Corollary}

\theoremstyle{definition}
\newtheorem{definition}[theorem]{Definition}

\theoremstyle{definition}
\newtheorem{example}[theorem]{Example}
\newtheorem{remark}[theorem]{Remark}
}

\newcommand\kvdash{\vdash_\mathsf{K}}
\newcommand\bvdash{\vdash_\mathsf{B}}
\newcommand\lvdash{\vdash_\mathsf{L}}

\newcommand\h{\mathtt{h}}
\newcommand\s{\mathtt{s}}

\begin{document}

\title{Higher-Order Linearisability}

\author{
  Andrzej S.\ Murawski\inst{1}
  \and Nikos Tzevelekos\inst{2}
}
\institute{
  $^1$\ University of Warwick\quad
  $^2$\  Queen Mary University of London
}

\maketitle

\begin{abstract}
Linearisability is a central notion for verifying concurrent
libraries: a given library is proven safe if its operational history
can be rearranged into a new sequential one which, in addition,
satisfies a given specification.
Linearisability has been examined for libraries in which method
arguments and method results are of ground type,
including libraries parameterised with such methods.
In this paper we extend linearisability to the general higher-order
setting: methods can be passed as arguments and returned as values.
A library may also depend on abstract methods of any order.
We use this generalised notion to show correctness of several
higher-order example libraries.
\end{abstract}

\section{Introduction}
\cutout{
\nt{Suggested changes to the Intro:
\begin{itemize}
\item Tone down the parameterised aspect and focus more on the high-order one. Remove the notion of ground libraries and mention parameterised linearisability as a special case of HO-linearisability
\item Talk about sequential histories and linearisability as being a rearrangement to a sequential history that is moreover safe. Sequential history now means one that can be broken into length-2 components of the same thread
\item Introduce first example
\end{itemize}}
}%
Computer programs often take advantage of \emph{libraries}, which are collections of routines, often of specialised nature, implemented
to facilitate software development and, among others, code reuse and modularity.
To support the latter, libraries should follow their specifications, which describe the range of expected behaviours the library should conform to for safe and correct deployment.
Adherence to given specifications can be formalised using the classic notion of contextual approximation (refinement),
which scrutinises the behaviour of code in any possible context. Unfortunately, the quantification makes it difficult
to prove contextual approximations directly, which motivates research into sound techniques for establishing it.

In the concurrent setting, a notion that has been particularly influential is that of \emph{linearisability}~\cite{HW90}.
Linearisability requires that, for each history generated by a library, one should be able to find another
history from the specification (its \emph{linearisation}),
which matches the former up to certain rearrangements of events.
In the original formulation by Herlihy and Wang~\cite{HW90}, these permutations were not allowed to disturb the order between library returns
and client calls. 
Moreover, linearisations were required to be \emph{sequential} traces, that is, sequences of method calls immediately followed by their returns.
This notion of linearisability only applies to closed, i.e.\ fully implemented, libraries in which both method arguments and results are of \emph{ground} types. 
The closedness limitation was lifted by
Cerone, Gotsman and Yang~\cite{CGY14}, who extended the techniques to 
\emph{parametric} libraries, whereby methods were divided into available routines (public methods) and unimplemented ones (abstract methods). 
However, both public and abstract methods were still restricted to first-order functions of type $\tint\rarr\tint$. 
In this paper, we make a further step forward and present linearisability 
for general higher-order concurrent libraries, where methods can be of arbitrary higher-order types. 
In doing so, we also propose a corresponding notion of sequential history for higher-order histories.

We examine libraries $L$ that can interact with their environments by means of public and abstract methods:
a library $L$ with abstract methods of types $\Theta=\theta_1,\cdots,\theta_n$ and public methods $\Theta'=\theta_1',\cdots,\theta_{n'}'$ 
is written as $L:\Theta\to \Theta'$. 
We shall work with arbitrary higher-order types generated from the ground types $\unit$ and $\tint$.
Types in $\Theta,\Theta'$ must always be function types, i.e.\ their order is at least $1$.

\begin{wrapfigure}{r}{0.6\textwidth}\vspace{-7.5mm}
\includegraphics[scale=0.35]{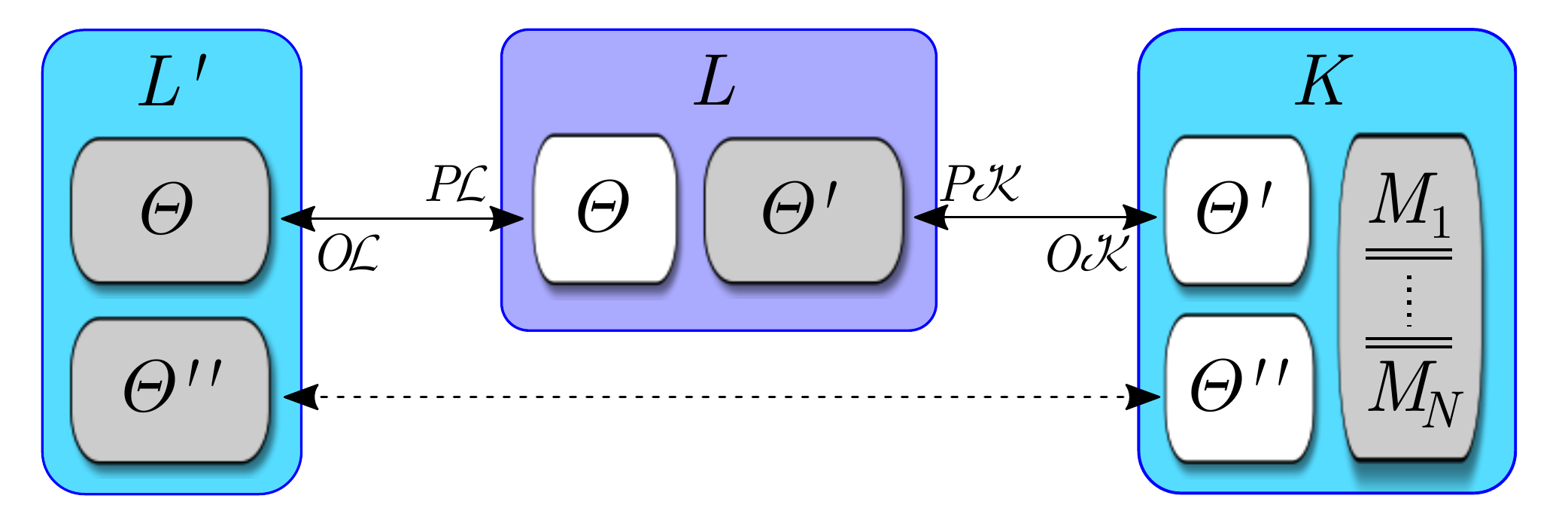}\vspace{-2mm}
\caption{A library $L:\Theta\to\Theta'$ in environment comprising a parameter library $L':\emptyset\to\Theta,\Theta''$ and a client $K$ of the form $\Theta',\Theta''\vdash M_1\|\cdots\|M_N$.}
\label{fig:frame}\vspace{-6mm}
\end{wrapfigure} 
A library $L$ may be used in computations by placing it in a context 
that will keep on calling its public methods (via a client $K$)
as well as providing implementations for the abstract ones (via a parameter library $L'$). 
The setting is depicted in Figure~\ref{fig:frame}.
%
Note that,  as the library $L$ interacts with $K$ and $L'$, they exchange functions between each other. 
Consequently, in addition to $K$ making calls to public methods of $L$ and $L$ making calls to its abstract methods, 
$K$ and $L'$ may also issue calls to functions that were passed to them as arguments during higher-order interactions.
Analogously, $L$ may call functions that were communicated to it via library calls. 

Our framework is operational in flavour and draws upon
concurrent~\cite{Lai01,GM04} and operational game semantics~\cite{JR05b,Lai07,GTz12}.
We shall model library use as a  game between two participants:
\emph{Player (P)}, corresponding to the library $L$, and \emph{Opponent (O)}, representing the environment {$(L',K)$} in which the library was deployed.
Each call will be of the form $\call{m(v)}$ with the corresponding return of the shape $\ret{m(v)}$,
where $v$ {is a value. As we work in a higher-order framework, $v$ may contain functions, which can participate in subsequent calls and returns.
Histories will be sequences of \emph{moves}, which are calls and returns paired with thread identifiers. 
A history is sequential just if every move produced by $O$ is immediately followed by a move by $P$ in the same thread. In other words, the library immediately responds to each call or return delivered by the environment. In contrast to classical linearisability, the move by $O$ and its response by $P$ need not be a call/return pair, as the higher-order setting provides more possibilities (in particular, the $P$ response may well be a call).}
{Accordingly, linearisable higher-order histories can be seen as sequences of atomic segments (linearisation points), starting at environment moves and ending with corresponding library moves.}

In the spirit of~\cite{CGY14}, we are going to consider two scenarios: one in which $K$ and $L'$ share an explicit communication channel (the general case)
as well as a situation in which they can only communicate through the library (the encapsulated case). Further, in the encapsulated case, we will 
handle the case
in which extra closure assumptions can be made about the parameter library (the relational case). The restrictions can deal with a variety of 
assumptions on the use of parameter libraries that may arise in practice.

In each of the three cases, we shall present a candidate definition of linearisability and illustrate it with  tailored examples.
The suitability of each kind of linearisability will be demonstrated by showing that it implies the relevant form of
contextual approximation (refinement). 
We shall also examine compositionality of the proposed concepts. 
One of our examples will discuss the correctness of an implementation of the flat-combining approach~\cite{HIST10,CGY14}, adapted to higher-order types.


\section{Higher-order linearisability}\label{sec:holin}\label{sec:defins}

\newcommand\XX{\mathcal{X}}
\newcommand\YY{\mathcal{Y}}
\newcommand\satsym{\diamond}
\newcommand\cS{\mathcal{S}}
\newcommand\mseta{I}
\newcommand\msetb{J}

As mentioned above, we examine libraries interacting with their context by means of abstract and public methods. 
In particular, we consider higher-order types given by the grammar on the left below,
and let $\Meths$ be a set of \emph{method names}.
\[
\theta ::= \unit\mid \tint\mid \theta\times\theta\mid\theta\to\theta
\qquad
\Meths = \biguplus\nolimits_{\theta,\theta'}\Meths_{\theta,\theta'}
\qquad
v::= ()\mid i\mid m\mid(v,v)
\]
Methods are ranged over my $m$ (and variants), and 
each set $\Meths_{\theta,\theta'}$ contains names for methods of type $\theta\to\theta'$.
Finally, we let $v$ range over computational \emph{values}, which include a unit value, integers, methods, and pairs of values.

The framework of a higher-order library and its environment is 
depicted in Figure~\ref{fig:frame}. 
Given $\Theta,\Theta'\subseteq\Meths$, a library $L$ is said to have type $\Theta\to\Theta'$ if it defines public methods with names (and types) as in $\Theta'$, using abstract methods $\Theta$.
The environment of $L$ consists of a \emph{client} $K$ (which invokes the public methods of $\Theta'$), and a \emph{parameter library} $L'$ (which provides the code for the abstract methods $\Theta$). In general, $K$ and $L'$ may be able to interact via a disjoint set of methods $\Theta''\subseteq\Meths$, to which $L$ has no access.

In the rest of this paper we will be implicitly assuming that  we work with a library $L$ operating in an environment presented in Figure~\ref{fig:frame}. 
The client $K$ will consist of a fixed number $N$ of concurrent threads.
Next we introduce a notion of history tailored to the setting and define how histories can be linearised.
In Section~\ref{sec:syntax} we present the syntax for libraries and clients; and in Section~\ref{sec:semantics} we 
define their semantics in terms of histories (and co-histories).

\subsection{Higher-order histories}

The operational semantics of libraries will be given in terms of \emph{histories}, 
which are sequences of method calls and returns each decorated with a thread identifier and a \emph{polarity index} $XY$, where $X\in\{O,P\}$ and $Y\in\{\l,\k\}$:
\[
(t,\call{ m(v)})_{XY} \qquad (t,\ret{ m(v)})_{XY}
\]
We refer to decorated calls and returns like above as \boldemph{moves}.
Here, $m$ is a method name and $v$ is a value of matching type. The index $XY$ is specifying which of the three entities ($L,L',K$) produces the move, and towards whom it is addressed. 
\begin{itemize}
\item If $X=P$ then the move is issued by $L$. Moreover, if $Y=\l$ then it is addressed to $L'$; 
otherwise, if $Y=\k$ then it is addressed to $K$.
\item If $XY=O\l$ then the move is issued by $L'$, and is addressed to $L$.
\item If $XY=O\k$ then the move is issued by $K$, and is addressed to $L$.
\end{itemize}
We can justify the choice of indices: 
the moves can be seen as defining a 2-player game between the library ($L$), which represents the \emph{Proponent} player in the game, and its enviroment ($L',K$) that represents the \emph{Opponent}. Moves played between $L$ and $L'$ are moreover decorated with $\l$; whereas those between $L$ and $K$ have $\k$. Note that the possible interaction between $L'$ and $K$ is invisible to $L$ and is therefore not accounted for in the game (but we will later see how it can affect it).
We use $O$ to refer to either $O\k$ or $O\l$, and $P$ to refer to either $P\k$ or $P\l$. 

\begin{definition}[Prehistories]
We define \emph{prehistories} as sequences of moves derived by one of the following grammars,
\[\begin{array}{rcl}
\hseq{O} &::=&\ \epsilon\ \mid\ \call{ m(v)}_{OY}\,\, \hseq{P}\,\, \ret{m(v')}_{P Y}\,\, \hseq{O}\\
\hseq{P} &::=&\ \epsilon\ \mid\ \call{ m(v)}_{PY} \,\, \hseq{O}\,\, \ret{m(v')}_{O Y}\,\, \hseq{P}
\end{array}\]
where, in each line, the two occurrences of $Y\in\{\k,\l\}$ and $m\in\Meths$ must each match.
Moreover, if $m\in\Meths_{\theta,\theta'}$, the types of $v,v'$ must match $\theta,\theta'$ respectively.
\end{definition}

The elements of $\hseq{O}$ are patterns of actions starting with an $O$-move, while those in $\hseq{P}$ start with a $P$-move. Note that, in each case, the polarities alternate
and the polarities of calls and matching returns always match the pattern $(XY,X'Y)$ for $X\not=X'$.

Histories will be interleavings of prehistories tagged with thread identifiers (natural numbers) that satisfy a number of technical conditions.
Given $h\in \hseq{O/P}$ and $t\in\natnum$, we write $ t\times h$ for $h$ in which each call or return is decorated with $t$.
We refer to such moves with $(t,\call{m(v)})_{XY}$ or $(t,\ret{m(v)})_{XY}$ respectively. 
If we only want to stress the $X$ or $Y$ membership, we shall drop $Y$ or $X$ respectively. Moreover, when no confusion arises, we may sometimes drop a move's polarity altogether.

\begin{definition}[Histories]\label{def:hist}
Given $\Theta,\Theta'$,
the set of \boldemph{histories} 
over $\Theta\to\Theta'$ is defined by:
\begin{align*}
\hist_{\Theta,\Theta'} &= \bigcup\nolimits_{N>0}\, \,\bigcup\nolimits_{h_1,\cdots,h_N\in \hseq{O}}  (1\times h_1)\mid \cdots\mid (N\times h_N)
\end{align*}
where $(1\times h_1)\mid \cdots\mid (N\times h_N)$ is the set of all interleavings of 
$(1\times h_1), \cdots, (N\times h_N)$, satisfying the following conditions.
\begin{compactenum}
\item For any $s_1 (t,\call{m(v)})_{XY} s_2\in \hist_{\Theta,\Theta'}$:
\begin{itemize}
\item either $m\in\Theta'$ and $XY=O\k$,
or $m\in\Theta$ and  $XY=P\l$,
\item or there is a  move $(t',x')_{X'Y}$ in $s_1$ with $X\neq X'$, such that $x'\!\in\{\call{m'(v)},\ret{m'(v)}\}$ and $v$ contains $m$.
\end{itemize}
\item For any $s_1 (t,x)_{XY} s_2\in \hist_{\Theta,\Theta'}$, where $x\in \{ \call{m(v)},\ret{m(v)}\}$
and $v$ includes some $m'\in\Meths$,
$m'$ must not occur in $s_1$.
\end{compactenum}
\end{definition}
Condition~1 in the definition above requires that any call must refer to  $\Theta$ or $\Theta'$, or be introduced earlier as a higher-order argument or result.
If the method is from $\Theta'$, the call must be tagged with $O\k$ (i.e.\ issued by $K$).
Dual constraints apply to $\Theta$. If a method name does not come from $\Theta$ or $\Theta'$, 
in order for the call $(t,\call{m(v)})_{XY}$ to be valid, $m$ must be introduced in an earlier action with the same tag $Y$ but with the opposite tag $X$.
Moreover, as specified by Condition~2, any action involving a higher-order value (i.e.\ a method name) in its argument or result must label it with a fresh name, one that has not been used earlier.
This is done to enable the history to refer unambiguously to each method name encountered during the interaction.

We shall range over $\hist_{\Theta,\Theta'}$ using  $h,s$. The subscripts $\Theta,\Theta'$ will often be omitted.
\begin{remark}
Histories will be used to define the semantics $\sem{L}$ of libraries (cf.\ Section~\ref{sec:semantics}). In particular, 
for each library $L:\Theta\to\Theta'$, we shall have $\sem{L}\subseteq\hist_{\Theta,\Theta'}$.
\end{remark}

\begin{example}\label{ex:difference}
Let $\Theta=\{m:\tint\to\tint\}$ and $\Theta'=\{m':\tint\to\tint\}$.
Note that single-threaded library histories from $\hist_{\emptyset,\Theta'}$ must have the shape 
$(1,\call{m'(i_1)})_{O\k}$ $(1,\ret{m'(j_1)})_{P\k}$ $\cdots$ $(1,\call{m'(i_k)})_{O\k}$ $(1,\ret{m'(j_k)})_{P\k}$.
In this case, the definition coincides with~\cite{CGY14}. However, in general, our notion of histories
is more liberal. For example, the sequence
$(1,\call{m'(1)})_{O\k}$ $(1,\call{m(2)})_{P\l}$ $(1,\call{m'(3)})_{O\k}$ $(1,\ret{m'(4)})_{P\k}$ $(1,\ret{m(5)})_{O\l}$ $(1,\ret{m'(6)})_{P\k}$
is in $\hist_{\Theta,\Theta'}$, even though it is not allowed by Definition 1 of~\cite{CGY14}.
The sequence represents a scenario in which the public method  $m'$ is called for the second time before the first call is answered.
In our higher-order setting,  this scenario may arise if the parameter library 
communicates with the client
and the communication
includes a function that can issue a call to the public method $m'$.\footnote{\nt{By comparison, in~\cite{CGY14}, each $(1,\call{m(v)})_{P\l}$ must be followed by some $(1,\ret{m(v')})_{O\l}$.}}
\end{example}

Finally, we present a notion of sequential history, which generalises that of~\cite{HW90}.
\begin{definition}
We call a history $h\in\hist_{\Theta,\Theta'}$ 
\boldemph{sequential} if it is of the form
\[
h=(t_1,x_1)_{OY_1}(t_1,x_1')_{PY_1'}\,\cdots\,(t_k,x_k)_{OY_k}(t_k,x_k')_{PY_k'}
\]
for some $t_i,x_i,x_i',Y_i,Y_i'$. 
We let $\hist_{\Theta,\Theta'}^{\sf seq}$ contain all sequential histories of $\hist_{\Theta,\Theta'}$.
\end{definition}

\cutout{
We call a history $h\in\hist_{\Theta,\Theta'}$ 
\boldemph{sequential} if $h=h_1h_2\cdots h_k$ and, for each $i$, 
the component $h_i$ contains moves from a single thread $t_i$ and 
\[
h_i=(t_i,x)_{OY}\,h_i'\,(t_i,x')_{PY}
\]
with $x,x'$ a matching call/return pair.
}

Next we consider our first example.
\begin{example}[{\bf Multiset}]\label{ex:multi}
Consider a concurrent multiset library $L_{\sf mset}$ that uses a private reference for storing the multiset's characteristic function. The implementation is given in 
Figure~\ref{fig:multi}.
This is a simplified version of the optimistic set algorithm of~\cite{HHLMSS05,OHRVYY10}, albeit extended with a higher-order update method.
{The method computes the new value for element $i$ without acquiring a lock on the characteristic function in the hope that
when the lock is acquired the value at $i$ will still be the same and the update can proceed (otherwise another attempt to update the value has to be made).}
The use of a single reference instead of a linked list means that memory safety is no longer problematic, so we focus on linearisability instead. \nt{Note we write 
$|j|$ for the absolute value of $j$.}

\begin{figure}[t]
\begin{minipage}{.5\linewidth}%
\begin{lstlisting}
${\mathsf{public}}$ count, update;

$\mathsf{Lock}$ lock; 
F := $\lambda x$.0;

count = $\lambda$i. (!F)i

update = $\lambda$(i, g). upd_r(i,g,(!F)i)
\end{lstlisting}\end{minipage}
\begin{minipage}{.5\linewidth}%
\begin{lstlisting}[firstnumber=last]
upd_r = $\lambda$(i, g, j).
  let y = |gj| in 
    lock.acquire();
    let f = !F in 
      if (j==fi) then {
        F := $\lambda$x. if (x==i) then y else fx;
        lock.release(); y }
      else {$\,$lock.release(); upd_r(i,g,fi)$\,$}

\end{lstlisting}\end{minipage}
\vspace{-5mm}
\caption{Multiset library $L_{\sf mset}$. 
[\lstinline{count} $:\tint\to\tint$,
\lstinline{update} $:\tint\times(\tint\to\tint)\to\tint$]}\label{fig:multi}\vspace{-1mm}
\end{figure}
Our verification goal will be to prove linearisability of $L_{\sf mset}$  to a specification $A_{\sf mset}\subseteq\hist_{\emptyset,\Theta}^{\sf seq}$,
where $\Theta=\{\textit{count}, \textit{update}\}$ {(the method $\mathit{upd\_r}$ is private)}. 
$A_{\sf mset}$ certifies that $L_{\sf mset}$ correctly implements some integer multiset $\mseta$ whose elements change over time according to the moves in $h$. That is, for each history 
$h\in A_{\sf mset}$ 
there is a multiset $I$ that is empty at the start of $h$ (i.e.\ $\mseta(i)=0$ for all $i$) and:\footnote{%
For a multiset $\mseta$ and a natural number $i$, we write $\mseta(i)$ for the multiplicity of $i$ in $\mseta$; moreover, we set 
$\mseta[i\mapsto j]$ to be $\mseta$ with its multiplicity of $i$ set to $j$. }
\begin{itemize}
\item
If $\mseta$ changes value between two consecutive moves in $h$ then the second move is a $P$-move. In other words, 
the client cannot update the elements of $\mseta$.  
\item
Each call to {\it count} on argument $i$ must be immediately followed by a return with value $\mseta(i)$, and with $\mseta$ remaining unchanged.
\item 
Each call to {\it update} on $(i,m)$ must be followed by a call to $m$ on $i$, with $\mseta$ unchanged.
Then, $m$ must later return with some value $j$. Assuming at that point the multiset will have value $\mseta'$, if $\mseta(i)=\mseta'(i)$ then the next move is a return of the original {\it update} call, with value $j$; otherwise, a new call to $m$ on $i$ is produced, and so on.
\end{itemize}
We formally define the specification next.

Let $\hist_{\emptyset,\Theta}^\circ$ contain 
\emph{extended histories} over $\emptyset\to\Theta$, which are histories where each move is accompanied by a multiset 
(i.e. the sequence consists of elements of the form $(t,x,\mseta)_{XY}$).
For each $s\in\hist_{\emptyset,\Theta}^\circ$, we let $\pi_1(s)$ be the history extracted by projection, i.e. $\pi_1(s)\in\hist_{\emptyset,\Theta}$. 
For each $t$, we let $\proj{s}{t}$ be the subsequence of $s$ of elements with first component $t$. 
Writing $\sqsubseteq$ for the prefix relation, and dropping the $Y$ index from moves ($Y$ is always $\k$ here),
we define \
$A_{\sf mset} = \{\pi_1(s)\mid s\in A_{\sf mset}^\circ\}$ \ where:
\begin{align*}
A_{\sf mset}^\circ &= \{s\in\hist_{\emptyset,\Theta}^\circ\mid
\pi_1(s)\in\hist_{\emptyset,\Theta}^{\sf seq}\land
(\forall s'(\_\,,\mseta)_{P}(\_\,,\msetb)_{O}\sqsubseteq s.\,
\mseta=\msetb)
\land\forall t.\,\proj{s}{t}\in \cS\}
\end{align*}
and, for each $t$, the set of $t$-indexed histories $\cS$ is given by the following grammar.
\begin{align*}
\cS\ &\to\ \epsilon \quad |\quad (t,\xcall{cnt}{i},\mseta)_O\,(t,\xret{cnt}{\mseta(i)},\mseta)_P\,\cS \\
      &\quad\qquad|\quad (t,\xcall{upd}{i,m},\mseta)_O\, {\cal M}_{\mseta,\msetb}^{i,j}\,
(t,\xret{upd}{|j|},\msetb[i\mapsto |j|])_P\,\cS
\\
{\cal M}_{\mseta,\msetb}^{i,j}\ &\to\ (t,\call{m(\mseta(i))},\mseta)_P\,
\cS\, (t,\ret{m(j)},\msetb)_O
&&\hspace{-14mm}(\text{if }\mseta(i)=\msetb(i))
\\
{\cal M}_{\mseta,\msetb}^{i,j}\ &\to\ (t,\call{m(\mseta(i))},\mseta)_P\,
\cS\, (t,\ret{m(j')},\msetb')_O\, {\cal M}_{\msetb',\msetb}^{i,j}
&&\hspace{-14mm}(\text{if }\mseta(i)\not=\msetb'(i))
\end{align*}
By definition, the histories in $A_{\sf mset}$ are all sequential. 
The elements of $A_{\sf mset}^\circ$ carry along the multiset $\mseta$ that is being represented. 
The conditions on $A_{\sf mset}^\circ$  stipulate that $O$ cannot change the value of $\mseta$, while the rest of the conditions above are imposed by the grammar for $\cS$. 
With the notion of linearisability to be introduced next, it will be possible to show that
$\sem{L_{\sf mset}}$  indeed linearises to $A_{\sf mset}$ (see Section~\ref{sec:case1}).
\end{example}

\subsection{General linearisability}\label{sec:genlin}

We begin with a definition of reorderings on histories. 
\begin{definition}
Suppose $X,X'\in\{O,P\}$ and $X\neq X'$.
Let $\sat{X}{X'} \subseteq\hist_{\Theta,\Theta'}\times\hist_{\Theta,\Theta'}$ be 
the smallest binary relation over $\hist_{\Theta,\Theta'}$ satisfying 
\[\begin{array}{rcl}
s_1 (t',x')_{X'} (t,x) s_2 &\sat{X}{X'}&  s_1 (t,x) (t',x')_{X'} s_2\\
s_1 (t',x') (t,x)_X s_2& \sat{X}{X'} & s_1 (t,x)_X (t',x') s_2 \\
\end{array}\]
where $t\neq t'$.
\end{definition}
Intuitively, two histories $h_1,h_2$ are related by $\sat{X}{X'}$ if the latter can be obtained from the former by
swapping two adjacent moves from different threads in such a way that, after the swap, an $X$-move will occur earlier or 
an $X'$-move will occur later. Note that, because of $X\neq X'$, the relation always applies to pairs of moves of the same polarity.
On the other hand, we cannot have $s_1 (t,x)_X (t',x')_{X'} s_2 \sat{X}{X'} s_1 (t',x')_{X'} (t,x)_X s_2$. 

\begin{definition}[{\bf General Linearisability}]\label{def:genlin}
Given $h_1,h_2\in\hist_{\Theta,\Theta'}$, we say that
$h_1$ \emph{is linearised by} $h_2$, written $h_1\genlin h_2$,  if $h_1 \sat{P}{O}^\ast h_2$. 

Given libraries $L,L': \Theta\rarr\Theta'$
and set of sequential histories $A\subseteq\hist_{\Theta,\Theta'}$
we write
$L\genlin A$,
and say that $L$ \emph{can be linearised to $A$}, if for any $h\in \sem{L}$ there exists $h'\in A$ such that $h\genlin h'$. Moreover, we write $L\genlin L'$ if ${L}\genlin\sem{L'}\cap\hist_{\Theta,\Theta'}^{\sf seq}$.
\end{definition}

\begin{remark}
The definition above follows the classic definition from~\cite{HW90} and allows us to (first) express linearisability in terms of a given library $L$ and a sequential specification $A$. The definition given in~\cite{CGY14}, on the other hand, expresses linearisability as a relation between two libraries. This is catered for at the end of Definition~\ref{def:genlin}.
Explicitly, we have that 
$L\genlin L'$ if for all $h\in\sem{L}$  there is some sequential $h'\in\sem{L'}$ such that $h\genlin h'$.
\end{remark}

\begin{example}\label{ex:histories}
Let $\Theta=\{m:\tint\to\tint\}$ and $\Theta'=\{m':\tint\to\tint\}$.
First, consider $h,h_1,h_2\in\hist_{\emptyset,\Theta'}$ given by:
\begin{align*}
h &= 
(1,\call{m(1)})_{O\k}\,  (1,\ret{m(2)})_{P\k}\,  (2,\call{m(3)})_{O\k}\,  (2,\ret{m(4)})_{P\k}\\
h_1 &= 
(1,\call{m(1)})_{O\k}\,  (2,\call{m(3)})_{O\k}\,  (1,\ret{m(2)})_{P\k}\,  (2,\ret{m(4)})_{P\k}\\
h_2 &= 
(2,\call{m(3)})_{O\k}\,  (1,\call{m(1)})_{O\k}\,  (2,\ret{m(4)})_{P\k}\,  (1,\ret{m(2)})_{P\k}
\end{align*}
The histories are related in the following ways: for any $i=1,2$, we have
$h \sat{O}{P}^\ast h_i$ and  $h_i \sat{P}{O}^\ast h$. Moreover,
$h_1 \sat{X}{X'}^\ast h_2$ and $h_2 \sat{X}{X'}^\ast h_1$ for any $X\neq X'$.
Note that we do \emph{not} have $h\sat{P}{O}^\ast h_i$ or $h_i\sat{O}{P}^\ast h$.

Consider now $h_3,h_4\in\hist_{\Theta,\Theta'}$ given respectively by:
\begin{align*}
h_3  &=
(1,\call{m(1)})_{O\k}\,  (1,\call{m'(2)})_{P\l}\,  (1,\ret{m'(3)})_{O\l}\, (1,\ret{m(4)})_{P\k}\\
&\quad (2,\call{m(5)})_{O\k}\,  (2,\call{m'(6)})_{P\l}\,  (2,\ret{m'(7)})_{O\l}\, (2,\ret{m(8)})_{P\k}\\
h_4 &=
(1,\call{m(1)})_{O\k}\,  (2,\call{m(5)})_{O\k}\,  (1,\call{m'(2)})_{P\l}\, (1,\ret{m'(3)})_{O\l}  \\
&\quad (2,\call{m'(6)})_{P\l}\,  (2,\ret{m'(7)})_{O\l}\,  (2,\ret{m(8)})_{P\k}\, (1,\ret{m(4)})_{P\k}
\end{align*}
Observe that $h_3\sat{O}{P}^\ast h_4$ (and, thus, $h_4\sat{P}{O}^\ast h_3$).
However, we do not have $h_4\sat{O}{P}^\ast h_3$ or $h_3\sat{P}{O}^\ast h_4$ .

Regarding linearisability, we can make the following remarks.
\begin{itemize}
\item Observe that in histories from $\hist_{\emptyset,\Theta'}$, we shall have
the following actions: $\call{m'(i)}_O$ and  $\ret{m'(j)}_P$. Thus, $\sat{P}{O}^\ast$ cannot  swap $(t,\ret{m'(j)})$ with $(t',\call{m'(i)})$,
as in the standard definition of linearisability~\cite{HW90}.
\item When $\hist_{\Theta,\Theta'}$ is considered, the available actions are $\call{m'(i)}_O$, $\ret{m(j)}_O$ and
$\call{m(i)}_P$, $\ret{m'(j)}_P$. Then $\sat{P}{O}^\ast$ coincides with Definition 2 of~\cite{CGY14} for second-order libraries.
\end{itemize}
\end{example}

\subsection{Encapsulated linearisability}

A different notion of linearisability will be applicable in cases where the parameter library $L'$ of Figure~\ref{fig:frame} is encapsulated, that is, 
the client $K$ can have no direct access to it (i.e.\ $\Theta''=\emptyset$).
In particular, we shall impose an extra condition on  histories in order to reflect the more restrictive nature of interaction.
Specifically, in addition to sequentiality in every thread, we shall   disallow switches between $\l$ and $\k$ components by $O$.
\begin{definition}
We call a history $h\in\hist_{\Theta,\Theta'}$ 
\boldemph{encapsulated} if, for each thread $t$,
if $h=s_1\, (t,x)_{PY}\, s_2\, (t,x')_{OY'}\, s_3$ and moves from $t$ are absent from $s_2$ then $Y=Y'$.
Moroever, we set
$\hist^{\sf enc}_{\Theta,\Theta'}=\{h\in\hist_{\Theta,\Theta'}\mid h\text{ encapsulated}\}$
and $\encsem{L}=\sem{L}\cap\hist^{\sf enc}_{\Theta,\Theta'}$ (if $L:\Theta\to\Theta'$).\!\!
\end{definition}

We define the corresponding linearisability notion as follows.

\begin{definition}[{\bf Enc-linearisability}]\label{def:enclin}
Let  $\satsym\subseteq \hist_{\Theta,\Theta'}\times\hist_{\Theta,\Theta'}$ be the smallest binary relation on $\hist_{\Theta,\Theta'}$ such that
\[
s_1 (t,m)_{Y} (t',m')_{Y'} s_2\quad \satsym\quad
s_1 (t',m')_{Y'} (t,m)_{Y} s_2
\]
for any $Y,Y'\in\{\k,\l\}$ such that $Y\neq Y'$ and $t\neq t'$.

Given $h_1,h_2\in\hist_{\Theta,\Theta'}^{\sf enc}$, we say that
$h_1$ is \emph{enc-linearised} by $h_2$, and write $h_1\enclin h_2$, if \ $h_1 (\sat{P}{O}\cup \satsym)^\ast h_2$ \ and $h_2$ is sequential. 

A library $L: \Theta\rarr\Theta'$
can be \emph{enc-linearised to $A$}, 
written $L\enclin A$,
if 
$A\subseteq\hist^{\sf seq}_{\Theta,\Theta'}\cap\hist^{\sf enc}_{\Theta,\Theta'}$
and
for any $h\in \encsem{L}$ there exists $h'\in A$ such that $h\enclin h'$.
Moreover, we write $L\enclin L'$ if ${L}\enclin\encsem{L'}\cap\hist_{\Theta,\Theta'}^{\sf seq}$.
\end{definition}

\begin{remark}
Recall $\Theta,\Theta'$ from Example~\ref{ex:histories}. Note that histories may contain the following actions only:
$\call{m'(i)}_{O\k}$, $\ret{m(i)}_{O\l}$, $\call{m(i)}_{P\l}$, $\ret{m'(i)}_{P\k}$. 
Then $(\sat{P}{O}\cup\satsym)^\ast$ preserves the order between $\call{m(i)}_{P\l}$ and $\ret{m(i)}_{O\l}$
as well as that between $\ret{m'(i)}_{P\k}$ and $\call{m'(i)}_{O\k}$, i.e.\ it coincides with Definition 3 of~\cite{CGY14}.
\end{remark}

\begin{example}[{\bf Parameterised multiset}]\label{ex:mult2}
We revisit the multiset library of Example~\ref{ex:multi} and extend it with an abstract method {\it foo} and a corresponding update method {\it update\_enc} which performs updates using {\it foo} as the value-updating function. 
In contrast to the {\it update} method of $L_{\sf mset}$, the method {\it update\_enc} is not optimistic: it retrieves the lock upon its call, and only releases it before return. In particular, the method calls {\it foo} while it retains the lock. 
We call this library $L_{\sf mult2}:\emptyset\to\Theta'$, with $\Theta'=\{\hbox{\it count, update, update\_enc}\}$.

Observe that, were {\it foo} able to externally call {\it update}, we would reach a deadlock: {\it foo} would be keeping the lock while waiting for the return of a method that requires the lock. On the other hand, if the library is encapsulated then the latter scenario is not plausible. In such a case, $L_{\sf mult2}$ linearises to the specification 
$A_{\sf mult2}$ which is defined as follows (cf.\ Example~\ref{ex:multi}).
Let \
$A_{\sf mset2} = \{\pi_1(s)\mid s\in A_{\sf mset2}^\circ\}$ \ where:
\begin{align*}
A_{\sf mset2}^\circ &= \{s\in\hist_{\emptyset,\Theta'}^\circ\mid
\pi_1(s)\in\hist_{\emptyset,\Theta'}^{\sf seq}\land
(\forall s'(\_\,,\mseta)_{P}(\_\,,\msetb)_{O}\sqsubseteq s.\,
\mseta=\msetb)
\land\forall t.\,\proj{s}{t}\in \cS\}
\end{align*}
and the set $\cS$ is now given by the grammar of Example~\ref{ex:multi} extended with the rule:
\begin{align*}
\cS &\to\ (t,\xcall{upd\_enc}{i},\mseta)_{O\k}\,
(t,\xcall{foo}{i},\mseta)_{P\l}\,(t,\xret{foo}{j},\mseta)_{O\l}\,
(t,\xret{upd\_enc}{|j|},\mseta')_{P\k}\,\cS
\end{align*}
with $\mseta'=\mseta[i\mapsto |j|]$. Linearisability is shown in Section~\ref{sec:case1}.
\end{example}

\begin{figure}[t]
\begin{minipage}{.5\linewidth}%
\begin{lstlisting}
${\mathsf{public}}$ count, update, update_enc;
${\mathsf{abstract}}$ foo;
$\mathsf{Lock}$ lock; 
F := $\lambda x$.0;

count = $\lambda$i. (!F)i

update = $\lambda$(i, g). upd_r(i,g,(!F)i)
\end{lstlisting}\end{minipage}
\begin{minipage}{.5\linewidth}%
\begin{lstlisting}[firstnumber=last]
upd_r = $\lambda$(i, g, j).
...
\end{lstlisting}\vspace{-3.65mm}
\begin{lstlisting}[firstnumber=20]
update_enc = $\lambda$i.
  lock.acquire();
  let y = |foo((!F)i)| in
    let f = !F in
      F := $\lambda$x. if (x==i) then y else fx;
      lock.release(); y 
\end{lstlisting}\end{minipage}
\vspace{-5mm}
\caption{Parameterised multiset library $L_{\sf mset2}$.  
[\lstinline{foo,update_enc}$:\tint\to\tint$,
lines 10-16 as in Fig.~\ref{fig:multi}]}\label{fig:multi_enc}\label{fig:mult2}\vspace{-1mm}
\end{figure}

\subsection{Relational linearisability}

\newcommand\relation{\clg{R}}
\newcommand\rellin[1][\relation]{\sqsubseteq_{#1}}
\newcommand\relapp{\capprox_\relation}

Finally, we consider a special subcase of the encapsulated case, in which parameter libraries must satisfy additional constraints, specified
through relational closure. This notion is desirable in cases when unconditional encapsulated linearisability does not hold, yet one may want 
to show that it would hold conditionally on the parameter library. 
The condition on the parameter library $L':\emptyset\to\Theta$ is termed as a relation $\relation\subseteq \hist_{\emptyset,\Theta}\times\hist_{\emptyset,\Theta}$. 

\begin{definition}\label{def:projcomp}
Given a history $h\in\hist_{\Theta,\Theta'}$ and $X\in\{\k,\l\}$, we write
$h\restriction X$ for the subsequence of $h$ with moves whose second index is $X$.

Given a sequence of moves $s$, we write $\overline{s}$ for the sequence obtained from $s$ by simply swapping the $O/P$ indexes in its moves (e.g.\ $\overline{x_{O\l}\,y_{P\k}\,z_{P\k}}=x_{P\l}\,y_{O\k}\,z_{O\k}$).
\end{definition}
\begin{definition}[{\bf Relational linearisability}]\label{def:rellin}
Let $\clg{R}\subseteq \hist_{\emptyset,\Theta}\times\hist_{\emptyset,\Theta}$ be a set closed under permutations of names in $\Meths\setminus\Theta$. 
Given $h_1,h_2\in\hist^{\sf enc}_{\Theta,\Theta'}$,
we say that $h_1$ is \emph{$\cal R$-linearised} by $h_2$, written $h_1\rellin h_2$, if $h_2$ is sequential
and $(h_1\,{\restriction} \k)\, \sat{P}{O}^\ast (h_2\,{\restriction} \k)$
and  $(\overline{h_1}\restriction \l)\, \relation\, (\overline{h_2}\restriction \l)$.

Given $L:\Theta\to\Theta'$ and $A\subseteq \hist_{\emptyset,\Theta}^{\sf seq}\cap\hist_{\emptyset,\Theta}^{\sf enc}$, we write $L\rellin A$ if for any $h\in \encsem{L}$ there exists $h'\in A$ such that $h\rellin h'$.
Moreover, $L\rellin L'$ if ${L}\genlin\encsem{L'}\cap\hist_{\Theta,\Theta'}^{\sf seq}$.
%
\end{definition}

Note the above permutation-closure requirement: if $h\,\relation\, h'$ and $\pi$ is a (type-preserving) permutation on $\Meths\setminus\Theta$ 
then $\pi(h)\,\relation\,\pi(h')$.
The requirement adheres to the fact that, apart from the method names from a library interface, the other method names in its semantics are arbitrary and  can be freely permuted without any observable effect. Thus, $\relation$ in particular should not be distinguishing between such names.

Our third example extends the flat-combining case study from~\cite{CGY14} by lifting it to higher-order types.

\begin{example}\label{ex:FC}
Flat combining~\cite{HIST10} is a synchronisation paradigm that advocates the use of single thread holding a global lock
to process requests of all other threads.
To facilitate this, threads share an array to which they write the details of their
requests and wait either until they acquire a lock or their request has been processed by another thread.
Once a thread acquires a lock, it executes all requests stored in the array and the outcomes  are written to the shared array
for access by the requesting threads.

The authors of \cite{CGY14} analysed a parameterised library that reacts to concurrent requests by calling the corresponding
abstract methods subject to mutual exclusion. In Figure~\ref{fig:fcbasic}, we present
the code adapted to arbitrary higher-order types, which is possible thanks to the presence of higher-order references in our framework.
We assume:
\begin{align*}
\Theta &= \{ m_i\in \Meths_{\theta_i,\theta_i'} \,|\, 1\le i\le k\} & 
\Theta' &=\{ m_i'\in \Meths_{\theta_i,\theta_i'} \,|\, 1\le i\le k\}
\end{align*}
Thus the setup of~\cite{CGY14} corresponds to  $\theta_i=\theta_i'=\tint$.
The library $L_\textrm{fc}$ is built following the flat combining approach and,
on acquisition of a lock, the winning thread acts as a combiner of all registered requests. Note that the requests will be attended to one after another (thus guaranteeing mutual exclusion)
and only one lock acquisition will suffice to process one array of requests.

In Section~\ref{sec:case2}, we shall show that $L_\textrm{fc}$ can be $\relation$-linearised to the specification given by the sequential histories of the library $L_\textrm{spec}$
that implements $m_i'$ as follows:
\begin{lstlisting}[firstnumber=6]
$m_i' =\lambda z.$ ( lock.acquire(); let result = $m_i$(z) in lock.release(); result )
\end{lstlisting}
Thus, each abstract call in $L_\textrm{spec}$ is protected by a lock.
\end{example}

\begin{figure}[t]
%
%
\begin{minipage}{.43\linewidth}%
\begin{lstlisting}
${\mathsf{abstract}\,m_i}\,$;  $\;{\mathsf{public}\, m_i'}\,$;  $\dots\,$;
$\mathsf{Lock}$ lock; 
struct {op, parm, wait, retv} 
       requests[N];

$m_i'$ = $\lambda z$. 
  requests[$\tid$].op := i; 
  requests[$\tid$].parm := z; 
  requests[$\tid$].wait := 1;
\end{lstlisting}\end{minipage}
\begin{minipage}{.5\linewidth}
\begin{lstlisting}[firstnumber=last]
  while (requests[$\tid$].wait)
    if (lock.tryacquire()) { 
      for (t=0; t<N; t++) 
        if (requests[t].wait) {
          let j = requests[t].op in
            requests[t].retv := $m_j$ (requests[t].parm); 
          requests[t].wait := 0 
        }; lock.release() };
  requests[$\tid$].retv;
\end{lstlisting}\end{minipage}\vspace{-2mm}
\caption{Flat combination library $L_{\rm fc}$.}\label{fig:fcbasic}\vspace{-3mm}
\end{figure}



\newcommand\blk{B} 

\section{Library syntax}\label{sec:syntax}

We now look at the concrete syntax of libraries and clients.
Libraries comprise collections of typed methods.
The order\footnote{Type order is defined by $\ord{\unit}=\ord{\tint}=0$, 
$\ord{\theta_1\times\theta_2}=\max(\ord{\theta_1},\ord{\theta_2})$,
$\ord{\theta_1\rarr\theta_2}=\max(\ord{\theta_1}+1,\ord{\theta_2})$.} of their argument and result types is unrestricted and 
will adhere to the grammar: \
$
\theta ::= \unit\mid \tint\mid \theta\to\theta \mid \theta\times\theta
$.

We shall use three disjoint enumerable sets of names, referred to as $\Vars$, $\Meths$ and $\Refs$,
to name respectively variables, methods and references.
$x,f$ (and their decorated variants) will be used to range over $\Vars$; $m$ will range over $\Meths$ and
$r$ over $\Refs$.

Methods and references are implicitly typed, that is, we assume 
\[
\Meths = \biguplus\nolimits_{\theta,\theta'}\Meths_{\theta,\theta'}
\qquad
\Refs =\Refs_\tint \uplus \biguplus\nolimits_{\theta,\theta'}\Refs_{\theta,\theta'}
\]
where $\Meths_{\theta,\theta'}$ contains names for methods of type $\theta\to\theta'$, 
$\Refs_\tint$ contains names of integer references
and 
$\Refs_{\theta,\theta'}$ contains names for references to methods of type $\theta\to\theta'$.
We write $\uplus$ to stress the disjointness of sets in their union.

\newcommand\myquad[1][]{\quad\;\;#1}

\begin{figure}[t]
\begin{align*}
\text{\em Libraries}&& L &::=\ \blk\mid {\sf abstract}\ m;\,L\mid 
{\sf public}\ m;\,L
\qquad\qquad\qquad
\text{\em Clients}\quad\;\; K ::=\ M\,\|\cdots\|\,M
\\
\text{\em Blocks}&&
\blk &::=\ \epsilon \mid m= \lambda x.M;\, \blk \mid r:= \lambda x.M;\,\blk \mid  r:= i;\,\blk\\
\text{\em Terms}&&
M &::=\ ()\mid i\mid\tid\mid x\mid m\mid M\oplus M \mid\ifthe{M}{M}{M}
\mid \abra{M,M}\mid\pi_1\,M\mid\pi_2\,M\\ 
&&&\;\; \quad \mid \lambda x^\theta\!.M \mid xM\mid mM\mid \letin{x=M}{M} \mid r:=M\mid {!r}
\\[-6.5mm] 
\end{align*}\hrule\vspace{-1.5mm}
\begin{gather*}
\infer{\Gamma\vdash ():\unit}{}
\myquad
\infer{\Gamma\vdash i:\tint}{}
\myquad
\infer{\Gamma\vdash \tid:\tint}{}
\myquad
\infer{\Gamma\vdash x:\theta}{\Gamma(x)=\theta}
\myquad
\infer{\Gamma\vdash m:\theta\to\theta'}{m\in\Meths_{\theta,\theta'}}
\myquad
\infer{\Gamma\vdash M_1\oplus M_2:\tint}{\Gamma\vdash M_1,M_2:\tint}
\\
\infer{\Gamma\vdash \ifthe{M}{M_1}{M_0}:\theta}{\Gamma\vdash M:\tint\quad\Gamma\vdash M_0,M_1:\theta}
\myquad
\infer{\Gamma\vdash \abra{M_1, M_2}:\theta_1\times\theta_2}{\Gamma\vdash M_i:\theta_i\;\;(i=1,2)}
\myquad
\infer{\Gamma\vdash \pi_i\,M:\theta_i\;\;(i=1,2)}{\Gamma\vdash \abra{M_1, M_2}:\theta_1\times\theta_2}
\myquad
\infer{\Gamma\vdash {!r}:\tint}{r\in\Refs_{\tint}}
\\
\infer{\Gamma\vdash r:=M:\unit}{r\in\Refs_{\tint}\quad \Gamma\vdash M:\tint}
\myquad
\infer{\Gamma\vdash {!r}:\theta\to\theta'}{r\in\Refs_{\theta,\theta'}}
\myquad
\infer{\Gamma\vdash r:=M:\unit}{r\in\Refs_{\theta,\theta'}\quad\Gamma\vdash M:\theta\to\theta'}
\myquad
\infer{\Gamma\vdash \lambda x^\theta\!.M:\theta\to\theta'}{\Gamma,x:\theta\vdash M:\theta'}
\\
\infer{\Gamma\vdash \letin{x=M}{N}:\theta'}{\Gamma\vdash M:\theta\quad\Gamma,x:\theta\vdash N:\theta'}
\myquad
\infer{\Gamma\vdash x M:\theta'}{\Gamma (x)=\theta\to\theta'\quad\Gamma\vdash M:\theta}
\myquad
\infer{\Gamma\vdash m M:\theta'}{m\in \Meths_{\theta,\theta'}\quad\Gamma\vdash M:\theta}
%
\\[-6.5mm]
\end{gather*}
\hrule\vspace{-1.5mm}
\begin{gather*}
\infer{\bvdash \epsilon:\emptyset}{}\myquad[]
%
\infer{\bvdash m= \lambda x.M;\, \blk:\btype\uplus\{m\}}{m\in\Meths_{\theta,\theta'} \quad x:\theta\vdash M:\theta' \quad \bvdash \blk:\btype}
\myquad[]
\infer{\bvdash r:= \lambda x.M;\, \blk:\btype}{r\in\Refs_{\theta,\theta'} \quad x:\theta\vdash M:\theta' \quad \bvdash\blk:\btype}
\\
\infer{\bvdash r:= i;\, \blk:\btype}{r\in\Refs_{\tint} \quad \bvdash\blk:\btype}
\myquad[\quad]
\infer{\Meths(\blk)\lvdash \blk:\emptyset\to\Theta}{\bvdash \blk:\Theta}
%
\myquad
\infer{\Theta\lvdash {\sf public}\ m; L:\Theta'\to\Theta''}{\Theta\uplus\{m\}\lvdash L:\Theta'\to\Theta''\quad m\in\Theta''}
\\
\infer{\Theta\lvdash {\sf abstract}\ m; L:\Theta'\uplus\{m\}\to\Theta''}{\Theta\uplus\{m\}\lvdash L:\Theta'\to\Theta''\quad m\notin\Theta''}
\myquad[\quad]
\infer{\Theta\kvdash M_1\|\cdots\| M_N:\unit}{\Meths(M_j)\subseteq \Theta\quad\vdash M_j:\unit
\;\,(j=1,\cdots,N)}
\\[-6.5mm]
\end{gather*}
\hrule\vspace{-1.5mm}
\caption{Library syntax and typing rules for terms ($\vdash$), blocks ($\bvdash$),  libraries ($\lvdash$), clients ($\kvdash$).}\label{fig:syntax}
\end{figure}

The syntax for building libraries is defined in Figure~\ref{fig:syntax}.
Thus, each library $L$ begins with a series of method declarations (public or abstract) followed by a block  $B$ consisting of 
method implementations ($m= \lambda x.M$) and reference initialisations ($r:= i$ or $r:= \lambda x.M$).
Our typing rules will ensure that each public method must be implemented within the block, in contrast to abstract methods. On the other hand, a client $K$ consists of a parallel composition of closed terms.

Terms $M$ specify the shape of allowable method bodies. $()$ is the skip command, $i$ ranges over integers, $\tid$ is the current thread identifier
and $\oplus$ represents standard arithmetic operations.
Thanks to higher-order references, we can simulate
divergence by $(!r)()$, where $r\in\Refs_{\unit,\unit}$ is initialised with $\lambda x^\unit.(!r)()$. Similarly, after $r:=\lambda x^\unit. \letin{y=M}{(\ifthe{y}{(N; (!r)())}{()})}$,
$\while{M}{N}$ can be simulated by $(!r)()$. 
We shall also use the standard  derived syntax for sequential composition, i.e. $M;N$ stands for $\letin{x=N}{M}$, where $x$ does not occur in $M$.

\begin{remark}
\nt{In Section~\ref{sec:defins} we used lock-related operations in our example libraries ($\textit{acquire}, \textit{tryacquire}, \textit{release}$),
on the understanding that they can be coded using shared memory. Similarly, the array of Example~\ref{ex:FC} can be constructed using references.}
\end{remark}

For each term $M$, we write $\Meths(M)$ for the set of method names occurring in $M$. We also use the same notation for method names in blocks and libraries.
Terms are typed in environments $\Gamma = \{x_1:\theta_1,\cdots,x_n:\theta_n\}
$
assigning types to their free variables.

Method blocks are typed through judgements $\bvdash \blk: \btype$, where $\btype\subseteq\Meths$.
The judgments collect the names of methods defined in a block as well as making sure that the definitions respect types and are not duplicated.
Also, any initialisation statements will be scrutinised for type compliance.

Finally, we type libraries using statements of the form $\btype \lvdash L: \btype'\rarr\btype''$, where $\btype, \btype',\btype''\subseteq\Meths$ and $\btype'\cap\btype''=\emptyset$.
The judgment $\emptyset\lvdash L:\btype'\to\btype''$ guarantees that 
any method occurring in $L$ is present either in $\btype'$ or $\btype''$, 
that all methods in $\btype'$ have been declared as abstract and not implemented,
while all methods in $\btype''$ have been declared as public and defined.
Thus, $\emptyset\lvdash L:\Theta\to\Theta'$ stands for a library in which 
 $\Theta,\Theta'$ are the abstract and public methods respectively.
In this case, we also write $L:\Theta\to\Theta'$.

 \begin{remark}
 For simplicity, we do not include private methods but the same effect could be achieved by storing them in higher-order references.
 As we explain in the next section, references present in library definitions are de facto private to the library.
 
Note also that, according to our definition, sets of abstract and public methods are disjoint.
However, given $m,m'\in\Refs_{\theta,\theta'}$, one can define a ``public abstract" method with: \
$
{\sf public}\ m;\, {\sf abstract}\ m'; \,m = \lambda x^\theta.m' x
$\,.
\end{remark}




\section{Semantics}\label{closed}\label{sec:semantics}

The semantics of our system will be given in several stages. 
First,
we define an operational semantics for  sequential and concurrent terms that may draw methods from a function repository.
We then adapt it to capture interactions of concurrent clients with libraries that do not feature
abstract methods. The extended notion is then used to define contextual approximation (refinement) for arbitrary libraries.
Subsequently, we introduce a trace semantics of arbitrary libraries that will be used
to define higher-order notions of linearisability and, ultimately, to relate them to contextual refinement.

\subsection{Library-client evaluation}

Libraries, terms and clients are evaluated in environments comprising:
\begin{itemize}
\item A {method environment} $\RR$, called \emph{own-method repository}, which is
a finite partial map on $\Meths$ assigning to each $m$ in its domain, with $m\in\Meths_{\theta,\theta'}$,  a term of the form $\lambda y.M$ (we omit type-superscripts from bound variables for economy). 
\item \nt{A finite partial map $S:\Refs\rightharpoonup (\Z\cup\Meths)$, called \emph{store},
which assigns to each $r$ in its domain an integer (if $r\in\Refs_\tint$) or name from $\Meths_{\theta,\theta'}$
(if $r\in\Refs_{\theta,\theta'}$).}
\ntnote{changed store type from $S:(\Refs_\tint\rightharpoonup\Z)+(\bigcup\nolimits_{\theta,\theta'}\Refs_{\theta,\theta'}\rightharpoonup\Meths_{\theta,\theta'})$, as the latter looked strange mathematically}
\end{itemize}
The evaluation rules are presented in Figure~\ref{fig:evaluation}.

\begin{remark}\label{rem:convention}
We shall assume that reference names used in libraries are library-private, i.e.\ 
sets of reference names used in different libraries are assumed to be disjoint.
Similarly, when libraries are being used by client code, this is done on the understanding
that the references available to that code do not overlap with those used by libraries.
Still, for simplicity, we shall rely on a single set $\Refs$ of references in our operational rules.
\end{remark}

\begin{figure}[t]
\begin{align*}
& (L) \redL (L,\emptyset,S_\mathrm{init})
&
&(r:=i; B,\RR,S)\redL (B,\RR,S[r\mapsto i])
\\
&({\sf abstract}\ m;L,\RR,S)\redL (L,\RR,S)
&
&(m=\lambda x.M;B\,,\RR,S)\redL
(B,\RR\uplus(m\mapsto \lambda x.M),S)
\\ 
&({\sf public}\ m;L,\RR,S)\redL (L,\RR,S)
&
&(r:=\lambda x.M;B,\RR,S)\redL (B,\RR\uplus(m\mapsto\lambda x.M),S_*)
\\[-6.5mm]
\end{align*}
\hrule
\[
\infer[(K_N)]{(M_1\|\cdots\|M_{t-1}\|M\|M_{t+1}\|\cdots\|M_N,\RR,S)\pred{N}{}
(M_1\|\cdots\|M_{t-1}\|M'\|M_{t+1}\|\cdots\|M_N,\RR',S')}
{(M,\RR,S)\tred{t}{} (M',\RR',S')}
\]
\hrule\vspace{-1.5mm}
\begin{align*}
E &::=\ \bullet\mid E \oplus M\mid i\oplus E \mid \ifthe{E}{M}{M}\mid
\pi_j\,E\mid\abra{E,M}\mid\abra{v,E} \mid mE\mid  \letin{x=E}{M} \mid  r:=E \\ 
  v &::=\ ()\mid i\mid m\mid\abra{v,v}\qquad\qquad\qquad 
\qquad\qquad\qquad 
\qquad\qquad
(\text{{\em Evaluation Contexts} and {\em Values}})
\\[-6.5mm]
\end{align*}
\hrule\vspace{-1.5mm}
\begin{align*}
&(E[\tid],\RR,S) \tred{t}{} (E[t],\RR,S)
&
&(E[i_1\oplus i_2],\RR,S) \tred{t}{} (E[i],\RR,S)\quad(i=i_1\oplus i_2)\\
&(E[{!r}],\RR,S) \tred{t}{} (E[S(r)],\RR,S)
&
&(E[\ifthe{0}{M_1}{M_0}],\RR,S) \tred{t}{} (E[M_0],\RR,S)\\
&(E[{r}:=v],\RR,S) \tred{t}{} (E[()],\RR,S_{**})
&
&(E[\ifthe{i_*}{M_1}{M_0}],\RR,S) \tred{t}{} (E[M_1],\RR,S)\\ 
&(E[\lambda x.M],\RR,S) \tred{t}{} (E[m],\RR_*,S)
&
&(E[mv],\RR_{**},S) \tred{t}{} (E[M\{v/x\}],\RR_{**},S)\\ 
&
(E[\pi_j\abra{v_1,\!v_2}],\RR,S) \tred{t}{} (E[v_j],\RR,S)
&
&(E[\letin{x=v}{M}],\RR,S) \tred{t}{} (E[M\{v/x\}],\RR,S)\\[-8.5mm]
\end{align*}
\caption{Evaluation rules for libraries ($\redL$), clients ($\pred{}{}$), and terms ($\tred{t}{}$). Here, 
\hbox{$S_*=S[r\mapsto m]$,\!\!}
$S_{**}=S[r\mapsto v]$,
$\RR_*=\RR\uplus(m\mapsto\lambda x.M)$;
and $i_{*}\neq 0$, $\RR_{**}(m)=\lambda x.M$.}\label{fig:evaluation}
\end{figure}

First we evaluate the library to create an initial repository and store.
This is achieved by the first set of rules in Figure~\ref{fig:evaluation},
where we assume that $S_\mathrm{init}$ is empty.
Note that $m$ in the last rule is any fresh method name of the appropriate type.
Thus, library evaluation produces a tuple $(\epsilon,\RR_0,S_0)$ including a method repository and a store,
which can be used as the initial repository and store for evaluating $M_1 \| \cdots\| M_N$
using the ($K_N$) rule.
We shall call the latter evaluation semantics for clients (denoted by $\pred{}{}$) the \emph{multi-threaded operational semantics}.

Reduction rules rely on evaluation contexts $E$, defined along with values $v$ in the third group in Figure~\ref{fig:evaluation}.
Finally,
rules for closed-term reduction ($\tred{t}{}$) are given in the last group, where $t$ is the current thread index.
Note that the rule for $E[\lambda x.M]$
involves the creation of a new method name $m$, which is used to put the function in the repository $\RR$.

We define termination for clients linked with libraries that have no abstract methods.

\begin{definition}
Let $L:\emptyset\to\Theta'$ and $\Theta'\kvdash M_1\|\cdots\|M_N:\unit$.\footnote{Recall our convention (Remark~\ref{rem:convention}) that $L$ and $M_1,\cdots,M_N$ must access disjoint parts of the store. Terms $M_1, \cdots, M_N$ can share reference names, though.}
We say that \emph{$M_1\|\cdots\| M_N$ terminates with linked library $L$}
if
\begin{align*}
(M_1\|\cdots\| M_N,\RR_0,S_0)\pred{N}{}^*(()\|\cdots\|(),\RR,S)
\end{align*}
for some $\RR,S$, where $(L)\redL^*(\epsilon,\RR_0,S_0)$.
Then we write $\link{L}{(M_1\|\cdots\| M_N)}\,{\Downarrow}$.
\end{definition}
We shall build a notion of contextual approximation of libraries on top of termination:
one library will be said to approximate another if, whenever the former terminates when composed with any parameter library and client, so does the latter.

\nt{There are several ways of composing libraries. Here we will be considering 
the notions of \emph{union} and \emph{sequencing}. The latter is derived from the former with the aid of a third construct, called \emph{hiding}.}
\cutout{
\begin{itemize}
\item One can compose a library $L:\Theta\to\Theta'$ with a parameter library $L':\emptyset\to\Theta,\Theta''$ to obtain a composite library $L'\comp L:\emptyset\to\Theta',\Theta''$.
\item On the other hand, one can merge the code of  two libraries $L_1:\Theta_1\to\Theta_2$ and $L_2:\Theta_1'\to\Theta_2'$ to form a larger library $L_1\uplus L_2:\Theta_1\uplus\Theta_1'\to\Theta_2\uplus\Theta_2'$.
\end{itemize}
While in~\cite{CGY14} these two notions are presented as orthogonal, we consider a unifying notion of \emph{union}. 
In order to express composition, we will also define the notion of \emph{hiding}.}%
%
Below, we denote a library $L$ as $L=D;B$, where $D$ contains all the (public/abstract) method declarations of $L$, and $B$ is its method block.

\begin{definition}[{\bf Library union, hiding, sequencing}]
Let $L_1:\Theta_1\to\Theta_2$ be a library of the form $D_1;B_1$.
\begin{compactitem}
\item
Given library
$L_2:\Theta_1'\to\Theta_2'$ ($=D_2;B_2$)
which accesses disjoint parts of the store from $L_1$ and
such that $\Theta_2\cap\Theta_2'=\emptyset$,
we define the \emph{union} of $L_1$ and $L_2$ as:
\[
L_1\cup L_2: (\Theta_1\cup\Theta_1')\setminus(\Theta_2\cup\Theta_2')\to\Theta_2\cup\Theta_2'\,\ = \
(D_1; B_1)\cup(D_2; B_2) = D_1';D_2';B_1;B_2
\]
where $D_1'$ is $D_1$ with any 
``${\sf abstract}\, m$''
declaration removed for $m\in\Theta_2'$;  dually for $D_2'$.\!\!\!\!\!\!\!
\item
Given some $\Theta=\{m_1,\cdots,m_n\}\subseteq\Theta_2$, we define the \emph{$\Theta$-hiding} of $L_1$ as:
\[
L_1\setminus\Theta\,
:\,\Theta_1\to(\Theta_2\setminus\Theta)\,\
 = \ (D_1;B_1)\setminus\Theta = D_1';B_1'\{{!r_1}/m_1\}\cdots\{{!r_n}/m_n\}
\]
where $D_1'$ is $D_1$ without ${\sf public}\ m$ declarations for $m\in\Theta$ and, for each $i$, $r_i$ is a fresh reference matching the type of $m_i$, and
$B_1'$ is obtained from $B_1$ by replacing each definition $m_i=\lambda x.M$\ by $r_i:=\lambda x.M$\,.\smallskip
\end{compactitem}
The \emph{sequencing} of $L' :\emptyset\to\Theta_1,\Theta'$ with $L_1$ is: \
$
L'\!\comp L_1 
:\emptyset\to\Theta_2,\Theta'\
= \ (L'\cup\ L_1)\setminus\Theta_1
$.
\end{definition}

Thus, the union of two libraries $L_1$ and $L_2$ as above corresponds to merging their code and removing any abstract method declarations for methods defined in the union. On the other hand, 
the hiding of a public method simply renders it private via the use of references. 
These notions are used in defining contextual notions for libraries, that is, notions that require quantification over all possible contexts.

\cutout{
Given libraries $L' :\emptyset\to\Theta,\Theta''$ and 
$L :\Theta\to\Theta'$ as below,
such that $L$ and $L'$ access disjoint parts of the store, 
their composition $L'\!\comp L:\emptyset\to\Theta',\Theta''$ is given by:
\begin{align*}
\blk\comp L &= L;B \\
({\sf public}\ m;L')\comp L &= {\sf public}\ m;(L'\comp L) && \text{ if }m\in\Theta''\\
({\sf public}\ m;L')\comp L &= (L'_{r/m}\comp L_{-m})\{{!r}/m\} && \text{ if }m\in\Theta
\end{align*}
for a fresh reference $r$, where 
$L_{-m}$ is obtained from $L$ by removing the declaration ${\sf abstract}\ m$,
and $L'_{r/m}$ is $L'$ where the definition $m=\lambda x.\cdots$\ is replaced by ${\sf ref}\ r=\lambda x.\cdots$\,.
}

\begin{definition}\label{def:general}
Given $L_1,L_2:\Theta\to\Theta'$, we say that $L_1$ \boldemph{contextually approximates} 
$L_2$, written $L_1\capprox L_2$, if for all $L':\emptyset\to \Theta,\Theta''$ and 
$\Theta',\Theta''\kvdash M_1\|\cdots\|M_N:\unit$, if 
$\link{L'\!\comp L_1}{(M_1\|\cdots\| M_N)}\,{\Downarrow}$ 
then  $\link{L'\!\comp L_2}{(M_1\|\cdots\| M_N)}\,{\Downarrow}$.
In this case, we also say that $L_2$ \boldemph{contextually refines} $L_1$.
%
\end{definition}
Note that, according to this definition, the parameter library $L'$ may communicate directly with the client terms through a common interface $\Theta''$.
We shall refer to this case as the \emph{general} case. Later on, we shall also consider more restrictive testing scenarios in which this  possibility of 
explicit communication is removed. Moreover, from the disjointness conditions in the definitions of sequencing and linking we have that $L_i$, $L'$ and $M_1\|\cdots\|M_N$ access pairwise disjoint parts of the store.

\begin{remark}
Our definition of contextual approximation models communication between the client and the parameter library explicitly through the shared interface $\Theta''$.
This is different in style (but not in substance) from~\cite{CGY14}, where the presence of public abstract methods inside the tested library provides  such a communication channel.
\end{remark}

\subsection{Trace semantics\label{open}}

Building on the earlier operational semantics, we next introduce a trace semantics of libraries, in the spirit of game semantics~\cite{AM98be}. As mentioned in Section~\ref{sec:defins},
the behaviour of a library will be represented as an exchange of moves between two players called $O$ and $P$, representing 
the library ($P$) and the corresponding context ($O$) respectively. The context consists of 
the client of the library as well as the parameter library, with an index on each move specifying which of them is involved in the move ($\k$ or $\l$ respectively).

In contrast to the semantics of the previous section, we will handle scenarios in which methods need not be present
in the repository $\RR$.
Calls to such undefined methods will be represented by labelled transitions\,--\,calls to the context made on behalf of the library ($P$).
The calls can later be responded to with labelled transitions corresponding to returns, made by the context ($O$).
On the other hand, $O$ will be able to invoke methods in $\RR$, which will also be represented through suitable labels.
Because we work in a higher-order setting, calls and returns made by both players may involve methods as arguments or results.
Such methods also become available for future calls: function arguments/results supplied by $P$ are added to the repository and can later be invoked by $O$,
while function arguments/results provided by $O$ can be queried in the same way as abstract methods.

After giving semantics to libraries, we shall also define a semantics for \boldemph{contexts}, i.e.\ clients paired with parameter libraries where the main library is missing. 
More precisely, given a parameter library $L':\emptyset\to\Theta,\Theta'$ and client  $\Theta',\Theta''\kvdash M_1\|\cdots\|M_N:\unit$,
we will define the semantics of  $M_1\|\cdots\|M_N$ when paired with $L'$.
In such a scenario, the roles of $O$ and $P$ will be reversed: $P\k$ will own moves played by the client, 
$P\l$ will be the parameter library, 
while $O$ will correspond to the missing main library ($O\k$ when interacting with the client, and $O\l$ when talking to the parameter library).

Recall the notions of history (Def.~\ref{def:hist}) and history complementation (Def.~\ref{def:projcomp}). We next define a dual notion of history that is used for assigning semantics to contexts. 

\begin{definition}
The set of \boldemph{co-histories} 
over $\Theta\to\Theta'$ is: \
$
\hist^{co}_{\Theta,\Theta'} = \{ \overline{h}\mid h\in \hist_{\Theta,\Theta'}\}
$.
\end{definition}
We shall range over $\hist^{co}_{\Theta,\Theta'}$ again using variables $h,s$. We can show the following.

\begin{lemma}
\begin{itemize}
\item For all $h\in\hist_{\Theta,\Theta'}$ we have $\proj{h}{\l}\in\hist_{\emptyset,\Theta}^{co}$ and  $\proj{h}{\k}\in\hist_{\emptyset,\Theta'}$.
\item For all $h\in\hist_{\Theta,\Theta'}^{co}$ we have $\proj{h}{\l}\in\hist_{\emptyset,\Theta}$ and $\proj{h}{\k}\in\hist_{\emptyset,\Theta'}^{co}$.
\end{itemize}
\cutout{Let us write $\bar h$ for the $O/P$ complement of the history $h$, that is, the history with the same sequence of moves as $h$ but where we have applied the function:
\[
{\sf comp}=(O\l\mapsto P\l,O\k\mapsto P\k, P\l\mapsto O\l, P\k\mapsto O\k)
\]
on move polarities. Then
$h\in\hist_{\Theta,\Theta'}$ iff $\overline h\in\hist_{\Theta,\Theta'}^{co}C$.}
\end{lemma}

The trace semantics will utilise configurations that carry more components than the previous semantics, 
in order to compensate for the fact that we need to keep track of the evaluation history that is currently processed, 
as well as the method names that have been passed between $O$ and $P$. We define two kinds of configurations:
\[
\text{\em O-configurations } (\EE,-,\RR,\PP,\AA,S)
\ \text{ and } \
\text{\em P-configurations }(\EE,M,\RR,\PP,\AA,S)\,,
\]
where 
the component $\EE$ is an \emph{evaluation stack}, that is, a stack of the form $[X_1,X_2,\cdots,X_n]$ with each $X_i$ being either an evaluation context or a method name. 
On the other hand, $\PP=(\PP_\l,\PP_\k)$ with $\PP_\l,\PP_\k\subseteq\dom(\RR)$ being sets of \emph{public} method names, and
$\AA=(\AA_\l,\AA_\k)$ is a pair of sets of \emph{abstract} method names. 
$\PP$ will be used to record all the method names produced by $P$ and passed to $O$: those passed to $O\k$ are stored in $\PP_\k$, while those leaked to $O\l$ are kept in $\PP_\l$. 
Inside $\AA$, the story is the opposite one: $\AA_\k$ ($\AA_\l$) stores the method names produced by $O\k$ (resp.\ $O\l$) and passed to $P$. Consequently, the sets of names stored in 
$\PP_\l,\PP_k,\AA_\l,\AA_k$ will always be disjoint.

Given a pair $\PP$ as above and a set $Z\subseteq\Meths$, we write $\PP\cup_\k Z$ for the pair $(\PP_\l,\PP_\k\cup Z)$. We define $\cup_\l$ in a similar manner, and extend it to pairs $\AA$ as well. 
Moreover, given $\PP$ and $\AA$, we let $\phi(\PP,\AA)$ be the set of \emph{fresh} method names for $\PP,\AA$: $\phi(\PP,\AA)=\Meths\setminus(\PP_\l\cup\PP_\k\cup\AA_\l\cup\AA_\k)$. 


We next give the rules generating the trace semantics. Note that the rules are parameterised by $Y\in\{\k,\l\}$.
This parameter will play a role in our treatment of the encapsulated case, as it allows us to track the labels
related to interactions with the client and the parameter library respectively.
 In all of the rules below,
whenever we write $m(v)$ or $m(v')$, we assume that the type of $v$ matches the argument type of $m$.
\medskip

\begin{asparadesc}
\item[Internal rule]
First we embed the internal rules, introduced earlier: if $(M,\RR,S) \tred{t}{} (M',\RR',S')$ and $\dom(\RR'\setminus\RR)$ consists of names that do not occur in $\EE, \AA$, 
we have:
\[
(\EE,M,\RR,\PP,\AA,S) \tred{t}{} (\EE,M',\RR',\PP,\AA,S')
\tag{\textsc{Int}}
\]
This includes internal method calls (i.e.\ $(E[mv],\RR,S)\tred{t}{}\cdots$ \ with $m\in\dom(\RR)$ ). 
\item[P calls]
In the next family of rules, the library ($P$) calls one of its abstract methods (either the original ones or those acquired via interaction). 
Thus, the rule applies to $m\in\AA_Y$\!.\!\!\!\!\!\!
\[
(\EE,E[mv],\RR,\PP,\AA,S) \tred{t}{\call{m(v')}_{PY}} (m::E::\EE,-,\RR',\PP',\AA,S)
\tag{\textsc{PQy}}
\]
If $v$ does not contain any method names then $v'=v$, $\RR'=\RR$, $\PP'=\PP$.
Otherwise, 
if $v$ contains the (pairwise distinct) names $m_1,\cdots,m_k$,
a fresh name $m_i'\in\phi(\PP,\AA)$ is created for each method name $m_i$ (for future reference to the method),
and replaced for $m_i$ in $v$. That is, \nt{$v'=v\{m_i'/m_i\mid 1\leq i\leq k\}$.}
We must also have that the $m_i'$'s are pairwise distinct (the rule fires for any such $m_i'$'s), and
also  $\RR'= \RR\uplus\{m_i'\mapsto \lambda x.m_ix\mid 1\leq i\leq k\}$ and $\PP'=\PP\cup_Y \{m_1',\cdots,m_k'\}$.
\item[P returns]
Analogously, the library ($P$) may return a result to an earlier call made by the context.
This rule is applicable provided $m\in\PP_Y$. 
\[
 (m::\EE,v,\RR,\PP,\AA,S) \tred{t}{\ret{m(v')}_{PY}} (\EE,-,\RR',\PP',\AA,S)
\tag{\textsc{PAy}}
\]
$v',\RR',\PP'$ are subject to the same constraints as in \textsc{(PQy)}.
\item[O calls]
The remaining rules are dual and represent actions of the context. Here the context calls a public method: either an original one 
or one that has been made public later (by virtue of having been passed to the context by the library).
Here we require $m\in\PP_Y$ and $\RR(m)=\lambda x.M$. 
\[
 (\EE,-,\RR,\PP,\AA,S) \tred{t}{\call{m(v)}_{OY}} (m::\EE,M\{v/x\},\RR,\PP,\AA',S) 
\tag{\textsc{OQy}}
\]
If $v$ contains the names $m_1,\cdots,m_k\in\Meths$, it must be the case that $m_i\in\phi(\PP,\AA)$, for each $i$,
and $\AA'=\AA\cup_Y\{m_1,\cdots,m_k\}$.
\item[O returns]
Finally, we have rules corresponding to values being returned by the context in response to calls made by the library.
In this case we insist on $m\in\AA_Y$. 
\[
 (m::E::\EE,-,\RR,\PP,\AA,S) \tred{t}{\ret{m(v)}_{OY}} (\EE,E[v],\RR,\PP,\AA',S)
\tag{\textsc{OAy}}
\]
As in the previous case, if $m\in\Meths$ is present in $v$ then we need $m\in\phi(\PP,\AA)$
and $\AA'$ is calculated in the same way.
\end{asparadesc}
\medskip

Finally, we extend the trace semantics to a concurrent setting where a fixed number of $N$-many threads run in parallel. Each thread has separate evaluation stack and term components, which we write as $\CC=(\EE,X)$ (where $X$ is a term or ``$-$'').
Thus, a configuration now is of the following form, and we call it an \emph{$N$-configuration}:
\[
(\CC_1\|\cdots\|\CC_N,\RR,\PP,\AA,S)
\]
where, for each $i$, $\CC_i=(\EE_i,X_i)$ and 
$(\EE_i,X_i,\RR,\PP,\AA,S)$ is a sequential configuration.
We shall abuse notation a little and write $(\CC_i,\RR,\PP,\AA,S)$ for 
$(\EE_i,X_i,\RR,\PP,\AA,S)$. Also, below we write $\vec \CC$ for $\CC_1\|\cdots\|\CC_N$ and 
$\vec \CC[i\mapsto \CC']=\CC_1\|\cdots\|\CC_{i-1}\|\CC'\|\CC_{i+1}\|\cdots\|\CC_N$.

The concurrent traces are produced by the following two rules
with the proviso that the names freshly produced internally in \textsc{(PInt)}
are fresh for the whole of $\vec\CC$. 
\begin{align*}
\infer[(\textsc{PInt})]{(\vec\CC,\RR,\PP,\AA,S) \pred{N}{} (\vec\CC[i\mapsto\CC'],\RR',\PP,\AA,S')}{(\CC_i,\RR,\PP,\AA,S) \tred{i}{} (\CC',\RR',\PP,\AA,S')} 
\\[2mm]
\infer[({\textsc{PExt}})]{(\vec \CC,\RR,\PP,\AA,S) \pred{N}{(i,x)_{XY}} (\vec \CC[i\mapsto \CC'],\RR',\PP',\AA',S')}{(\CC_i,\RR,\PP,\AA,S) \tred{i}{x_{XY}} (\CC',\RR',\PP',\AA',S')}
\end{align*}
We can now define the trace semantics of a library $L$. 
We call a configuration component $\CC_i$ \boldemph{final} if it is in one of the following forms:
\begin{align*}
&\CC_i=([],-)\,\text{ or }\,
 \CC_i=([],())
\end{align*}
for $O$- and $P$-configurations respectively. We call $(\vec\CC,\RR,\PP,\AA,S)$ final just if $\vec\CC=\CC_1\|\cdots\|\CC_N$ and each $\CC_i$ is final.

\begin{definition}\label{def:semL}
For each $L:\Theta\to\Theta'$, we define the $N$-trace semantics of $L$ to be:
\[
\sem{L}_N = \{\,s\,\mid\, (\vec\CC_0,\RR_0,(\emptyset,\Theta'),(\Theta,\emptyset),S_0)\,\pred{N}{s}{\!\!}^\ast\rho\ \land\ \rho\text{ final}\,\}
\]
where 
$\vec\CC_0=([],-)\|\cdots\|([],-)$ and
$(L)\redL^*(\epsilon,\RR_0,S_0)$. 
\end{definition}

For economy, in the sequel we might be dropping the
index $N$ from $\sem{L}_N$.

\cutout{The polarity function $\lambda$ assigns to each tagged action $(t,x)$:
\begin{itemize}
\item $KP$ if the action $x$ is produced via (\textsc{PQk}) or (\textsc{PAk}),
\item $LP$ if the action $x$ is produced via (\textsc{PQl}) or (\textsc{PAl}),
\item $KO$ if the action $x$ is produced via (\textsc{OQk}) or (\textsc{OAk}),
\item $LO$ if the action $x$ is produced via (\textsc{OQl}) or (\textsc{OAl}).
\end{itemize}}

We conclude this section by providing a semantics for library contexts.
Recall that in the definition of contextual approximation the library $L:\Theta\to\Theta'$ 
is deployed in a context consisting of a parameter library $L':\emptyset\to\Theta,\Theta''$ and a concurrent composition of client threads $\Theta',\Theta''\vdash M_i:\unit$ ($i=1,\cdots,N$).
This context makes internal use of methods defined in the parameter library, while it recurs to the trace system for the methods in $\Theta'$.
At the same time, the context provides the methods in $\Theta$ in the trace system.
We shall write $\link{L'\!\comp-}{(M_1\|\cdots\| M_N)}$, or simply $\ctx$, to refer to such contexts.
\cutout{
Formally, given $\Theta'=\{m_1,\cdots,m_k\}$, we write $L'[\Theta']$ for the library:
\[
L'[\Theta']=
{\sf abstract}\ m_1;\cdots\,; {\sf abstract}\ m_k;
L'\Theta''
\]}
We give the following semantics to contexts.

\begin{definition}\label{semcon}
Let $\Theta',\Theta''\kvdash M_1\|\cdots\|M_N:\unit$ and $L':\emptyset\to\Theta,\Theta''$. 
We define the semantics of the context formed by $L'$ and $M_1,\cdots,M_N$ to be:
\begin{align*}
\sem{\link{L'\!\comp-}{(M_1\|\cdots\| M_N)}}=
\{\,s\mid (\vec\CC_0,\RR_0,(\Theta,\emptyset),(\emptyset,\Theta'),S_0)\,\pred{N}{s}{\!\!}^\ast\rho\,\land\, \rho\text{ final}\,\}
\end{align*}
where $(L')\redL^*(\epsilon,\RR_0,S_0)$ and
$\vec\CC_0=([],M_1)\|\cdots\|([],M_N)$.
\end{definition}

\begin{lemma}
For any $L:\Theta\to\Theta'$,
$L':\emptyset\to\Theta,\Theta''$  and
 $\Theta',\Theta''\kvdash M_1\|\cdots\|M_N:\unit$ 
we have 
 $\sem{L}_N\subseteq \hist_{\Theta,\Theta'}$ and
$\sem{\link{L'\!\comp-}{(M_1\|\cdots\| M_N)}}\subseteq\hist_{\Theta,\Theta'}^{co}$.
\end{lemma}
\begin{proof}
The relevant sequences of moves are clearly alternating and well-bracketed, when projected on single threads, 
because the LTS is bipartite ($O$- and $P$-configurations) and\linebreak separate evaluation stacks control the evolution in each thread.
Other conditions for histories follow from the partitioning of names
into $\AA_\k,\AA_\l,\PP_\k,\PP_\l$ and 
suitable initialisa- tion: $\Theta,\Theta'$ are inserted into $\AA_\l, \PP_\k$ respectively (for $\sem{L}$)
and into $\PP_\l, \AA_\k$ for $\sem{C}$.
\end{proof}



\section{Examples}\label{sec:case}
We now revisit the example libraries from Section~\ref{sec:holin} and show they each linearise to their respective specification, according to the relevant notion of linearisability 
(general/encapsulated/relational). 

\subsection{Multiset examples}\label{sec:case1}

Recall the multiset library $L_{\sf mset}$ and the specification $A_{\sf mset}$ 
of Example~\ref{ex:multi} and Figure~\ref{fig:multi}.
We show that $L_{\sf mset}\genlin A_{\sf mset}$. More precisely, taking an arbitrary history $h\in\sem{L_{\sf mset}}$ we show that $h$ can be rearranged using $\sat{P}{O}^\ast$ to match an element of $A_{\sf mset}$.
We achieve this by identifying, for each $O$-move $(t,x)_O$ and its following $P$-move $(t,x')_P$ in $h$, a \emph{linearisation point} between them, i.e.\ a place 
in $h$ to which $(t,x)_O$ can moved right and to which $(t,x')_P$ can be moved left so that they become consecutive and, moreover, the resulting history is still produced by $L_{\sf mset}$.
After all these rearrangements, we obtain a sequential history $\hat h$ such that $h\genlin\hat h$ and $\hat h$ is also produced by $L_{\sf mset}$. It then suffices to show that $\hat h\in A_{\sf mset}$.

\begin{lemma}[{\bf Multiset}]\label{lem:mult}
$L_{\sf mset}$ linearises to $A_{\sf mset}$.
\end{lemma}

\begin{proof}
Taking an arbitrary  $h\in\sem{L_{\sf mset}}$, we
demonstrate the linearisation points for pairs of $(O,P)$ moves in $h$, by case analysis on the moves (we drop $\k$ indices from moves as they are ubiquitous). Let us assume that $h$ 
has been generated by a sequence $\rho_1\Rightarrow\rho_2\Rightarrow\cdots\Rightarrow\rho_k$ of atomic transitions and that 
the variable $F$ of $L_{\sf mset}$ is instantiated with the reference $r_F$. Line numbers used below will refer to Figure~\ref{fig:multi}.
\begin{compactenum}
\item 
$h=\cdots\,(t,\xcall{cnt}{i})_{O}\,s\,(t,\xret{cnt}{j})_P\,\cdots$\ . Here the linearisation point is the configuration $\rho_i$ that dereferences $r_F$ as per line~6 in $L_{\sf mset}$ (the $!F$ expression).
\item
$h=\cdots\,(t,\xcall{upd}{i,m})_{O}\,s\,(t,\xcall{m}{j})_P\,\cdots$\ . The linearisation point is the dereferencing of $r_F$ in line~8.
\item
$h=\cdots\,(t,\xret{m}{j'})_{O}\,s\,(t,\xret{upd}{|j'|})_P\,\cdots$\ . The linearisation point is the update to $r_F$ in line~14.
\item
$h=\cdots\,(t,\xret{m}{j'})_{O}\,s\,(t,\xcall{m}{j''})_P\,\cdots$\ . The linearisation point is the dereferencing of $r_F$ in line~12.
\end{compactenum}
{Each of the linearisation points above specifies a $PO$-rearrangement of moves. For instance, 
for $h=s_0\,(t,\xcall{cnt}{i})_{O}\,s\,(t,\xret{cnt}{j})_P\,s'$, let $s=s_1s_2$ where $s_0\, s_1$ is the prefix of $h$ produced by $\rho_1\Rightarrow\rho_2\Rightarrow\cdots\Rightarrow\rho_i$. The rearrangement of $h$ is then the history $\hat h=s_0\,s_1\,(t,\xcall{upd}{i,m})_{O}\,(t,\xcall{m}{j})_P\,s_2\,s'$.} We thus obtain $h\sat{P}{O}^\ast\hat h$.\ntnote{I think this should simply say that we locate the point in the trace corresponding to the specified LP, and then move the preceding O-move right so that it gets just before the LP; and the P-move to the left, just after the LP.}

The selection of linearisation points is such that it guarantees that $\hat h\in\sem{L_\mathsf{mset}}$.
E.g.\ in case~1, the transitions occurring in thread $t$ between the configuration that follows $(t,\xcall{cnt}{i})_{O}$ and $\rho_i$ do not access $r_F$. Hence, we can postpone them and fire them in sequence just $\rho_i$. After $\rho_{i+1}$ and until 
$(t,\xret{cnt}{j})_P$ there is again no access of $r_F$ in $t$ and we can thus bring forward the corresponding transitions just after $\rho_{i+1}$. Similar reasoning applies to case~2. In case~4, we reason similarly but also take into account that rendering the acquisition of the lock by $t$ atomic is sound (i.e.\ the semantics can produce the rearranged history).
Case~3 is similar, but we also use the fact that the access to $r_F$ in lines~11-16 is inside the lock, and hence postponing dereferencing (line~12) to occur in sequence before update (line~14) is sound.

Now, any transition sequence $\alpha$ which produces $\hat h$ (in $\sem{L_{\sf mult}}$) can be used to derive an extended history $h^\circ\in A_{\sf mult}^\circ$, by attaching to each move in $\hat h$ the multiset represented in the configuration that produces the move ($\rho$ produces the move $x$ if $\rho\pred{}{x}\rho'$ in $\alpha$). By projection we then obtain $\hat h\in A_{\sf mult}$. 
\end{proof}

On the other hand, the multiset library of Example~\ref{ex:mult2} and Figure~\ref{fig:mult2} requires encapsulation in order to linearise (cf.\ Example~\ref{ex:mult2}).

\begin{lemma}[{\bf Parameterised multiset}]
$L_{\sf mset2}$ enc-linearises to $A_{\sf mset2}$.
\end{lemma}

\begin{proof}
Again, we identify linearisation points, this time for given $h\in\encsem{L_{\sf mult2}}$. For cases~1-4 as above we reason as in Lemma~\ref{lem:mult}. For {\it upd\_enc} we have the following case.
\[
h=s\,(t,\xcall{upd\_enc}{i})_{O\k}\,s_1(t,\xcall{foo}{j})_{P\l}\,s_2(t,\xret{foo}{j'})_{O\l}\,s_3(t,\xret{upd\_enc}{|j'|})_{P\k}\,\cdots
\]
Here, we need a linearisation point for all four moves above. We pick this to be the point corresponding to the update of the multiset reference $F$ on line~24 (Figure~\ref{fig:mult2}).
We now transform $h$ to $\hat h$ so that the four moves become consecutive, in two steps:
\begin{compactitem}
\item 
Let us write $s_3$ as $s_3=s_3^1 s_3^2$, where the split is at the linearisation point.\ntnote{the only $\l$ moves that we can have are foo ones. In $s_2s_3^1$, the lock is held by $t$. If there were an $\l$ move $x$ in it from $t'\not=t$ then, at the point when $x$ occurs, the lock would be with $t'$, which can't be.}
{Since the lock is constantly held by thread $t$ in $s_2s_3^1$, there can be no calls or returns to {\it foo} in $s_2s_3^1$. Hence, all moves in $s_2s_3^1$ are in component $\k$ and can be transposed with the $\l$-moves above, using $\satsym^*$, to obtain $h'=s\,(t,\xcall{upd\_enc}{i})_{O\k}\,s_1\,s_2\,s_3^1(t,\xcall{foo}{j})_{P\l}$
$(t,\xret{foo}{j'})_{O\l}\,s_3^2(t,\xret{upd\_enc}{|j'|})_{P\k}\,\cdots$\,}
\item {Next, by $PO$-rearrangement we obtain  
$\hat h=s\,s_1\,s_2\,s_3^1(t,\xcall{upd\_enc}{i})_{O\k}(t,\xcall{foo}{j})_{P\l}$
$(t,\xret{foo}{j'})_{O\l}(t,\xret{upd\_enc}{|j'|})_{P\k}\,s_3^2\,\cdots$\,. Thus, $h(\sat{P}{O}\cup\satsym)^*\hat h$.}\ntnote{$\text{\qquad}$\quad Then, the remaining P-move can be moved left by PO rearrangement.}
\end{compactitem}
To prove that $\hat h\in A_{\sf mult2}$ we work as in Lemma~\ref{lem:mult}, i.e.\ via showing that $\hat h\in\encsem{L_{\sf mult2}}$. For the latter, we rely on the fact that the linearisation point was taken at the reference update point (so that any dereferencings from other threads are preserved), and that the dereferences of lines~22 and~23 are within the same lock as the update.
\end{proof}

\subsection{Flat combining}\label{sec:case2}

Recall the libraries $L_\mathsf{fc}$ and $L_\mathsf{spec}$ from Example~\ref{ex:FC} and Figure~\ref{fig:fcbasic}.
We shall investigate the impact of introducing higher-order types to the flat combining algorithm, which will lead to several surprising discoveries.
First of all, let  us observe that, even if $\theta_i=\theta_i'=\tint$, higher-order interactions of both $L_\mathsf{spec}$ and $L_\mathsf{fc}$  with
a client and parameter library according to Definition~\ref{def:general} (general case) may lead to deadlock.
In this case, a client can communicate
with the parameter library (via $\Theta''$) and, for example, he could supply it with a function that calls a public method of the library, say, $m_i'$. 
That function could then be used to implement $m_i$ and, consequently, a client call to $m_i'$ would result in lock acquisition,
then a call to $m_i$, which would trigger another call to $m_i'$ and an attempt to acquire the same lock, while $m_i$ cannot return (cf. Example~\ref{ex:difference}).

Deadlock can also arise in the encapsulated case (Definition~\ref{def:encapsulated}) if the library contains a public method, say $m_i'$, with functional arguments.
Then the client can pass a function that calls $m_i'$ as an argument to $m_i'$ and, if the abstract method $m_i$ subsequently calls the argument, deadlock would follow
in the same way as before.
Correspondingly, in these cases there exist sequences of transitions induced by our transition system that cannot be extended to a history,
e.g. for $\theta_i=\unit\rarr\unit$ this happens after $(1,\call{m'(v)})_{O\k}$ $(1,\call{m(v')}_{P\l})$ $(1,\call{v'()})_{O\l}$ $(1,\call{v()})_{P\k}$ $(1,\call{m'(v'')})_{O\k}$. 
Consequently, the protocol should not really be used in an unrestricted higher-order setting.

However, the phenomenon described above cannot be replicated in the encapsulated case provided the argument types are ground 
 ($\theta_i=\unit,\tint)$.
 \cutout{and
result types are first-order types (i.e. types generated by the grammar $X::= B \mid B\rarr X$, where $B::=\unit,\tint$).}
In this case, without imposing any restrictions on the result types $\theta_i'$,  
we shall show that $L_\mathsf{fc}\rellin {L_\mathsf{spec}}$, where $\R$ stands for thread-name
invariance. Note that this is a proper extension of the result in~\cite{CGY14}, where $\theta_i'$ had to be equal to $\unit$ or $\tint$.
It is really necessary to use the finer notion of $\rellin$ here, as we do not have $L_\mathsf{fc}\enclin {L_\mathsf{spec}}$ (a parameter library that is sensitive 
to thread identifiers may return results that allow one to detect that a request has been handled by a combiner thread which is different from the original one).

\newcommand\igno{}
{
\begin{lemma}[Flat combining]
Let $\Theta=\{ m_i\in \Meths_{\tint,\theta_i'} \,|\, 1\le i\le k\}$ and
$\Theta'=\{ m_i'\in \Meths_{\tint,\theta_i'} \,|\, 1\le i\le k\}$ be such that $\Theta\cap\Theta'=\emptyset$.
Let $\R$ consist of pairs $(h_1,h_2)\in \hist_{\emptyset,\Theta}\times \hist_{\emptyset,\Theta}$  that are identical once thread identifiers
are ignored. Then $L_\mathsf{fc}\rellin L_\mathsf{spec}$.
\end{lemma}
\begin{proof}
Observe that histories from $\encsem{L_{\mathsf{spec}}}$ feature threads built from
segments of one of the three forms (we suppress integer arguments for economy):
\begin{itemize}
\item $(t,\call{m_i' })_{O\k}\ (t,\call{m_i})_{P\l}\ (t,\ret{m_i(v)})_{O\l}\  (t,\ret{m_i'(v')})_{P\k}$, or
\item $(t',\call{w (v)})_{OY}\ (t',\call{w'(v')})_{PY'}$ for $Y\neq Y'$, where $w$ is a name introduced in an earlier move $(t'',  x)_{PY}$ and $w'$ is \nt{a corresponding name introduced by the move preceding $(t'',  x)_{PY}$}, or
\item $(t',\ret{w'(v'')})_{OY'}\  (t',\ret{w(v''')})_{PY}$ such that a segment $(t',\call{w (v)})_{OY}\ (t',\call{w'(v')})_{PY'}$ already occurred earlier.
\end{itemize}
We shall refer to moves in the second and third kind of segments as \emph{inspection moves} and write $\insp$ to refer
to sequences built exclusively from such sequences (we will use $\insp^\k$ and $\insp^\l$ if we want to stress that the moves are
exclusively tagged with $\k$ or $\l$).
\nt{Let us write $\cal X$ for the subset of  $\encsem{L_{\mathsf{spec}}}$ containing plays of the form:}
\begin{align*}
&
(t_0,\call{m_{i_0}'\igno})(t_0,\call{m_{i_0}\igno})(t_0,\ret{m_{i_0}(v_0)})(t_0,\ret{m_{i_0}'(v_0')})\, \insp_0\\
&\;\;
(t_1,\call{m_{i_1}'\igno})(t_1,\call{m_{i_1}\igno})\, \insp_1\, (t_1,\ret{m_{i_1}(v_{1})})(t_1,\ret{m_{i_1}'(v_{1}')})\, \insp_2\\
&\cdots\;
(t_k,\call{m_{i_k}'\igno})(t_k,\call{m_{i_k}\igno})\, \insp_{2k-1}\, (t_k,\ret{m_{i_k}(v_{k})})(t_k,\ret{m_{i_k}'(v_{k})})\, \insp_{2k}
\end{align*}
where each $\insp_i$ contains moves alternating between $O$ and $P$. 

Let $h_1\in\encsem{L_\mathsf{fc}}$. Threads in $h_1$ are built from blocks of the shapes:
\begin{align*}
&(t,\call{m_i'\igno})_{O\k}\ ((t,\call{m_j\igno})_{P\l}\,(t,\ret{m_j(v)}_{O\l}))^\ast\ (t,\ret{m_i'(v')})_{P\k}\\
&\text{or }\ (t',\call{w(v)})_{OY}\, (t',\call{w'(v')})_{PY'}\ \text{ or }\ (t',\ret{w'(v'')})_{OY}\, (t',\ret{w(v''')})_{PY'}.
\end{align*}
In the first case, the $j$'s and $v$'s are meant to represent different values in each iteration. 
 In the second kind of block, $w$ needs to be introduced earlier by some $(t'',x)_{P\nt{Y}}$ move
and $w'$ is then \nt{a} name introduced by the preceding move.
For the third kind, an earlier calling sequence of the second kind must exist in the same thread.
\item Note that, due to locking and sequentiality of loops, $h_1\restriction\l$  takes the shape:
\[\begin{array}{l}
(t_0,\call{m_{i_0}\igno})_P\ (t_0,\ret{m_{i_0}(v_0)})_O\ \insp_{0}^{\l}\
(t_1,\call{m_{i_1}\igno})_P\ \insp_{1}^{\l}\  (t_1,\ret{m_{i_1}(v_{1})})_O\ 
\insp_{2}^{\l}\\
\cdots\
(t_k,\call{m_{i_k}\igno})_P\ \insp_{2k-1}^{\l}\ (t_k,\ret{m_{i_k}(v_{k})})_O\ \insp_{2k}^{\l}
\end{array}\]
Each segment \
$
S_j=(t_j,\call{m_{i_j}\igno})\, \insp_{2j-1}^{\l}\, (t_j,\ret{m_{i_j}(v_{j})})
$ \
in $h_1\restriction\l$ must be preceded (in $h_1$) by a corresponding public call $(t_j',\call{m_j'\igno})_{O\k}$
and followed by a matching return $(t_j',\ret{m_{\nt{j}}'(w_j)})_{P\k}$, where $t_j'$ need not be equal to $t_j$.
Note that there can be no other moves from $t_j'$ separating the two moves in $h_1\restriction\k$.
\\
Let $h_1'$ be obtained  from $h_1$ via the following operations around each segment $S_j$:
\begin{compactitem}
\item
move $(t_j',\call{m_{i_j}'\igno})$ right to precede $(t_j,\call{m_{i_j} \igno})$,
\item
move $(t_j',\ret{m_{i_j}'(v_j')})$ left to follow $(t_j,\ret{m_i(v_j)})$, 
\item
rename $(t_j,\call{m_i \igno})_{P\l} (t_j,\ret{m_{i_j}(v_j)})_{O\l}$ to 
$(t_j',\call{m_{i_j} \igno})_{P\l} (t_j',\ret{m_{i_j}(v_j)})_{O\l}$.
\end{compactitem}
We stress that the changes are to be performed simultaneously  for each segment $S_j$.
The rearrangements result in a library history, because they bring forward the points at which various $v_j, v_j'$ have been
introduced and, thus, inspection moves are legal. 
Then we have $h_1\rellin h_1'$, i.e.\ $(\overline{h_1}\restriction\l)\R(\overline{h_1'}\restriction\l)$ and $(h_1\restriction\k)\sat{P}{O}^\ast (h_1'\restriction\k)$.
The former follows from the multiple renaming of the moves originally tagged with $t_j$ to $t_j'$ and the fact
that their order in $h_1\restriction \l$ is unaffected. The latter holds, because $O$-moves move right and $P$-moves move left 
past other moves in $h_1\restriction\k$ that are not from the same thread.
To conclude, we show that there exists $h_2\in\clg{\nt X}$ such that 
$(\overline{h_1'}\restriction\l) = (\overline{h_2}\restriction\l)$ and $(h_1'\restriction\k)\sat{P}{O}^\ast (h_2\restriction\k)$.
We can obtain $h_2$ by rearranging inspection moves in different threads of $h_1'$ so that they alternate between $O$ and $P$.
Since the inspection moves come in pairs this can simply be done by bringing the paired moves next to each other. 
Because one of them is always from $\k$ and the other from $\l$, this can be achieved
 by rearranging moves in $h_1'\restriction \k$ only: if the $O$-move is from $\k$ it can be moved to the right, otherwise
 the $P$-move from $\k$ can be moved left.
 Then we have $h_2\in{\cal\nt X}\subseteq\encsem{L_\mathsf{spec}}$ with $h_1\rellin h_2$, as required.
\end{proof}}

\section{Correctness}\label{sec:sat}

In this section we prove that the linearisability notions we introduce are correct: linearisability implies contextual approximation.
The approach is based on showing that, in each case, the semantics of contexts is saturated relatively to conditions that are dual to linearisability. Hence, linearising histories does not alter the observable behaviour of a library.
We start by presenting two compositionality theorems on the trace semantics, which will be used for relating library and context semantics.

\subsection{Compositionality}

The semantics we defined is {compositional} in the following ways:
\begin{compactitem}
\item To compute the semantics of a library $L$ inside a context $\ctx$, it suffices to compose the semantics of $\ctx$ with that of $L$, for a suitable notion of context-library composition ($\sem{\ctx}\oslash\sem{L}$).
\item To compute the semantics of a union library $L_1\cup L_2$, we can compose the semantics of $L_1$ and $L_2$, for a suitable notion of library-library composition ($\sem{L_1}\otimes\sem{L_2}$).
\end{compactitem}
The above
are proven using bisimulation techniques for
connecting syntactic and semantic compositions,
and are presented in Appendices~\ref{sec:comp} and
~\ref{sec:comp2} respectively.
%

The latter correspondence is used in Appendix~\ref{sec:cong} for proving that linearisability is a congruence for library composition.
From the former correspondence we obtain the following result, which we shall use for showing correctness.

\begin{theorem}\label{t:comp}
Let $L:\Theta\to\Theta'$, $L':\emptyset\to\Theta,\Theta''$ and $\Theta',\Theta''\kvdash M_1\|\cdots\|M_N:\unit$,
with $L$, $L'$ and $M_1;\cdots;M_N$ accessing pairwise disjoint parts of the store.
Then: 
\begin{align*}
&\link{L'\comp L}{(M_1\|\cdots\| M_N)}\,{\Downarrow}
\iff
\exists h\in\sem{L}_N.\
\bar h\in\sem{\link{L'\comp-}{(M_1\|\cdots\| M_N)}}
\end{align*}
\end{theorem}

\subsection{General linearisability}\label{sec:gensat}

Recall the general notion of linearisability defined in Section~\ref{sec:genlin}, which is based on move-reorderings inside histories. 

In Def.s~\ref{def:semL} and~\ref{semcon} we have defined the trace semantics of libraries and contexts.
The semantics turns out to be closed under $\sat{O}{P}^\ast$.
\begin{lemma}[Saturation]\label{lem:genclosure}
Let $X=\sem{L}$ (Def.~\ref{def:semL}) or $X=\sem{\link{L'\!\comp-}{(M_1\|\cdots\| M_N)}}$ (Def.~\ref{semcon}).
Then if $h\in X$ and $h\sat{O}{P}^\ast h'$ then $h'\in X$.
\end{lemma}
\begin{proof}
Recall that the same labelled transition system underpins the definition of $X$ in either case.
We make several observations about the single-threaded part of that system.
\begin{compactitem}
\item The store is examined and modified only during $\epsilon$-transitions.
\item The only transition possible after a $P$-move is an $O$-move. In particular, it is never the case that a $P$-move is separated from the following $O$-move
by an $\epsilon$-transition.
\end{compactitem}
Let us now consider the multi-threaded system and $t\neq t'$.
\begin{compactitem}
\item  Suppose
$\rho\pred{}{(t',m')_P} \rho_1\pred{}{\epsilon^\ast} \rho_2 \pred{}{(t,m)}  \rho_3$.
Then the $(t',m')_P$-transition can be delayed inside $t'$ until after $(t,m)$, i.e.\
$\rho\pred{}{\epsilon^\ast}  \rho_1' \pred{}{(t,m)}  \rho_2'\pred{}{(t',m')_P}  \rho_3$
for some $\rho_1',\rho_2'$.
This is possible because the ($(t',m')_P$-labelled) transition does not access or modify the store,
and none  of the $\epsilon$-transitions distinguished above can be in $t'$, 
thanks to our earlier observations about the behaviour of the single-threaded system.
\item Analogously, suppose
$\rho\pred{}{(t',m')}  \rho_1\pred{}{\epsilon^\ast}  \rho_2 \pred{}{(t,m)_O}  \rho_3$.
Then the $(t,m)_O$-transition can be brought forward, 
i.e.
$\rho\pred{}{(t,m)_O}  \rho_1' \pred{}{(t',m')}  \rho_2' \pred{}{\epsilon^\ast}  \rho_3$,
because it does not access or modify the store and the preceding $\epsilon$-transitions 
cannot be from $t$.\qedhere
\end{compactitem}
\end{proof}

This, along with the fact that
$
h_1 \sat{X}{X'} h_2 \iff h_2\sat{X'}{X} h_1 \iff \overline{h_1}\sat{X'}{X}\overline{h_2}.
$
lead us to the notion of linearisability defined in Def.~\ref{def:genlin}.

\begin{theorem}\label{thm:genlin}
$L_1 \genlin L_2$ implies $L_1\genapp L_2$.
\end{theorem}
\begin{proof}
Consider $\ctx$ such that $\ctx[L_1]\Downarrow$. We need to show $\ctx[L_2]\Downarrow$.
Because $\ctx[L_1]\Downarrow$, Theorem~\ref{t:comp} implies that there exists 
$h_1\in\sem{L_1}$ such that  $\overline{h_1}\in\sem{\ctx}$.
Because $L_1\genlin L_2$, there exists $h_2\in\sem{L_2}$ with $h_1\sat{P}{O}^\ast h_2$.
Note that $\overline{h_1} \sat{O}{P}^\ast \overline{h_2}$.
By Lem.~\ref{lem:genclosure}, $\overline{h_2}\in\sem{\ctx}$.
Because $h_2\in \sem{L_2}$ and $\overline{h_2}\in\sem{\ctx}$, using Theorem~\ref{t:comp} we can conclude $\ctx[L_2]\Downarrow$.
\end{proof}



\begin{theorem}\label{thm:gencomp}
If $L_1\genlin L_2$ then, for suitably typed $L$ accessing disjoint part of the store than $L_1$ and $L_2$, we have $L\cup L_1 \genlin L\cup L_2$.
\end{theorem}

\cutout{
\begin{theorem}
If $L_1\genlin L_2$ then, for suitably typed $L$, we have $L ;L_1\genlin L; L_2$ and $L_1; L \genlin L_2; L$, as applicable.
\end{theorem}
\begin{proof}
We tackle $L; L_1\genlin L; L_2$. The other case is symmetric. Assume $L_1\genlin L_2$ and suppose $h_1\in \sem{L; L_1}$.
By Corollary~\ref{supercomp}, $h_1 = h'\doublewedge_\sigma h_1'$, where $h'\in\sem{L}$ and $h_1'\in\sem{L_1}$.
Because $L_1\genlin L_2$, there exists $h_2'\in \sem{L_2}$ such that $h_1'\genlin h_2'$, i.e. $h_1'\sat{P}{O}^\ast h_2'$.
Note that some of the rearrangements necessary to transform $h_1'$ into $h_2'$ may concern actions shared by $h_1'$ and $h'$;
their polarity will then be different in $h'$. Let $h''$ be obtained by applying such rearrangements to $h'$. Note that $h'\sat{O}{P}^\ast h''$ (because of the difference in polarities)\amnote{gap?}
and that $h''$ and $h_2'$ will be $\sigma$-compatible.
Since $h'\in\sem{L}$,  Lem.~\ref{lem:genclosure} implies $h''\in\sem{L}$. 
Take $h_2$ to be  $h''\doublewedge_\sigma  h_2'$. 
We then have $h_2\in\sem{L ;L_2}$. Moreover, $h_1\genlin h_2$ thanks to $h_1'\genlin h_2'$ and the fact that the rearrangements 
made to obtain $h''$ from $h'$ have been hidden in $h_2$.
Hence, $h_2\in \sem{L ;L_2}$ and $h_1\genlin h_2$. Thus,  $L ;L_1 \genlin L ;L_2$.
\end{proof}
}

\subsection{Encapsulated linearisability} 

In this case libraries are being tested by clients that do not communicate with the parameter library explicitly.

\begin{definition}[{\bf Encapsulated $\capprox$}]\label{def:encapsulated}
Given libraries $L_1,L_2:\Theta\to\Theta'$,
we write $L_1\capprox_{\textrm{enc}}\ L_2$ when, for all $L':\emptyset\to \Theta$ and $\Theta'\kvdash M_1\|\cdots\|M_N:\unit$,
if $\link{L'\comp L_1}{(M_1\|\cdots\| M_N)}\,{\Downarrow}$ then
$\link{L'\comp L_2}{(M_1\|\cdots\| M_N)}\,{\Downarrow}$.
\end{definition}
We shall call contexts of the above kind \emph{encapsulated}, because the parameter library $L'$ can
no longer communicate directly with the client, unlike in Def.~\ref{def:general}, where they shared methods in $\Theta''$. 
Consequently, $\sem{\link{L'\comp -}{(M_1\|\cdots\| M_N)}}$ can be decomposed via parallel composition into two components, 
whose labels correspond to  $\l$ (parameter library) and $\k$ (client) respectively.
\begin{lemma}[Decomposition]\label{lem:decomp}
Suppose $L':\emptyset\to\Theta$ and $\Theta'\kvdash M_1\|\cdots\|M_N:\unit$, where $\Theta\cap\Theta'=\emptyset$.
Then, setting $\ctx'\equiv\link{\emptyset\comp-}{(M_1\|\cdots\| M_N)}$, we have:
\[
\sem{\link{L'\comp-}{(M_1\|\cdots\| M_N)}}
=\{\, h\in \hist^{co}_{\Theta,\Theta'} \,\,|\,\, (h\restriction\l)\in\sem{L'},\,\,(h\restriction\k)\in\sem{\ctx'}\,\}\,.
\]
\end{lemma}
\cutout{
The Lemma states that  $(h\restriction \k)\in \hist^C$ and $(h\restriction \l)\in\hist^L$. In order for their interleaving to
be included in $\sem{\ctx}$ they must satisfy correctness conditions for membership in $\hist^C$. In particular,
the first action must come from $(h\restriction \l)$ and, inside each thread, the local history must alternate between $O$ an $P$,
as well as satisfying the well-bracketed discipline between calls and returns.
}

\begin{remark}\label{rem:noleak}\nt{%
Consider parameter library
$L':\emptyset\to \{m\}$ and client $\{m'\}\kvdash M:\unit$ with $m,m'\in\Meths_{\unit\to(\unit\to\unit)}$, and
suppose we insert in their context a ``copycat'' library $L$ which implements $m'$ as \ $m'=\lambda x.mx$\,.
Then the following scenario may seem to contradict encapsulation: 
\begin{inparaitem}
\item
$M$ calls $m'()$; 
\item $L$ calls $m()$; 
\item $L'$ returns with $m(m'')$ to $L$; 
\item and finally $L$ copycats this return to $M$.
\end{inparaitem}
However, by definition the latter copycat is done by $L$ returning $m'(m''')$ to $M$, for some \emph{fresh} name $m'''$, and recording internally that $m'''\mapsto \lambda x.m''x$. Hence, no methods of $L'$ can leak to $M$ and encapsulation holds.}
\end{remark}

Because of the above decomposition, the context semantics satisfies a stronger closure property
than that already specified in Lem.~\ref{lem:genclosure}, which in turn leads to the notion of encapsulated linearisability of Def.~\ref{def:enclin}. The latter is defined
in term of the symmetric reordering relation $\satsym$,
which allows for swaps (in either direction) between moves from different threads if they are tagged with $\k$ and $\l$ respectively.

Moreover, we can show the following.

\begin{lemma}[Encapsulated saturation]\label{lem:encclosure}
Consider 
$
X=\sem{\link{L'\comp-}{(M_1\|\cdots\| M_N)}}
$ (Definition~\ref{semcon}). Then:
\begin{compactitem}
\item
If $h\in X$ and $h\, (\sat{O}{P}\cup\satsym)^\ast\, h'$ then $h'\in X$.
\item
Let $s_1\, (t,x)_{OY}\, s_2\, (t,x')_{PY'}\, s_3 \in X$ be such that no move in $s_2$ comes from thread $t$.
Then $Y=Y'$, i.e.\ inside a thread only $O$ can switch between $\k$ and $\l$.
\end{compactitem}
\end{lemma}
Due to Theorem~\ref{t:comp}, the above property of contexts means that, in order to study termination in the encapsulated case, one can safely restrict
attention to library traces satisfying a dual property to the one above, i.e.\ to elements of $\encsem{L}$.
\cutout{
\begin{definition}
Given $\Theta,\Theta'$, let $\hist^{\Lenc}_{\Theta,\Theta'}$  be the subset of $\hist^L_{\Theta,\Theta'}$ consisting of histories $h$ satisfying: for each thread $t$,
if $h=s_1\, (t,x)_{PY}\, s_2\, (t,x')_{OY'}\, s_3$ and moves from $t$ are absent from $s_2$ then $Y=Y'$, i.e. 
$O$ never switches between $\k$ and $\l$ in $t$.
\\
Given a library $L:\Theta\rarr\Theta'$, let $\encsem{L}$ be $\sem{L}\cap\hist^{\Lenc}_{\Theta,\Theta'}$.
\end{definition}
}%
Note that $\encsem{L}$ can be obtained directly from our labelled transition system by restricting its single-threaded part to reflect the switching condition.
Observe that Theorem~\ref{t:comp} will still hold for $\encsem{L}$ (instead of $\sem{L}$), because we have preserved all the histories that are compatible with
context histories. 
We are ready to prove correctness of encapsulated linearisability.

\begin{theorem}\label{thm:enclin}
$L_1 \enclin L_2$ implies $L_1\encapp L_2$.
\end{theorem}
\begin{proof}
Similarly to Theorem~\ref{thm:genlin},
except we invoke
Lemma~\ref{lem:encclosure} instead of Lemma~\ref{lem:genclosure}.
\end{proof}

We next examine the behaviour of $\enclin$ with respect to library composition. In contrast to Section~\ref{sec:gensat}, we need to restrict composition for it to be compatible with encapsulation.

\begin{remark}
The general case of union does not conform with encapsulation in the sense that 
encapsulated testing of $L\cup L_i$ ($i=1,2$) according to Def.~\ref{def:encapsulated} 
may subject $L_i$ to unencapsulated testing.
For example, because method names of $L$ and $L_i$ are allowed to overlap, methods in $L$ may call public methods from $L_i$
as well as implementing abstract methods from $L_i$. This amounts to $L$ playing the role of both $\k$ and $\l$, which in addition can
communicate with each other, as both are inside $L$.

Even if we make $L$ and $L_i$ non-interacting (i.e.\ without common abstract/public methods), if higher-order parameters are still involved, the encapsulated tests of $L\cup L_i$ can violate the encapsulation hypothesis for $L_i$.
For instance, consider the methods \
$
m_2,m_1',m_2'\in\Meths_{\unit,\unit}$ \ and \ $ m_1\in\Meths_{(\unit\to\unit),\unit}
$\,, 
and libraries $L_1,L_2:\{m_1\}\to\{m_2\}$ and $L:\{m_1'\}\to\{m_2'\}$, as well as the unions $L\cup L_i:\{m_1,m_2\}\to\{m_1',m_2'\}$. A possible trace in $\encsem{L\cup L_i}$ is this one:
\begin{align*}
h_i &=(1,\call{m_2()})_{O\k}\ (1,\call{m_1(v)})_{P\l}\ (1,\ret{m_1()})_{O\l}\\
&\quad (1,\ret{m_2()})_{P\k}\ (1,\call{m_2'()})_{O\k}\ (1,\call{m_1'()})_{P\l}\
(1,\call{v()})_{O\l}
\end{align*}
which decomposes as $h_i=h'\doublewedge^\sigma_{\Pi,\emptyset}h_i'$,  with $\Pi=\{m_1,m_2,m_1',m_2'\}$,
$\sigma=2222112$, $h'=(1,\call{m_2'()})_{O\k}\, (1,\call{m_1'()})_{P\l}$
and:
\begin{align*}
h_i' &=(1,\call{m_2()})_{O\k}\ (1,\call{m_1(v)})_{P\l}\ (1,\ret{m_1()})_{O\l}\
(1,\ret{m_2()})_{P\k}\ 
(1,\call{v()})_{O\l}
\end{align*}
We now see that $h_i'\notin\encsem{L_i}$ as in the last move O is changing component from $\k$ to $\l$.
\cutout{
 the library $L$ may facilitate the communication between the $\k$ and $\l$ contexts of $L_i$ (or, more explicitly, $L$ can itself play both context roles for $L_i$).}
\end{remark}
\cutout{
\begin{remark}
The general case of union does not conform with encapsulation, in the sense that taking union of $L_i$ ($i=1,2$) with an arbitrary $L$ may break the encapsulation hypothesis for $L_i$. For instance, consider the methods:
\[
m_2,m_1',m_2'\in\Meths_{\unit,\unit}\quad m_1\in\Meths_{(\unit\to\unit),\unit}
\]
and libraries $L_1,L_2:\{m_1\}\to\{m_2\}$ and $L:\{m_1'\}\to\{m_2'\}$, as well as the unions $L\cup L_i:\{m_1,m_2\}\to\{m_1',m_2'\}$. A possible trace in $\encsem{L\cup L_i}$ is this one:
\begin{align*}
h_i &=(1,\call{m_2()})_{O\k}\ (1,\call{m_1(v)})_{P\l}\ (1,\ret{m_1()})_{O\l}\\
&(1,\ret{m_2()})_{P\k}\ (1,\call{m_2'()})_{O\k}\ (1,\call{m_1'()})_{P\l}\
(1,\call{v()})_{O\l}
\end{align*}
which decomposes as $h_i=h'\doublewedge^\sigma_{\Pi,\emptyset}h_i'$,  with $\Pi=\{m_1,m_2,m_1',m_2'\}$,
$\sigma=2222112$, $h'=(1,\call{m_2'()})_{O\k}\, (1,\call{m_1'()})_{P\l}$
and:
\begin{align*}
h_i' &=(1,\call{m_2()})_{O\k}\ (1,\call{m_1(v)})_{P\l}\ (1,\ret{m_1()})_{O\l}\\
&\quad\;(1,\ret{m_2()})_{P\k}\ 
(1,\call{v()})_{O\l}
\end{align*}
We can observe now that $h_i'\notin\encsem{L_i}$ as in the last move O is changing component from $\k$ to $\l$.
\cutout{
 the library $L$ may facilitate the communication between the $\k$ and $\l$ contexts of $L_i$ (or, more explicitly, $L$ can itself play both context roles for $L_i$).}
\end{remark}
}
We therefore look at compositionality for two specific cases: \nt{encapsulated sequencing (e.g.\ of $L\!:\Theta\to\Theta'$ with  $L'\!:\Theta'\to\Theta''$)} and disjoint union for first-order methods.
%
Given 
$L :\Theta_1\to\Theta_2$ and 
$L' :\Theta_1'\to\Theta_2'$, we define their \emph{disjoint union} \
$
L\uplus L'=L\cup L'\ :\ (\Theta_1\cup\Theta_1')\to(\Theta_2\cup\Theta_2')
$
under the assumption that $(\Theta_1\cup\Theta_2)\cap(\Theta_1'\cup\Theta_2')=\emptyset$.

\begin{theorem}\label{thm:enccomp}
Let $L_1,L_2:\Theta_1\to\Theta_2$ and $L:\Theta_1'\to\Theta_2'$.
If $L_1\enclin L_2$ then: 
\begin{compactitem}
\item 
assuming $\Theta_2'=\Theta_1$, we have
$L\comp L_1\enclin L\comp L_2$ and $L_1\comp L \enclin L_2\comp L$;
\item if $\Theta_1,\Theta_2,\Theta_1',\Theta_2'$ are first-order then $L \uplus L_1\enclin L \uplus L_2$.
\end{compactitem}
\end{theorem}

\subsection{Relational linearisability}

Finally, we examine relational linearisability (Def.~\ref{def:rellin}).
We begin with a suitable notion of relation $\relation$.
We next restrict encapsulated contextual testing to $\relation$-closed contexts.

\begin{definition}
Let $\clg{R}\subseteq \hist_{\emptyset,\Theta}\times\hist_{\emptyset,\Theta}$ be a set closed under permutation of names in $\Meths\setminus\Theta$. 
We say that
$L:\emptyset\to \Theta$ is \emph{$\relation$-closed} if, for any $h,h'$ such that $h\,\relation\, h'$, if $h\in \sem{L}$ then $h'\in \sem{L}$.
\end{definition}

\begin{definition}[{\bf $\R$-closed encapsulated $\capprox$}]
Given $L_1,L_2:\Theta\to \Theta'$,
we write $L_1\relapp L_2$ if, for all $\relation$-closed $L':\emptyset\to \Theta$ and for all $\Theta'\kvdash M_1\|\cdots\|M_N:\unit$,
whenever  $\link{L'\comp L_1}{(M_1\|\cdots\| M_N)}\,{\Downarrow}$ then we also have
$\link{L'\comp L_2}{(M_1\|\cdots\| M_N)}\,{\Downarrow}$.
\end{definition}


\begin{theorem}\label{thm:relcor}
$L_1\rellin L_2$ implies $L_1\relapp L_2$.
\end{theorem}

We conclude by showing that $\rellin$ is also compositional in the sense proposed in~\cite{CGY14}.
Given $\relation, \clg{G}\subseteq \hist\times\hist$, we say that $L$ is \emph{$\relation\choose\clg{G}$-closed}
if, for all $k\in\hist$ and $h'\in\encsem{L}$, $(h'\restriction \k)\,\,\relation\,\, k$ implies that there is $h''\in\encsem{L}$
with $({h''}\restriction \k) =k$ and $(\overline{h'}\restriction \l)\,\,\clg{G}\,\,(\overline{h''}\restriction \l)$. 

\begin{theorem}\label{thm:relcomp}
Let $\relation, \clg{G}\subseteq \hist\times\hist$, $L_1,L_2:\Theta_1\to\Theta_2$ 
and $L:\Theta_1'\to\Theta_2'$,
such that $L_1\rellin L_2$. If $L$ is suitably typed:
\begin{compactitem}
\item 
if $L$ is $\relation\choose\clg{G}$-closed,
we have $L\comp L_1\sqsubseteq_{\clg{G}} L\comp  L_2$;\quad {\bf --} $L_1\comp L\rellin L_2\comp L$;
\item if $\Theta_1,\Theta_2,\Theta_1',\Theta_2'$ are first-order then $L\uplus L_1\rellin[\relation^+]\! L\uplus L_2$, where $\relation^+=\{(s,s')\in\hist_{\emptyset,\Theta_1\cup\Theta_1'}\times\hist_{\emptyset,\Theta_1\cup\Theta_1'}\mid (s{\restriction}\Theta_1)\relation(s'{\restriction}\Theta_1),(s{\restriction}\Theta_1')=(s'{\restriction}\Theta_1')\}$and 
$s\restriction\Theta$ is the largest subsequence of $s$ belonging to $\hist_{\emptyset,\Theta_1}$.
\end{compactitem}
\end{theorem}


\section{Related and future work}

%

Since the work of Herlihy and Wing~\cite{HW90}, linearisability has been consistently used as a correctness criterion for concurrent algorithms on a variety of data structures~\cite{MoirS04},
%
%
and has yielded a variety of proof methods~\cite{DongolD15}.
As mentioned in the Introduction,
the field has focussed on libraries with methods of base-type inputs and outputs, 
with 
Cerone et al.\ 
recently catering for the presence of abstract as well as public methods~\cite{CGY14}. 
An explicit connection between linearisability and refinement was made by Filipovic et al.\ in~\cite{FORY10},
where it was shown that, in base-type settings, linearisability and refinement coincide. Similar results have been proved in~\cite{DerrickSW11,GotsmanY11,LiangF13,CGY14}.
Our contributions herein are notions of linearisability 
that can serve as correctness criteria
for libraries with methods of arbitrary higher-order types. Moreover, we relate them to refinement, thus establishing  the soundness of linearisability, and demonstrate they are well-behaved with respect to library composition.

Verification of concurrent higher-order programs has been extensively studied outside of linearisability; we next mention works most closely related to linearisability reasoning. 
At the conceptual level,~\cite{FORY10} proposed that the verification goal behind linearisability is observational refinement.
In the same vein,~\cite{TuronTABD13} utilised logical relations as a direct method for proving refinement in a higher-order concurrent setting, while~\cite{TuronDB13} introduced a program logic that builds on logical relations.
On the other hand,
proving conformance to a history specification has been addressed in~\cite{SergeyNB15} by supplying history-aware interpretations to off-the-shelf Hoare logics for concurrency. 
Other logic-based approaches for concurrent higher-order libraries, which do not use linearisability or any other notion of logical atomicity, include Higher-Order and Impredicative Concurrent Abstract Predicates~\cite{SvendsenB14,SvendsenBP13}.


One possible avenue for expansion of this work, following the example of~\cite{FORY10}, would be to identify language fragments where higher-order linearisability {coincides} with observational refinement. 
Based on the game semantic results of~\cite{GM04}, such a correspondence may be possible to demonstrate already in the language examined herein.

The higher-order language we examined used memory in the form of references, which were global and moreover followed the standard memory model (sequential consistency). 
Therefore, future research also includes
enriching the setting with dynamically allocated memory and expanding its reach to weak memory models. In the latter direction, our traces could need to be strengthened towards truly-concurrent structures, such as event-structures, following the recent examples of~\cite{Cast16,JJ16}.

%

\bibliographystyle{abbrv}
\bibliography{my}

\clearpage
\appendix


\section{Big-step vs small-step reorderings}

\newcommand\satbig{\triangleleft_{PO}^\textrm{\rm big}}

\cite{CGY14} defines linearisation in the general case using a ``big-step'' relation that applies a single permutation
to the whole sequence. This contrasts with our definition as $\sat{P}{O}^\ast$, in which we combine multiple adjacent swaps.
We show that the two definitions are equivalent.

\begin{definition}[\cite{CGY14}]
Let $h_1, h_2\in\hist_{\Theta,\Theta'}$ of equal length. We write $h_1\satbig h_2$ 
if there is a permutation $\pi: \{1,\cdots, |h_1|\}\to\{1,\cdots,|h_2|\}$ such that, writing $h_i(j)$ for the $j$-th element of $h_i$:
for all $j$,  we have $h_1(j)=h_2(\pi(j))$ and, for all $i<j$:
\[\begin{array}{l}
((\exists t.\,h_1(i)=(t,-)\land\,h_1(j)=(t,-))\\
\lor (\exists t_1,t_2.\, h_1(i)=(t_1,-)_P \land h_1(j)=(t_2,-)_O))\implies h_2(i)<h_2(j)
\end{array}\]
\end{definition}
In other words, $h_2$ is obtained from $h_1$ by permuting moves in such a way that their order in threads is preserved
and whenever a $O$-move occurred after an $P$-move in $h_1$, the same must apply to their permuted copies in $h_2$.

\begin{lemma}
$\satbig = \sat{P}{O}^\ast$.
\end{lemma}
\begin{proof}
It is obvious that $\sat{P}{O}^\ast\subseteq \satbig$, so it suffices to show the converse.

Suppose $h_1\satbig h_2$. Consider the set $X_{h_1, h_2} = \{ h \,|\, h_1\sat{P}{O}^\ast h,\, h\satbig h_2\}$.
Note that $X_{h_1,h_2}$ is not empty, because $h_1\in X_{h_1,h_2}$. 

For two histories $h',h''$, define $\delta(h',h'')$ to be the length of the longest common prefix of $h'$ and $h''$.
Let $N=\max\limits_{h} \{ \delta(h,h_2)\,|\,h\in X_{h_1,h_2}\}$. Note that $N\le |h_1|=|h_2|$.

\begin{itemize}
\item If $N=|h_2|$ then we are done, because $N=|h_2|$ implies $h_2\in X_{h_1,h_2}$ and, thus, $h_1\sat{P}{O}^\ast h_2$.
\item Suppose $N < |h_2|$ and consider $h$ such that $N=\delta(h,h_2)$.
We are going to arrive at a contradiction by exhibiting $h'\in X_{h_1,h_2}$ such that $\delta(h',h_2)>N$.

Because $N=\delta(h,h_2)$ and $N<|h_2|$, we have 
\[\begin{array}{rcl}
h_2&=& a_1\cdots a_N (t,m) u\\ 
h&=& a_1\cdots a_N (t_1,m_1) \cdots (t_k,m_k) (t,m) u',
\end{array}\]
where $t_i\neq t $, because order in threads must be preserved.
Consider 
\[
h'=a_1\cdots a_N (t,m) (t_1,m_1)\cdots (t_k,m_k) u'.
\]
Clearly $\delta(h',h_2) > N$ so, for a contradiction,
it suffices to show that $h'\in X_{h_1,h_2}$. Note that because $h\satbig h_2$, we must also have $h'\satbig h_2$, 
because the new $PO$ dependencies in $h'$ (wrt $h$) caused by moving $(t,m)$ forward are consistent with $h_2$.
Hence, we only need to show that $h\sat{P}{O}^\ast h'$.
Let us distinguish two cases.
\begin{itemize} 
\item If $(t,m)$ is a $P$-move  then, clearly, $h\sat{P}{O}^\ast h'$ ($P$-move moves forward).
\item If $(t,m)$ is an $O$-move then, because $h\satbig h_2$, all of the $(t_i,m_i)$ actions must be $O$-moves (otherwise their position
wrt $(t,m)$ would have to be preserved in $h_2$ and it isn't). Hence, $h\sat{P}{O}^\ast h'$, as required.
\end{itemize}
\end{itemize}
\end{proof}

\section{Proofs from Section~\ref{sec:sat}}

\begin{proof}[Proof of Lemma~\ref{lem:encclosure}]
For the first claim,
closure under $\sat{O}{P}$ (resp.\ $\satsym$) follows from Lemma~\ref{lem:genclosure}. 
(resp.\ Lemma~\ref{lem:decomp}).

Suppose $h=s_1\, (t,x)_{OY}\, s_2\, (t,x')_{PY'}\, s_3$ violates the second claim and $(t,x)$, $(t,x')$ is the earliest such violation in $h$, i.e. no violations occur in $s_1$.
Observe that then $h$ restricted to moves of the form $(t,z)_{XY'}$ would not be alternating, which contradicts the fact that $h \restriction Y'$ is a history (Lemma~\ref{lem:decomp}).
\end{proof}

\begin{proof}[Proof of Theorem~\ref{thm:relcor}]
Consider $\ctx$ such that $\ctx[L_1]\Downarrow$. We need to show $\ctx[L_2]\Downarrow$.
Since $\ctx[L_1]\Downarrow$, by Theorem~\ref{t:comp} there exists $h_1\in\encsem{L_1}$ such that $\overline{h_1}\in\encsem{\ctx}$.
Also, by Lemma~\ref{lem:decomp}, $(\overline{h_1}\restriction \k)\in \encsem{\ctx'}$ and $(\overline{h_1}\restriction \l)\in\encsem{L'}$
for $\ctx', L'$ specified in that lemma.
Because $L_1\rellin L_2$, there exists $h_2\in\encsem{L_2}$ such that
$(h_1\restriction \k)\, \sat{P}{O}^\ast\, (h_2\restriction \k)$ 
and $(\overline{h_1}\restriction \l)\,\relation\, (\overline{h_2}\restriction \l)$.
Note that the former implies $(\overline{h_1}\restriction \k)\, \sat{O}{P}^\ast\, (\overline{h_2}\restriction \k)$.
Because $(\overline{h_1}\restriction \l)\in\encsem{L'}$, $(\overline{h_1}\restriction \l)\,\relation\, (\overline{h_2}\restriction \l)$ and $L'$ is $\relation$-closed,
we have $(\overline{h_2}\restriction \l)\in\encsem{L'}$.
On the other hand, because $(\overline{h_1}\restriction \k)\in \encsem{\ctx'}$ and
$(\overline{h_1}\restriction \k)\, \sat{O}{P}^\ast\, (\overline{h_2}\restriction \k)$
Lemma~\ref{lem:encclosure} implies $(\overline{h_2}\restriction \k)\in\encsem{\ctx'}$.
Consequently, $(\overline{h_2}\restriction \k)\in\encsem{\ctx'}$ and $(\overline{h_2}\restriction \l)\in\encsem{L'}$, so Lemma~\ref{lem:decomp} entails $\overline{h_2}\in\encsem{\ctx}$.
Hence, since $h_2\in \sem{L_2}$ and $\overline{h_2}\in\encsem{\ctx}$,  
we can conclude $\ctx[L_2]\Downarrow$ by Theorem~\ref{t:comp}.
\end{proof}


\section{Trace compositionality}\label{sec:comp}

\newcommand\qweto{\hookrightarrow}

In this section we demonstrate how the semantics of a library inside a context can be drawn by composing the semantics of the library and that of the context. 
%
The result played a crucial role in our arguments about linearisability and contextual refinement in Section~\ref{sec:sat}.


Let us divide (reachable) evaluation stacks into two classes:
{$L$-stacks}, which can be produced in the trace semantics of a library; and {$C$-stacks}, which appear in traces of a context.
\begin{align*}
\EE_L &::=\ []\mid m::E::\EE_L' 
&
\EE_C &::=\ []\mid m::\EE_C' 
\\
\EE_L' &::=\ m::\EE_L
&
\EE_C' &::=\ m::E::\EE_C
\end{align*}
From the trace semantics definition
we have that $N$-configurations in the semantics of a library feature evaluation stacks of the forms $\EE_L$ (in $O$-configurations) and $\EE_L'$ (in $P$-configurations): these we will call \boldemph{$L$-stacks}.
On the other hand,
those produced from a context utilise
\boldemph{$C$-stacks} which are of the forms
$\EE_C$ (in $P$-configurations) and $\EE_C'$ (in $O$-configurations). 

From here on, when we write $\EE$ we will mean an
$L$-stack or a $C$-stack.
Moreover, we will call an $N$-configuration $\rho$ an \boldemph{$L$-configuration} (or a \boldemph{$C$-configuration}), if $\rho=(\vec\CC,\cdots)$ and, for each $i$, $\CC_i=(\EE_i,\cdots)$ with $\EE_i$ an $L$-stack (resp.\ a $C$-stack).

Let $\rho,\rho'$ be $N$-configurations and suppose $\rho=(\vec\CC,\RR,\PP,\AA,S)$ is a $C$-configuration and $\rho'=(\vec\CC',\RR',\PP',\AA',S')$
an $L$-configuration. We say that $\rho$ and $\rho'$ are \boldemph{compatible}, written $\rho\asymp\rho'$,
if $S$ and $S'$ have disjoint domains and, for each $i$:
\begin{compactitem}
\item $\CC_i=(\EE_C,M)$ and $\CC'_i=(\EE_L,-)$, or
$\CC_i=(\EE_C',-)$ and $\CC'_i=(\EE_L',M)$.
\item If the public and abstract names of $\CC_i$ are $(\PP_\l,\PP_\k)$ and $(\AA_\l,\AA_\k)$ respectively, and those of $\CC_i'$ are $(\PP_\l',\PP_\k')$ and $(\AA_\l',\AA_\k')$, then $\PP_\l=\AA'_\l$, $\PP_\k=\AA'_\k$,
$\AA_\l=\PP'_\l$ and $\AA_\k=\PP'_\k$.
\item The private names of $\rho$ (i.e.\ those in $\dom(\RR)\setminus\PP_\l\setminus\PP_\k$) do not appear in $\rho'$, and dually for the private names of $\rho'$.
\item If $\CC_i=(\EE,\cdots)$ and $\CC_i'=(\EE',\cdots)$ then $\EE$ and $\EE'$ are in turn {compatible}, that is:
\begin{itemize}
\item either $\EE=m::E::\EE_1$, $\EE'=m::\EE_1'$ and $\EE_1,\EE_1'$ are compatible,\!\!\!\!\!\!
\item or $\EE=m::\EE_1$, $\EE'=m::E::\EE_1'$ and $\EE_1,\EE_1'$ are compatible,
\end{itemize}
or $\EE=\EE'=[]$.
\end{compactitem}
Note, in particular, that if $\rho\asymp\rho'$ then $\rho$ must be a context configuration, and $\rho'$ a library configuration. 

We next define a trace semantics on compositions of compatible such $N$-configurations. We use the symbol $\oslash$ for configuration composition: 
we call this \boldemph{external composition}, to distinguish it from the composition of $\rho$ and $\rho'$ we can obtain by merging their components, which we will examine later.
\begin{gather*}
\infer[\textsc{Int}_1]{\rho_1\oslash\rho_2\lto\rho_1'\oslash\rho_2}{\rho_1\pred{}{}\rho_1'}
\qquad
\infer[\textsc{Int}_2]{\rho_1\oslash\rho_2\lto\rho_1\oslash\rho_2'}{\rho_2\pred{}{}\rho_2'}
\\
\infer[\textsc{Call}]{\rho_1\oslash\rho_2\lto\rho_1'\oslash\rho_2'}{\rho_1\pred{N}{(t,\call{m(v)})}\rho_1'\quad\rho_2\pred{N}{(t,\call{m(v)})}\rho_2'}
\\
\infer[\textsc{Retn}]{\rho_1\oslash\rho_2\lto\rho_1'\oslash\rho_2'}{\rho_1\pred{N}{(t,\ret{m(v)})}\rho_1'\quad\rho_2\pred{N}{(t,\ret{m(v)})}\rho_2'}
\end{gather*}
The {\sc Int} rules above have side-conditions imposing that the resulting pairs of configurations are still compatible. Concretely, this means that the names created fresh in internal transitions do not match the names already present in the configurations of the other component. 
Note that external composition is not symmetric, due to the context/library distinction we mentioned.

Our next target is to show a correspondence between the above-defined semantic composition and the semantics obtained by (syntactically) merging compatible configurations. 
This will demonstrate that composing the semantics of two components is equivalent to first syntactically composing them and then evaluating the result. In order to obtain this correspondence, we need to make the semantics of syntactically composed configurations more verbose: in external composition methods belong either to the context or the library, and when e.g.\ the client wants to evaluate $mm'$, with $m$ a library method, the call is made explicit and, more importantly, $m'$ is replaced by a fresh method name. On the other hand, when we compose syntactically such a call will be done internally, and without refreshing $m'$.

To counter-balance the above mismatch, we extend the syntax of terms and evaluation contexts, and the operational semantics of closed terms as follows. The semantics will now involve quadruples of the form:
\[
(E[M],\RR_1,\RR_2,S)\text{ written also }(E[M],\vec\RR,S)
\]
where the two repositories correspond to context and library methods respectively, so in particular $\dom(\RR_1)\cap\dom(\RR_2)=\emptyset$. Moreover, inside $E[M]$ we tag method names and lambda-abstractions with indices $1$ and $2$ to record which of the two components (context or library) is enclosing them: the tag $1$ is used for the context, and $2$ for the library. Thus e.g.\ a name $m^1$ signals an occurrence of method $m$ inside the context. Tagged methods are passed around and stored as ordinary methods, but their behaviour changes when they are applied. 
Moreover, we extend (tagged) evaluation contexts by explicitly marking return points of methods:
\[
E\, ::=\ \bullet\mid \cdots \mid\letin{x=E}{M}\mid mE\mid r:=E\mid \rett{m^i}E
\]
In particular, $E[M]$ may not necessarily be a (tagged) term, due to the return annotations.
The new reduction rules are as follows (we omit indices when they are not used in the rules).
\begin{align*}
&(E[i_1\oplus i_2],\vec\RR,S) \tred{t}{}' (E[i],\vec\RR,S')\quad (i=i_1\oplus i_2)\\
&(E[\tid],\vec\RR,S) \tred{t}{}' (E[t],\vec\RR,S')\\
&(E[\pi_j\abra{v_1,v_2}],\vec\RR,S) \tred{t}{}' (E[v_j],\vec\RR,S')\\
&(E[\ifthe{i}{M_0}{M_1}],\vec\RR,S) \tred{t}{}' (E[M_j],\vec\RR,S)\;\; (j=(i>0))\\
&(E[\lambda^i x.M],\vec\RR,S) \tred{t}{}' (E[m^i],\vec\RR\uplus_i(m\mapsto\lambda x.M),S)\\
&(E[m^iv],\vec\RR,S) \tred{t}{}' (E[M\{v/x\}^i],\vec\RR,S)\quad
 \text{if }\RR_i(m)=\lambda x.M \\
&(E[m^iv],\vec\RR,S) \tred{t}{}' (E[\rett{m^i}M\{v'/x\}^{3-i}],\vec\RR',S)\quad
\text{if }\RR_{3-i}(m)=\lambda x.M\text{ with}\\
&\qquad\Meths(v)=\{m_1,\cdots,m_k\},v'=v\{\,m_j'/m_j\mid1\leq j\leq k\},
\vec\RR'=\vec\RR\uplus_i\{m_j'\mapsto\lambda y.m_jy\mid 1\leq j\leq k\} \\
&(E[\rett{m^i}v],\vec\RR,S) \tred{t}{}' (E[v'^i],\vec\RR\uplus_{3-i}\!\{m'_j\mapsto\lambda y.m_jy\},S) \text{ with $m_j,m_j'$ and $v'$ as above}
\\
&(E[\letin{x=v}{M}],\vec\RR,S) \tred{t}{}' (E[M\{v/x\}],\vec\RR,S)\\
&(E[{!r}],\vec\RR,S) \tred{t}{}' (E[S(r)],\vec\RR,S)\\
&(E[{r}:=i],\vec\RR,S) \tred{t}{}' (E,\vec\RR,S[r\mapsto i])\\
&(E[{r}:=m^i],\vec\RR,S) \tred{t}{}' (E,\vec\RR,S[r\mapsto m^i])
\end{align*}
Above we write $M^i$ for the term $M$ with all its methods and lambdas tagged (or re-tagged) with $i$. 
Moreover, we use the convention e.g.\ $\vec\RR\uplus_1(m\mapsto \lambda x.M)=(\RR_1\uplus(m\mapsto \lambda x.M),\RR_2)$.
Note that the repositories need not contain tags as, whenever a method is looked up, we subsequently tag its body explicitly.

Thus, the computationally observable difference of the new semantics is in the rule for reducing $E[m^iv]$ when $m$ is not in the domain of $\RR_i$: this corresponds precisely to the case where e.g.\ a library method is called by the context with another method as argument. 
A similar behaviour is exposed when such a method is returning.
However, this novelty merely adds fresh method names by $\eta$-expansions and does not affect the termination of the reduction. 

Defining parallel reduction $\pred{N}{}'$ in an analogous way to $\pred{N}{}$, we can show the following.
We let a quadruple $(M_1\|\cdots\|M_N,\RR,S)$ 
be \emph{ final} if $M_i=()$ for all $i$, and we write 
$(M_1\|\cdots\|M_N,\RR,S)\Downarrow$ if $(M_1\|\cdots\|M_N,\RR,S)$ can reduce to some final quadruple; these notions are defined for $(M_1\|\cdots\|M_N,\RR_1,\RR_2,S)$ in the same manner.

\begin{lemma}\label{lem:eta}
For any legal $(M_1\|\cdots\|M_N,\RR_1,\RR_2,S)$, we have that $(M_1\|\cdots\|M_N,\RR_1,\RR_2,S)\Downarrow$ iff $(M_1\|\cdots\|M_N,\RR_1\cup\RR_2,S)\Downarrow$.
\end{lemma}

We now proceed to syntactic composition of $N$-configurations.
Given a pair $\rho_1\asymp\rho_2$, we define a single quadruple corresponding to their {syntactic} composition, called their \boldemph{internal composition}, as follows. Let 
$\rho_1=(\vec\CC,\RR_1,\PP_1,\AA_1,S_1)$ and $\rho_2=(\vec\CC',\RR_2,\PP_2,\AA_2,S_2)$ and, for each $i$,
$\CC_i=(\EE_i,X_i)$ and $\CC_i'=(\EE_i',X_i')$, with $\{X_i,X_i'\}=\{M_i,-\}$,
and we let $k_i=1$ just if $X_i=M_i$.
We let the internal composition of $\rho_1$ and $\rho_2$ be the quadruple:
\[
\rho_1\doublewedge\rho_2 = ((\EE_1\doublewedge\EE_1')[M_1^{k_1}]\|\cdots\|(\EE_N\doublewedge\EE_N')[M_N^{k_N}],\RR_1,\RR_2,S_1\uplus S_2)
\]
where compatible evaluation stacks $\EE,\EE'$ are composed into a single evaluation context $\EE\doublewedge\EE'$, as follows.
\begin{align*}
(m::E::\EE)\doublewedge(m::\EE')
&=(\EE\doublewedge\EE')[E[\rett{m}\bullet]^1]\\
(m::\EE')\doublewedge(m::E::\EE)
&=(\EE\doublewedge\EE')[E[\rett{m}\bullet]^2]
\end{align*}
and 
$[]\doublewedge[] = \bullet$.
Unfolding the above, we have that, for example:
\begin{align*}
&[m_k,E_k,m_{k-1},m_{k-2},E_{k-2},\cdots,m_1,E_1]\\
&\!\doublewedge[m_k,m_{k-1},E_{k-1},m_{k-2},\cdots,m_1]
= E_1^1[\rett{m_1^1}E_2^2[\cdots E_k^{k'}[\rett{m_k^{k'}\!}\bullet]\cdots]]
\end{align*}
where $k'=2-(k\mod 2)$.
\cutout{
Note that the configuration above is a closed-term configuration (i.e.\ there are only three components), and corresponds to evaluating $N$ terms in parallel. We will be using the symbol\ $\tred{t}[N]$\ for reductions between configurations of this form. 

We will next relate the transition systems induced by composing compatible configurations semantically (via $\oslash$) and syntactically (via $\doublewedge$). }

We proceed to fleshing out the correspondence.
We observe that an $L$-configuration $\rho$ can be the final configuration of a trace just if all its components are $O$-configurations with empty evaluation stacks. On the other hand, for $C$-configurations, we need to reach $P$-configurations with terms (). Thus, we call an $N$-configuration $\rho$ \emph{final} if $\rho=(\vec\CC,\RR,\PP,\AA,S)$ and
either $\CC_i=([],-)$ for all $i$, or 
$\CC_i=([],())$ for all $i$.

Let us write $(\SS_1,\qweto_1,\FF_1)$ for the transition system induced from external composition, and $(\SS_2,\qweto_2,\FF_2)$ be the transition system derived from internal composition:
\begin{itemize}
\item $\SS_1=\{\rho\oslash\rho'\mid \rho\asymp\rho'\}$, $\FF_1=\{\rho\oslash\rho'\in\SS_1\mid\rho,\rho'\text{ final}\}$, and $\qweto_1$ the transition relation $\lto$ defined previously.
\item $\SS_2=\{(M_1\|\cdots\|M_N,\vec\RR,S)\mid (M_1\|\cdots\|M_N,\RR_1\uplus\RR_2,S)\text{ valid}\}$, $\FF_2=\{x\in\SS_2\mid x\text{ final}\}$, and $\qweto_2$ the transition relation $\pred{N}{}'$ defined above.
\end{itemize}
A relation $R\subseteq\SS_1\!\times\SS_2$ 
is called a \emph{bisimulation} if, for all $(x_1,x_2)\in\! R$:
\begin{itemize}
\item $x_1\in\FF_1$ iff $x_2\in\FF_2$,
\item if $x_1\qweto_1 x_1'$ then $x_2\qweto_2 x_2'$ and $(x_1',x_2')\in R$,
\item if $x_2\qweto_2 x_2'$ then $x_1\qweto_1 x_1'$ and $(x_1',x_2')\in R$.
\end{itemize}
\cutout{
where we write $\_^=$ for the reflexive closure operator. Here,
$\mapsto_{\eta}$ is defined by the following eta-rule:
\[
(\textsc{Eta})
\quad
{(M_1\|\cdots\|M\|\cdots\|M_N,\RR,S)\mapsto_{\eta}(M_1\|\cdots\|M'\|\cdots\|M_N,\RR\uplus(m'\mapsto \lambda x.mx),S)}
\]
for $m\in\dom(\RR)$, with $M'$ obtained from $M$ by replacing an arbitrary number of occurrences of $m$ inside it with $m'$.
Note that the addition of eta-rules as above does not affect the termination of a term evaluation apart from adding $\eta$-redexes to be reduced.}
Given  $(x_1,x_2)\in\SS_1\times\SS_2$, we say that $x_1$ and $x_2$ are \emph{bisimilar}, written $x_1\sim x_2$, if $(x_1,x_2)\in R$ for some bisimulation $R$.

\begin{lemma}\label{l:comp1}
Let $\rho\asymp\rho'$ be compatible $N$-configurations. Then, $(\rho\oslash\rho')\sim(\rho\doublewedge\rho')$.
\end{lemma}

\cutout{
In general, though, we can have the case where $M_1''$ is a value $v'$, which needs to be returned internally between $\rho_1$ and $\rho_2$ before the redex, which is present in some evaluation context in $\EE_1$ or $\EE_2$, is reached. Thus, a series of returns is necessary before the transition can be simulated.
Let us consider the most interesting case, where $M_1=mv'$ and $M_1'=M\{v'/x\}$, with $\RR(m)=\lambda x.M$ and $v'\in\Meths$.
The described reasoning gives us the following:
\begin{align*}
(\CC_1^1,S_1)&\xr{\ret{m_1(v_1)}}\cdots\xr{\ret{m_n(v_n)}}(\EE_1',E'[mv_n],\cdots)\\
&\xr{\call{m(v_{n+1})}}(m::E'::\EE_1',\RR_1',\PP_1',\AA_1',S_1)\\
(\CC_1^2,S_2)&\xr{\ret{m_1(v_1)}}\cdots\xr{\ret{m_n(v_n)}}(\EE_2',\cdots)\\
&\xr{\call{m(v_{n+1})}}(m::\EE_2',M\{v_{n+1}/x\},\RR_2',\PP_2',\AA_2',S_2)
\end{align*}
where $\PP_1'$ extends $\PP_1$ with the odd-indexed $v_i$'s,
$\AA_1'$ extends $\AA_1$ with the even-indexed $v_i$'s, and $\RR_1'=\RR_1\uplus\{(v_i\mapsto\lambda x.v_{i-1}x)\mid i\text{ odd}\}$, with $v_0=v'$; and dually for $\RR_2',\PP_2',\AA_2'$. 
Moreover, $E=(\EE_1'\doublewedge\EE_2')[E']$.
Writing $\rho_1'$ and $\rho_2'$ for the resulting $N$-configurations, we have $\rho_1\oslash\rho_2\to^*\rho_1'\oslash\rho_2'$.
In addition, $(M_1'\|M_2\|\cdots\|M_N,\RR',S')\mapsto_\eta\cdots\mapsto_\eta\rho_1'\doublewedge\rho_2'$ by $n+1$ $\eta$-expansions.
}

Recall we write $\bar h$ for the $O/P$ complement of the history $h$. We can now prove Theorem~\ref{t:comp}, which states that the behaviour of a library $L$ inside a context $C$ can be deduced by composing the semantics of $L$ and $C$.

\paragraph{Theorem~\ref{t:comp}}{\em%
Let $L:\Theta\to\Theta'$, $L':1\to\Theta,\Theta_1$ and $\Theta',\Theta_1\vdash M_1,\cdots,M_N:\unit$,
with $L$, $L'$ and $M_1;\cdots;M_N$ accessing pairwise disjoint parts of the store.
Then,
 $\link{L'\comp L}{(M_1\|\cdots\| M_N)}\,{\Downarrow}$ iff there is 
$h\in\sem{L}_N$ such that 
$\bar h\in\sem{\link{L'\comp-}{(M_1\|\cdots\| M_N)}}$.}

\begin{proof}
Let $C$ be the context $\link{L'\comp-}{(M_1\|\cdots\| M_N)}$, and suppose $(L)\redL^*(\epsilon,\RR_0,S_0)$ and
$(L')\redL^*(\epsilon,\RR_0',S_0')$ with $\dom(\RR_0)\cap\dom(\RR_0')=\dom(S_0)\cap\dom(S_0')=\emptyset$.
We set:
\begin{align*}
\rho_0 &=(([],-)\|\cdots\|([],-),\RR_0,(\emptyset,\Theta'),(\Theta,\emptyset),S_0)\\ \rho_0'&=(([],M_1)\|\cdots\|([],M_N),\RR_0',(\Theta,\emptyset),(\emptyset,\Theta'),S_0')
\end{align*}
We pick these as the initial $N$-configurations for  $\sem{L}_N$ and $\sem{C}$ respectively. Moreover,
we have that $(L'\comp L)\redL^*(\epsilon,\RR_0'',S_0'')$ where 
$\RR_0''=\{(m,(\RR_0\uplus\RR_0')(m)\{{!\vec r}/\vec m\})\mid m\in\dom(\RR_0\uplus\RR_0')\}$ and $S_0''=(S_0\uplus S_0')\{{!\vec r}/\vec m\}\uplus_\s\{(r_i,m_i)\mid i=1,\cdots,n\}$, assuming $\Theta=\{m_1,\cdots,m_n\}$ and $r_1,\cdots,r_n$ are fresh references of corresponding types.
Hence, the initial triple for $\sem{C[L]}$ is taken to be
$\phi_0=(([],M_1)\|\cdots\|([],M_N),\RR_0'',S_0'')$.
On the other hand,
$\rho_0'\doublewedge\rho_0=(([],M_1)\|\cdots\|([],M_N),\RR_0',\RR_0,S_0\uplus S_0')$ and, using also Lemma~\ref{lem:eta}, we have that $\phi_0\Downarrow$ iff $\rho_0'\doublewedge\rho_0\Downarrow$.

Then, for the forward direction of the claim, from 
$\phi_0\,{\Downarrow}$ we obtain that $\rho_0'\doublewedge\rho_0\,{\Downarrow}$. From the previous lemma, we have that so does $\rho_0'\oslash\rho_0$. From the latter reduction we obtain the required common history.
Conversely, suppose 
$h\in\sem{L}_N$ and $\bar h\in\sem{C}$. WLOG, assume that $\Meths(h)\cap(\dom(\RR_0)\cup\dom(\RR_0'))\subseteq\Theta\cup\Theta_1\cup\Theta'$ (we can appropriately alpha-covert $\RR_0$ and $\RR_0'$ for this).
Then, $\rho_0$ and $\rho_0'$ both produce $h$, with opposite polarities. By definition of the external composite reduction, we then have that $\rho_0'\oslash\rho_0$ reduces to some final state. By the previous lemma, we have that $\rho_0'\doublewedge\rho_0$ reduces to some final quadruple, which in turn implies that $\phi_0\,{\Downarrow}$, i.e.\
$\link{L'\comp L}{(M_1\|\cdots\| M_N)}\,{\Downarrow}$.
\end{proof}



\subsection{Lemma~\ref{lem:eta}}
\newcommand\und[1]{\underline{#1}}
\newcommand\td{\tilde}
\newcommand{\sz}{\mathsf{size}}

We purpose to show that, for any legal $(M_1\|\cdots\|M_N,\RR_1,\RR_2,S)$,  $(M_1\|\cdots\|M_N,\RR_1,\RR_2,S)\Downarrow$ iff $(M_1\|\cdots\|M_N,\RR_1\cup\RR_2,S)\Downarrow$.

We prove something stronger. For any repository $\RR$ whose entries are of the form $(m,\lambda x.m'x)$, we define a directed graph ${\cal G}(\RR)$ where vertices are all methods appearing in $\RR$, and $(m,m')$ is a (directed) edge just if $\RR(m)=\lambda x.m'x$. 
In such a case, we call $\RR$ an \boldemph{expansion class} if ${\cal G}(\RR)$ is acyclic and all its vertices have at  most one outgoing edge.\ntnote{internal note: isn't the outgoing degree anyway always leq $1$?}
Moreover, given an expansion class $\RR$, we define the method-for-method substitution $\{\RR\}$ that assigns to each vertex $m$ of ${\cal G}(\RR)$ the (unique) leaf $m'$ such that there is a directed path from $m$ to $m'$ in
${\cal G}(\RR)$. Let us write ${\cal L}(\RR)$ for the set of leaves of ${\cal G}(\RR)$.
For any quadruple $\phi=(E_1[M_1]\|\cdots\|E_N[M_N],\RR_1,\RR_2,S)$ and expansion class $\RR\subseteq\RR_1\cup\RR_2$, we define the triple:
\begin{align*}
\phi^{\#\RR} &=(\und{E_1[M_1]}\|\cdots\|\und{E_N[M_N]},\RR_1\cup\RR_2,S)\{\RR\}
\\
&=(\und{E_1[M_1]}\{\RR\}\|\cdots\|\und{E_N[M_N]}\{\RR\},(\RR_1\cup\RR_2)\{\RR\},S\{\RR\})
\end{align*}
where 
$\RR'\{\RR\}=\{(m,\RR'(m)\{\RR\})\ |\ m\in\dom(\RR'\setminus\RR)\cup {\cal L}(\RR)\}$, 
$S\{\RR\}=(S\upharpoonright\Refs_\tint)\cup\{(r,S(r)\{\RR\})\ |\ r\in\dom(S)\setminus\Refs_\tint\}$, and
$\underline{E[M]}$ is the term obtained from $E[M]$ by removing all tagging.

We next define a notion of indexed bisimulation between the transition systems produced from quadruples and triples respectively. Given an expansion class $\RR$, a relation $R_\RR$ between quadruples and triples is called an \emph{$\RR$-bisimulation} if, whenever $\phi_1 R_\RR\phi_2$:
\begin{itemize}
\item $\phi_1$ final implies $\phi_2$ final
\item $\phi_2$ final implies $\phi_2\Downarrow$
\item $\phi_1\pred{N}{}'\phi_1'$ implies $\phi_2\pred{N}{}^=\phi_2'$ and $\phi_1' R_{\RR'} \phi_2'$ for some expansion class $\RR'\supseteq\RR$
\item $\phi_2\pred{N}{}\phi_2'$ implies $\phi_1\pred{N}{}'^*\phi_1'$ and $\phi_1' R_{\RR'} \phi_2'$ for some expansion class $\RR'\supseteq\RR$.
\end{itemize}
Thus, Lemma~\ref{lem:eta} directly follows from the next result.

\begin{lemma}
For all expansion classes $\RR$, the relation $R_\RR = $
\[
\{
(\phi,\phi^{\#\RR}\!)\mid \phi=({E_1[M_1]}\|\cdots\|E_N[M_N],\vec\RR,S)\text{ legal }\land\RR\subseteq\RR_1\cup\RR_2
\}
\]
is a bisimulation.
\end{lemma}
\begin{proof}
Suppose $\phi R_{\RR}\phi^{\#\RR}$.
We note that finality conditions are satisfied: if $\phi$ is final then so is $\phi^{\#\RR}$; while if $\phi^{\#\RR}$ is final then all its contexts are from the grammar:
\[
E'\,::=\,\bullet\mid\rett{m^i}E'
\]
so $\phi\Downarrow$ by acyclicity of ${\cal G}(\RR)$.
\\
Suppose now $\phi\pred{N}{}\phi'$, say due to $(E_1[M_1],\RR_1,\RR_2,S)\tred{1}{}'(E_1'[M_1'],\RR_1',\RR_2',S')$. In case the reduction is not a function call or return, then it can be clearly simulated by $\phi^{\#\RR}$. Otherwise, suppose:
\begin{itemize}
\item $(E_1[m^iv],\vec\RR,S)\tred{1}{}'(E_1[M\{v/x\}^i],\vec\RR,S)$. If $m\notin\dom(\RR)$ then, writing $\RR_{12}$ for $\RR_1\cup\RR_2$, the above can be simulated by $(\underline{E_1}[mv],\RR_{12},S)\{\RR\}\tred{1}{}(\underline{E_1}[M\{v/x\}],\RR_{12},S)\{\RR\}$. If, on the other hand, $m\in\dom(\RR)$, suppose $\RR_i(m)=\lambda x.m'x$, then $M=m'x$ and $m\{\RR\}=m'\{\RR\}$ so we have:
\[
\underline{E_1[M\{v/x\}^i]}\{\RR\}=\underline{E_1[(m'v)^i]}\{\RR\}=
\underline{E_1[(mv)^i]}\{\RR\}
\]
and $E_1[(mv)^i]=E_1[m^iv]$ by the way the semantics was defined,
so $\phi'^{\#\RR}=\phi^{\#\RR}$.
\cutout{
\item $(E_1[m^iv],\vec\RR,S)\tred{1}{}'(E_1[\rett{m^i}M\{v/x\}^{3-i}],\vec\RR,S)$,
with $\RR_{3-i}(m)=\lambda x.M$ and $v\notin\Meths$.
If $m\notin\dom(\RR)$ then, as above, the reduction can be simulated by $\phi^{\#\RR}$. If $\RR(m)=\lambda x.m'x$ then 
\[
\underline{E_1[\rett{m^i}M\{v/x\}^{3-i}]}\{\RR\}=
\underline{E_1[m'^{3-i}v]}\{\RR\}=
\underline{E_1[m^iv]}\{\RR\}
\]
so $\phi'^{\#\RR}=\phi^{\#\RR}$.}
\item 
$(E_1[m^iv],\vec\RR,S)\tred{1}{}'(E_1[\rett{m^i}M\{v'/x\}^{3-i}],\vec\RR',S)$,
with $\RR_{3-i}(m)=\lambda x.M$, $\Meths(v)=\{m_1,\cdots,m_k\}$, $v'=\{\vec m'/\vec m\}$ and $\vec\RR'=\vec\RR\uplus_i\{m_j'\mapsto \lambda x.m_jx\mid 1\leq j\leq k\}$. 
Let $\RR'=\RR\uplus \{m_j'\mapsto \lambda x.m_jx\mid 1\leq j\leq k\} \subseteq\RR_1'\cup\RR_2'$.
If $m\notin\dom(\RR)$ then 
$(\und{E_1}[mv],\RR_{12},S)\{\RR\}\tred{1}{}(\und{E_1}[M\{v/x\}],\RR_{12},S)\{\RR\}$, and we have:
\begin{align*}
\underline{E_1[\rett{m^i}M\{v'/x\}^{3-i}]}\{\RR'\} &=
\underline{E_1}[M\{v'/x\}]\{\RR'\}
\\
&=
\underline{E_1}[M\{v/x\}]\{\RR\}
\end{align*}
Moreover, $\RR_{12}\{\RR\}=(\RR_1'\cup\RR_2')\{\RR'\}$ and $S\{\RR\}=S\{\RR'\}$, so $\phi'^{\#\RR}= 
(\und{E_1}[M\{v/x\}],\RR_{12},S)\{\RR\}$.
\\
On the other hand, if $\RR(m)=\lambda x.m''x$ then:
\begin{align*}
\und{E[\rett{m^i}M\{v'/x\}^{3-i}]}\{\RR'\} &=
\und{E}[m''v']\{\RR'\}\\ & = \und{E}[m''v]\{\RR\} = \und{E}[mv]\{\RR\}  
\end{align*}
so $\phi^{\#\RR}=\phi'^{\#\RR'}$.
\item Finally, the cases for method-return reductions are treated similarly as above. 
\end{itemize}
Suppose now $\phi^{\#\RR}\pred{N}{}\phi'$, 
where recall that we write $\phi$ as $({E_1[M_1]}\|\cdots\|E_N[M_N],\vec\RR,S)$.
We show by induction on $\sz_\RR(E_1[M_1],\cdots,E_N[M_N])$ that 
$\phi\pred{N}{}'\phi''$ and $\phi' R_{\RR'}\phi''$ for some $\RR'\supseteq\RR$.
The size-function we use measures the length of ${\cal G}(\RR)$-paths that appear inside its arguments:
\begin{align*}
\sz_\RR(E_1[M_1],\cdots,E_N[M_N]) 
& =
\sz_\RR(E_1[M_1])+\cdots+\sz_\RR(E_N[M_N])\\
\sz_\RR(E[M]) &=\sum_{m\in X_1}2|m|_\RR+\sum_{m\in X_2}1
\end{align*}
where $X_1$ is the multiset containing all occurrences of methods $m\in\dom(\RR)$ inside $E[M]$ in call position (e.g.\ $mM'$), and $X_2$ contains all occurrences of methods $m\in\dom(\RR)$ inside $E[M]$ in return position (i.e.\ $\rett{m^i}\cdots$).
We write $|m|_\RR$ for the length of the unique directed path from $m$ to a leaf in ${\cal G}(\RR)$. The fact that $X_1,X_2$ are multisets reflects that we count all occurrences of $m$ in call/return positions.
Suppose WLOG that the reduction to $\phi'$ 
is due to some 
$(\und{E_1[M_1]},\RR_{12},S)\{\RR\}\tred{1}{}(E'[M'],\RR',S')$. If the reduction happens inside $\und{M_1}\{\RR\}$ (this case also encompasses the base case of the induction) then the only case we need to examine is that of the reduction being a method call. In such a case, suppose we have $\und{E_1[M_1]\{\RR\}}=E[mv]$, $E'=E$, $M'=M\{v/x\}$ and $\RR_{12}\{\RR\}(m)=\lambda x.M$.
Then, $E_1[M_1]=\tilde E[\tilde m^i\tilde v]$ for some $\tilde E,\tilde m,\tilde v$ such that $\tilde m\{\RR\}=m$, 
$\tilde v\{\RR\}=v$ and $\und{\tilde E}\{\RR\}=E$. 
If $m\not=\td m$ then, supposing $\RR(\td m)=\lambda x.\td m'x$ we have the following cases:
\begin{itemize}
\item $(\td E[\tilde m^i\tilde v],\vec\RR,S)\tred{1}{}'(\td E[\td m'^i\td v],\vec\RR,S)=\phi_1''$
\item $(\td E[\tilde m^i\tilde v],\vec\RR,S)\tred{1}{}'(\td E[\rett{\td m^i}(\td m'v')^{3-i}],\vec\RR',S)=\phi_1''$, with $\vec\RR'=\vec\RR\uplus_{3-i}\{m'_j\mapsto\lambda x.m_j x\mid 1\leq j\leq k\}$, etc.
\end{itemize}
Let $\phi''$ be the extension of $\phi_1''$ to an $N$-quadruple by using the remaining $E_i[M_i]$'s of $\phi$, so that $\phi\pred{N}{}'\phi''$.
In the first case above we have that $\phi''^{\#\RR}=\phi$, and in the latter that $\phi''^{\#\RR'}=\phi$ (with $\RR'=\RR\uplus\{m'_j\mapsto\lambda x.m_j x\mid 1\leq j\leq k\}$), and we appeal to the IH.
\\
Suppose now that $\td m=m$ and $\RR_{12}(m)=\lambda x.\td M$. Then, one of the following is the case:
\begin{itemize}
\item $(\td E[\tilde m^i\tilde v],\vec\RR,S),\vec\RR,S)\tred{1}{}'(\td E[\td M\{\td v/x\}^i],\vec\RR,S)=\phi_1''$
\item $(\td E[\tilde m^i\tilde v],\vec\RR,S)\tred{1}{}'(\td E[\rett{\td m^i}\td M\{v'/x\}^{3-i}],\vec\RR',S)=\phi_1''$, with $\vec\RR'=\vec\RR\uplus_{3-i}\{m'_j\mapsto\lambda x.m_j x\mid 1\leq j\leq k\}$, etc. 
\end{itemize}
Extending $\phi_1''$ to $\phi''$ as above,
in the former case we then have that $\phi''^{\#\RR}=\phi'$, and in the latter that $\phi''^{\#\RR'}=\phi'$, as required.
\\
Finally, let us suppose that $M_1$ is some value $v$.
Then, we can write $E_1$ as $E_1=E_2[E']$, with $E'$ coming from the grammar
$E'\,::=\,\bullet\mid\rett{m^i}E'$
and $E_2$ not being of the form $E''[\rett{m^i}\bullet]$. Observe that $\und{E_1}=\und{E_2}$. If $E'=\bullet$ then by a case analysis on $E_1$ we can see that $\phi^{\#\RR}$ can simulate the reduction. Otherwise, $(E_2[E'[v]],\vec\RR,S)\tred{1}{}'(E_2[E''[v'^i]],\vec\RR',S)$ whereby 
$E'=E''[\rett{m^i}\bullet]$ and
$\vec\RR'=\vec\RR\uplus_{3-i}\{m'_j\mapsto\lambda x.m_jx\mid 1\leq j\leq k\}$, etc. We have that
\[
\phi_1''=(\und{E_2[E''[v'^i]]},\vec\RR',S)\{\RR'\}=(\und{E_2[E'[v]]},\vec\RR,S)\{\RR\}
\]
and hence, extending $\phi_1''$ to $\phi''$, we have $\phi''^{\#\RR'}=\phi^{\#\RR}$. We can now appeal to the IH.
\end{proof}


\subsection{Lemma~\ref{l:comp1}}
\emph{%
Let $\rho\asymp\rho'$ be compatible $N$-configurations. Then, $(\rho\oslash\rho')\sim(\rho\doublewedge\rho')$.}
\smallskip

\noindent
We prove that the relation
$
R =\{ (\rho_1\oslash\rho_2,\rho_1\doublewedge\rho_2)\mid \rho_1\asymp\rho_2\}
$
 is a bisimulation. Let us suppose that $(\rho_1\oslash\rho_2,\rho_1\doublewedge\rho_2)\in R$.
\begin{itemize}
\item 
Suppose $\rho_1\oslash\rho_2\qweto_1\rho_1'\oslash\rho_2'$. If the transition is due to ({\sc Int1}) then $\rho_2=\rho_2'$ and we can see that $\rho_1\doublewedge\rho_2\pred{N}{}'\rho_1'\doublewedge\rho_2$. Similarly if the transition is due to ({\sc Int2}). Suppose now we used instead ({\sc Call}), e.g.\ $\rho_1\pred{N}{(1,\call{m(v)})}\rho_1'$ and
$\rho_2\pred{N}{(1,\call{m(v)})}\rho_2'$, and let us consider the case where $v\in\Meths$ (the other case is simpler). Then, assuming $\rho_1=(\CC_1^1\|\cdots,\RR_1,\PP_1,\AA_1,S_1)$ and $\rho_2=(\CC_1^2\|\cdots,\RR_2,\PP_2,\AA_2,S_2)$, we have that either of the following scenarios holds, for some $\texttt{x}\in\{\k,\l\}$:
$\CC_1^1=(\EE_1,E[mm'])$, $\CC_1^2=(\EE_2,-)$ and
\begin{align*}
&
(\EE_1,E[mm'],\RR_1,\PP_1,\AA_1,S_1)\tred{1}{\call{m(v)}}\\ 
&\hspace{7.5mm}
(m::E::\EE_1,\RR_1\uplus(v\mapsto\lambda x.m'x),\PP_1\cup_{\mathtt{x}}\{v\},\AA_1,S_1)\\
&(\EE_2,-,\RR_2,\PP_2,\AA_2,S_2)\tred{1}{\call{m(v)}}\\ 
&\hspace{23mm}
(m::\EE_2,M\{v/x\},\RR_2,\PP_1,\AA_1\cup_{\mathtt{x}}\{v\},S_2)
\end{align*}
or its dual, where $\rho_2$ contains the code initiating the call. 
Focusing WLOG in the former case and setting $S=S_1\uplus S_2$:
\begin{align*}
\rho_1\doublewedge\rho_2&=((\EE_1\doublewedge\EE_2)[E[m^1m']]\|\cdots,\RR_1,\RR_2,S)\\
&\qweto_2
((\EE_1\doublewedge\EE_2)[E[\rett{m^1}M\{v/x\}^2]]\|\cdots,\RR_1',\RR_2,S)\\&=\rho_1'\doublewedge\rho_2'\quad(\RR_1'=\RR_1\uplus(v\mapsto\lambda x.m'x))
\end{align*}
The case for ({\sc Retn}) is treated similarly.
\item 
Suppose $\rho_1\doublewedge\rho_2=(E[M_1]\|M_2\|\cdots\|M_N,\vec\RR,S)\qweto_2 (E[M_1']\|$ $M_2\|\cdots\|M_N,\vec\RR',S')$ and let 
$\rho_1=((\EE_1,M_1'')\|\cdots,\RR_1,\PP_1,\AA_1,S_1)$
and 
$\rho_2=((\EE_2,-)\|\cdots,\RR_2,\PP_2,\AA_2,S_2)$, where $(\EE_1\doublewedge\EE_2)[M_1'']=E[M_1]$. If the redex $M_1$ is not of the forms  $M_1=m^1v$ or $M_1=\rett{m^1}v$, with $m\in\dom(\RR_2)$, then the reduction can clearly be simulated by $\rho_1\oslash\rho_2$ (internally, by $\rho_1$).
Otherwise, similarly as above, the reduction can be simulated by a mutual call/return of $m$.
\end{itemize}
Finally, it is clear that
$\rho_1\oslash\rho_2$ is final iff $\rho_1\doublewedge\rho_2$ is final.
\qed

\section{General compositionality}\label{sec:comp2}

This  compositionality result will  allow us to compose histories of component libraries in order to obtain those of their composite library.
Let  $L_1:\Theta_1\to\Theta_2$ and $L_2:\Theta_1'\to\Theta_2'$. The 
semantic
composition will be guided by two sets of names $\Pi,\Rho$.
$\Pi$ contains method names that are shared between by the respective libraries and their context. Thus
$\Pi\supseteq\Theta_1\cup\Theta_1'\cup\Theta_2\cup\Theta_2'$.
The names in $\Rho$, on the other hand, will be used for private communication between $L_1$ and $L_2$.
Consequently,  $\Pi\cap\Rho$ consists of names that can be used both for internal communication between $L_1$ and $L_2$, and   for contextual interactions, i.e.
$\Pi\cap\Rho=(\Theta_1\cup\Theta_1')\cap(\Theta_2\cup\Theta_2')$.

Given $h_i\in\sem{L_i} (i=1,2)$, we define the \emph{composition} of $h_1$ and $h_2$, written $h_1\doublewedge^{\sigma}_{\Pi,\Rho}h_2$, as a partial operation depending 
on $\Pi,\Rho$ and an additional parameter $\sigma\in\{0,1,2\}^*$ which we call a \emph{scheduler}.
It is given inductively as follows.
We let $\epsilon\doublewedge^{\epsilon}_{\Pi,\Rho}\epsilon=\epsilon$ and:
\begin{align*}
&(t,\call{m(v)})s_1\doublewedge^{0\sigma}_{\Pi,\Rho}(t,\call{m(v)})s_2
=s_1\doublewedge^{\sigma}_{\Pi,\Rho'}s_2\\
&(t,\ret{m(v)})s_1\doublewedge^{0\sigma}_{\Pi,\Rho}(t,\ret{m(v)})s_2
=s_1\doublewedge^{\sigma}_{\Pi,\Rho'}s_2\\
&(t,\call{m(v)})_{PY}s_1\doublewedge^{1\sigma}_{\Pi,\Rho}s_2
=(t,\call{m(v)})_{PY}(s_1\doublewedge^{\sigma}_{\Pi',\Rho}s_2)
\\
&(t,\ret{m(v)})_{PY}s_1\doublewedge^{1\sigma}_{\Pi,\Rho}s_2
=(t,\ret{m(v)})_{PY}(s_1\doublewedge^{\sigma}_{\Pi',\Rho}s_2)
\\
&(t,\call{m(v)})_{OY}s_1\doublewedge^{1\sigma}_{\Pi,\Rho}s_2
=(t,\call{m(v)})_{OY}(s_1\doublewedge^{\sigma}_{\Pi',\Rho}s_2)
\\
&(t,\ret{m(v)})_{OY}s_1\doublewedge^{1\sigma}_{\Pi,\Rho}s_2
=(t,\ret{m(v)})_{OY}(s_1\doublewedge^{\sigma}_{\Pi',\Rho}s_2)
\end{align*}
along with the dual rules for the last four cases (i.e.\ where we schedule 2 in each case). 
Note that the definition uses 
sequences of moves that are suffixes of histories (such as $s_i$). The above equations are subject to the following side conditions:
\begin{compactitem}
\item $\Meths(v)\cap(\Pi\cup \Rho)=\emptyset$, $\Pi'=\Pi\uplus\Meths(v)$ and $\Rho'=\Rho\uplus\Meths(v)$;
\item $m\in\Rho$ in the 0-scheduling cases;
\item $m\in\Pi$ in the 1-scheduling cases and, also, $m\in\Pi\setminus\Rho$ in the third case (the $P$-call);
\item in the 1-scheduling cases, we also require that the leftmost move with thread index $t$ in $s_2$ is not a $P$-move.
\end{compactitem}
History composition is a partial function: if the conditions above are not met, or $h_1,h_2,\sigma$ are not of the appropriate form, 
then the composition is undefined.
The above conditions ensure that the composed histories are indeed compatible and can be produced by composing actual libraries. For instance, the last condition 
corresponds to determinacy of threads: there can only be at most one component starting with a $P$-move in each thread $t$.
We then have the following correspondence. 

\cutout{
Given a 
history $h$ and a
set of methods $\Theta$, we let 
$h\setminus \Theta$ be the history obtained from $h$ by removing all moves that are calls or returns to methods in $\Theta$.}

\begin{theorem}\label{supercomp}
If $L_1:\Theta_1\to\Theta_2$ and  
$L_2:\Theta_1'\to\Theta_2'$ 
access disjoint parts of the store
then 
\[
\sem{L_1\cup L_2}_N = 
\{\, h\in\clg{H}\,\mid\, \exists \sigma,h_1\!\in\!\sem{L_1}_N\!, h_2\!\in\!\sem{L_2}_N\!.\,\,
h=h_1\doublewedge_{\Pi_0,\Rho_0}^{\sigma}h_2
\}\]
with $\Pi_0=\Theta_1\cup\Theta_2\cup\Theta_1'\cup\Theta_2'$ and
$\Rho_0=(\Theta_1\cup\Theta_1')\cap(\Theta_2\cup\Theta_2')$.
\end{theorem}

\cutout{
\begin{theorem}\label{supercomp}
For all $L'\!:\emptyset\to\Theta,\Theta_1$ and $L\!:\Theta\to\Theta'$,
\begin{align*}
&\sem{L'\comp L}_N\\
&=\{ h\setminus\Theta\mid h\in\clg{H}^L\!\land
\exists \sigma,h_2\in\sem{L}_N, h_1\in\sem{L'}_N.\,
h=h_1\!\doublewedge_{\Pi_0,\Theta}^{\sigma}h_2
\}
\end{align*}
with $\Pi_0=\Theta\cup\Theta'\cup\Theta_1$.
\end{theorem}
}


The rest of this section is devoted in proving the Theorem.
\bigskip

\noindent
Recall that
we examine library composition in the sense of union of libraries. This scenario is more general than the one of Appendix~\ref{sec:comp} as, during composition via union, the calls and returns of each of the component libraries may be caught by the other library or passed as a call/return to the outer context. 
Thus, the setting of this section comprises given libraries $L_1:\Theta_1\to\Theta_2$ and  
$L_2:\Theta_1'\to\Theta_2'$, such that $\Theta_2\cap\Theta_2'=\emptyset$, and relating their semantics to that of their union 
$L_1\cup L_2:(\Theta_1\cup\Theta_1')\setminus(\Theta_2\cup\Theta_2')\to\Theta_2\cup\Theta_2'$.

Given configurations for $L_1$ and $L_2$, in order to be able to reduce them together we need to determine which of their methods can be used for communication between them, and which for interacting with the external context, which represents player $O$ in the game. We will therefore employ a set of method names, denoted by $\Pi$ and variants, to register those methods used for interaction with the external context.
Another piece of information we need to know is in which component in the composition was the last call played, or whether it was an internal call instead. This is important so that, when $O$ (or $P$) has the choice to return to both components, in the same thread, we know which one was last to call and therefore has precedence. We use for this purpose sequences $w=(w_1,\cdots,w_N)$ where, for each $i$, $w_i\in\{0,1,2\}^*$. Thus, if e.g.\ $w_1=2w_1'$, this would mean that, in thread 1, the last call to $O$, was done from the second component; 
if, on the other hand, $w_1=0w_1'$ then the last call in thread 1 was an internal one between the two components. Given such a $w$ and some $j\in\{0,1,2\}$, for each index $t$, we write $j+_tw$ for $w[t\mapsto(jw_t)]$.

Let us fix libraries $L_1:\Theta_1\to\Theta_2$ and $L_2:\Theta_1'\to\Theta_2'$.
Let $\rho_1,\rho_2$ be $N$-configurations, and in particular $L$-configurations, and suppose that $\rho_1=(\vec\CC,\RR,\PP,\AA,S)$ and $\rho_2=(\vec\CC',\RR',\PP',\AA',S')$. 
Moreover, let $\Theta_1\cup\Theta_2\cup\Theta_1'\cup\Theta_2'\subseteq\Pi$.
We say that $\rho_1$ and $\rho_2$ are \boldemph{$(w,\Pi)$-compatible}, written $\rho_1\asymp_\Pi^w\rho_2$,
if $S,S'$ have disjoint domains 
and, for each $i$;
\begin{itemize}
\item $\CC_i=(\EE_L',M)$ and $\CC'_i=(\EE_L,-)$, or
$\CC_i=(\EE_L,-)$ and $\CC'_i=(\EE_L',M)$,
or
$\CC_i=(\EE_{L1},-)$ and $\CC'_i=(\EE_{L2},-)$.
\item We have 
$\Theta_1\subseteq\AA_l$, $\Theta_2\subseteq\PP_\k$, $\Theta_1'\subseteq\AA_\l'$, $\Theta_2'\subseteq\PP_\k'$ and,
setting
\[
\Rho=(\PP_\k\cap\AA_\l')\uplus(\PP_\l\cap\AA_\k')\uplus(\PP_\k'\cap\AA_\l)\uplus(\PP_l'\cap\AA_\k)
\]
we also have:
\begin{itemize}
\item $(\PP_\l\uplus\PP_\k\uplus\AA_l\uplus\AA_\k)\cap(\PP_\l'\uplus\PP_\k'\uplus\AA_l'\uplus\AA_\k')=\Rho\uplus(\Theta_1\cap\Theta_1')$,
\item
$\Pi\cap\Rho=(\Theta_2\cup\Theta_2')\cap(\Theta_1\cup\Theta_1')$, 
\item
$\Pi\cup\Rho=\PP_\l\cup\PP_\k\cup\PP_l'\cup\PP_\k'\cup\AA_\l\cup\AA_\k\cup\AA'_\l\cup\AA_\k'$.
\end{itemize}
\item The private names of $\RR$ do not appear in $\rho_2$, and dually for the private names of $\RR'$.
\item If $\CC_i=(\EE,\cdots)$ and $\CC_i'=(\EE',\cdots)$ then $\EE$ and $\EE'$ are \emph{$w_i$-compatible}, that is, either $\EE=\EE'=[]$ or:
\begin{itemize}
\item $\EE=m::\EE_1$ and $\EE'\in\EE_L$, with $m\in\Pi$, $w_i=1u$ and $\EE_1,\EE'$ are $u$-compatible,
\item or $\EE=m::\EE_1$ and $\EE'=m::E::\EE_2$, with $m\in\Rho$, $w_i=0u$ and $\EE_1,\EE_2$ are $u$-compatible,
\item or $\EE=m::E::\EE_1$ and $\EE'\in\EE_L$, with $m\in\Pi\setminus\Rho$, $w_i=1u$ and $\EE_1,\EE'$ are $u$-compatible,
\end{itemize}
or the dual of one of the three conditions above holds.
\end{itemize}
Given $\rho_1\asymp^w_\Pi\rho_2$, we let their external composition be denoted as $\rho_1\otimes^w_\Pi\rho_2$ (and note that now the notation is symmetric for $\rho_1$ and $\rho_2$) and define the semantics for external composition by these rules:
\begin{gather*}
\infer[\textsc{Int}_1]{\rho_1\otimes_\Pi^w\rho_2\lto\rho_1'\otimes_\Pi^w\rho_2}{\rho_1\pred{N}{}\rho_1'}
\\
\infer[\textsc{Call}\;(m\in\Rho)]{\rho_1\otimes_\Pi^w\rho_2\lto\rho_1'\otimes_\Pi^{0+_tw}\rho_2'}{\rho_1\pred{N}{(t,\call{m(v)})}\rho_1'\quad\rho_2\pred{N}{(t,\call{m(v)})}\rho_2'}
\\
\infer[\textsc{Retn}\;(m\in\Rho)]{\rho_1\otimes_\Pi^{0+_tw}\rho_2\lto\rho_1'\otimes_\Pi^w\rho_2'}{\rho_1\pred{N}{(t,\ret{m(v)})}\rho_1'\quad\rho_2\pred{N}{(t,\ret{m(v)})}\rho_2'}
\\
\infer[\textsc{PCall}_1\;(m\in\Pi\setminus\Rho)]{\rho_1\otimes_\Pi^w\rho_2\xr{(t,\call{m(v)})_{PY}}\rho_1'\otimes_{\Pi'}^{1+_tw}\rho_2}{\rho_1\pred{N}{(t,\call{m(v)})_{PY}}\rho_1'}
\\
\infer[\textsc{PRetn}_1\;(m\in\Pi)]{\rho_1\otimes_{\Pi}^{1{+_t}w}\rho_2\xr{(t,\ret{m(v)})_{PY}}\rho_1'\otimes_{\Pi'}^{w}\rho_2}{\rho_1\pred{N}{(t,\ret{m(v)})_{PY}}\rho_1'}
\\
\infer[\textsc{OCall}_1\;(m\in\Pi)]{\rho_1\otimes_\Pi^w\rho_2\xr{(t,\call{m(v)})_{OY}}\rho_1'\otimes_{\Pi'}^{1+_tw}\rho_2}{\rho_1\pred{N}{(t,\call{m(v)})_{OY}}\rho_1'}
\\
\infer[\textsc{ORetn}_1\;(m\in\Pi\setminus\Rho)]{\rho_1\otimes_{\Pi}^{1{+_t}w}\rho_2\xr{(t,\ret{m(v)})_{OY}}\rho_1'\otimes_{\Pi'}^{w}\rho_2}{\rho_1\pred{N}{(t,\ret{m(v)})_{OY}}\rho_1'}
\end{gather*}
along with their dual counterparts {\sc(Int$_2$, XCall$_2$, XRetn$_2$)}.
The internal rules above have the same side-conditions on name privacy as before.
Moreover, in ({\sc XRetn}$i$) and ({\sc XCall}$i$), for {\sc X=O,P}, we let $\Pi'=\Pi\uplus_t\Meths(v)$  and impose that 
the $t$-th component of $\rho_{3-i}$ be an $O$-configuration
and $\Meths(v)\cap\Meths(\rho_{3-i})=\emptyset$.

We can now show the following.

\begin{lemma}
Let $\rho_1\asymp_\Pi^w\rho_2$ and suppose $\rho_1\otimes_\Pi^w\rho_2\xr{s}{\!\!}^*\;\rho_1'\otimes_{\Pi'}^{w'}\rho_2'$ for some sequence $s$ of moves. Then, $\rho_1'\asymp_{\Pi'}^{w'}\rho_2'$.
\end{lemma}

We next juxtapose the semantics of external composition to that obtained by internally composing the libraries and then deriving the multi-threaded semantics of the result. As before, we call the latter form \emph{internal composition}. The traces we obtain are produced from a transition relation, written $\pred{N}{}'$, between configurations of the form $(\CC_1\|\cdots\|\CC_N,\RR_1,\RR_2,\PP,\AA,S)$, 
also written $(\vec\CC,\vec\RR,\PP,\AA,S)$. In particular, in each $\CC_i=(\EE_i,X_i)$ with $X_i=E_i[M_i]$ or $X_i=-$, $E_i$ is selected from the extended evaluation contexts and $\EE_i$ is 
an \emph{extended $L$-stack}, that is, of either of the following two forms:
\[
\EE_{\sf ext}\, ::=\, [] \mid m^i::E::\EE_{\sf ext}'\qquad
\EE_{\sf ext}'\,::=\,m::\EE_{\sf ext}
\]
where $E$ is again from the extended evaluation contexts.

First, given $u$-compatible evaluation stacks $\EE,\EE'$, 
we construct a pair $\EE\doublewedge^{\!u}\EE'$ consisting of an extended evaluation context and an extended $L$-stack, as follows.
Given $\EE\doublewedge^{\!u}\EE'=(E',\EE'')$:
\begin{align*}
&(m::E::\EE)\doublewedge^{\!0u}(m::\EE')
= (E'[E[\rett{m}\bullet]^1],\EE'')
\\
&(m::\EE)\doublewedge^{\!0u}(m::E::\EE')
= (E'[E[\rett{m}\bullet]^2],\EE'')
\\
&(m::\EE)\doublewedge^{\!1u}\EE'
= \EE\doublewedge^{\!2u}(m::\EE')
=(\bullet,m::E'::\EE'')
\\
&(m::E::\EE)\doublewedge^{\!1u}\EE'
= \EE\doublewedge^{\!2u}(m::E::\EE')\\
&\hspace{40mm}=(\bullet,m::E'[E]::\EE'')
\;\;\text{if }\EE'\in\EE_L
\end{align*}
and 
$[]\doublewedge^\epsilon[] = (\bullet,[])$.

For each pair $\rho_1\asymp_\Pi^w\rho_2$, we define a configuration corresponding to their {syntactic} composition as follows. Let 
$\rho_1=(\CC_1\|\cdots\|\CC_N,\RR_1,\PP_1,\AA_1,S_1)$ and $\rho_2=(\CC_1'\|\cdots\|\CC_N',\RR_2,\PP_2,\AA_2,S_2)$ and, for each $i$, 
$\CC_i=(\EE_i,X_i)$ and $\CC_i'=(\EE_i',X_i')$. 
If $\EE_i\doublewedge^{\!u}\EE_i'=(E_i,\EE_i'')$, 
we set:
\[
\CC_i\doublewedge^{\!u}\CC_i'=
\begin{cases}
(\EE_i'',E_i[M^1])
& \text{ if } X_i=M\text{ and }X_i'=-\\
(\EE_i'',E_i[M^2])
& \text{ if } X_i=-\text{ and }X_i'=M\\
(\EE_i'',-)
& \text{ if }X_i=X_i'=-
\end{cases}
\]
We then let the internal composition of $\rho_1$ and $\rho_2$ be:
\[
\rho_1\doublewedge_\Pi^w\rho_2 = (\CC_1\doublewedge^{\!w_1}\!\CC_1'\|\cdots\|\CC_N\doublewedge^{\!w_N}\!\CC_N',
\RR_1,\RR_2,\PP',\AA',
S_1\uplus S_2)
\]
where we set
$\PP'=((\PP_{1\l}\uplus\PP_{2\l})\cap\Pi,(\PP_{1\k}\uplus\PP_{2\k})\cap\Pi)$ and $\AA'=((\AA_{1\l}\cup\AA_{2\l})\cap(\Pi\setminus\Rho),(\AA_{1\k}\uplus\AA_{2\k})\cap\Pi)$.

Now, as expected, the definition of $\pred{N}{}'$ builds upon $\tred{t}{}'$. The definition of the latter is given by the following rules.
\begin{gather*}
\infer[(\textsc{Int}')]
{(\EE,E[M],\vec\RR,\PP,\AA,S) \tred{t}{}' (\EE,E'[M'],\vec\RR',\PP,\AA,S')}
{(E[M],\vec\RR,S) \tred{t}{}' (E'[M'],\vec\RR',S')}
\\
(\EE,E[m^iv],\vec\RR,\PP,\AA,S) \tred{t}{\call{m(v')}_{PY}}{\!\!\!}'\ (m^i\!::\!E::\EE,-,\vec\RR',\PP',\AA,S)
\tag{\textsc{PQy}$'$}
\\
 (m::\EE,v,\vec\RR,\PP,\AA,S) \tred{t}{\ret{m(v')}_{PY}}{\!\!\!}'\ (\EE,-,\vec\RR',\PP',\AA,S)
\tag{\textsc{PAy}$'$}
\\
(\EE,-,\vec\RR,\PP,\AA,S) \tred{t}{\call{m(v)}_{OY}}{\!\!\!}'\ (m::\EE,M\{v/x\}^i,\vec\RR,\PP,\AA',S) 
\tag{\textsc{OQy}$'$}
\\ (m^i\!::\!E::\EE,-,\vec\RR,\PP,\AA,S) \tred{t}{\ret{m(v)}_{OY}}{\!\!\!}'\ (\EE,E[v^i],\vec\RR,\PP,\AA',S)
\tag{\textsc{OAy}$'$}
\end{gather*}
The side-conditions are similar to those for the relation $\tred{t}{}$ between ordinary configurations, with the following exceptions:
in (\textsc{PQy}$'$), if $\Meths(v)=\{m_1,\cdots,m_k\}$ then 
$v'=v\{m_j'/m_j\mid 1\leq j\leq k\}$, for fresh $m_j'$'s, and
$\vec\RR'=\vec\RR\uplus_i\{m_j'\mapsto\lambda x.m_jx\}$; and in
(\textsc{PAy}$'$), if $m\in\dom(\RR_i)$ then $\vec\RR'=\vec\RR\uplus_i\{m_j'\mapsto\lambda x.m_jx\}$, etc. Moreover, in
(\textsc{OQy}$'$) we have that $m\in\dom(\RR_i)$.
Finally, we let
\[
(\vec\CC,\vec\RR,\PP,\AA,S)\pred{N}{(t,x)_{XY}}{\!\!}'\,
(\vec\CC[t\mapsto \CC'],\vec\RR',\PP',\AA',S')
\]
just if 
$(\CC_t,\vec\RR,\PP,\AA,S)\tred{t}{\;x_{XY}\;}{\!\!\!}'\
(\CC',\vec\RR',\PP',\AA',S')$.

We next relate the transition systems induced by external (via $\otimes$) and internal composition (via $\doublewedge$). 
Let us write $(\SS_1,\qweto_1,\FF_1)$ for the transition system induced by external composition of compatible $N$-configurations (so $\qweto_1$ is $\lto$), and $(\SS_2,\qweto_2,\FF_2)$ be the one for internal composition (so $\qweto_2$ is $\pred{N}{}'$). Finality of extended $N$-configurations $(\CC_1\|\cdots\|\CC_N,\vec\RR,\cdots)$ is defined as expected: all $\CC_i$'s must be $([],-)$.
A relation $R\subseteq\SS_1\times\SS_2$ 
is called a \emph{bisimulation} if, for all $(x_1,x_2)\in R$:
\begin{itemize}
\item $x_1\in\FF_1$ iff $x_2\in\FF_2$,
\item if $x_1\qweto_1 x_1'$ then $x_2\qweto_2x_2'$ and $(x_1',x_2')\in R$,
\item if $x_1\xhookrightarrow{(t,x)_{XY}}_1 x_1'$ then $x_2\xhookrightarrow{(t,x)_{XY}}_2x_2'$ and $(x_1',x_2')\in R$,
\item if $x_2\qweto_2 x_2'$ then $x_1\qweto_1 x_1'$  and $(x_1',x_2')\in R$,
\item if $x_2\xhookrightarrow{(t,x)_{XY}}_2 x_2'$ then $x_1\xhookrightarrow{(t,x)_{XY}}_1 x_1'$  and $(x_1',x_2')\in R$.
\end{itemize}
Again, we say that $x_1$ and $x_2$ are \emph{bisimilar}, and write $x_1\sim x_2$, if there exists a bisimulation $R$ such that $(x_1,x_2)\in R$.

\begin{lemma}\label{lem:bisim2}
Let $\rho\asymp_\Pi^w\rho'$ be compatible $N$-configurations. Then, $(\rho\otimes_\Pi^w\rho')\sim(\rho\doublewedge_\Pi^w\rho')$.
\end{lemma}
\begin{proof}
We prove that the relation
$
R =\{ (\rho_1\otimes_\Pi^w\rho_2,\rho_1\doublewedge_\Pi^w\rho_2)\mid \rho_1\asymp_\Pi^w \rho_2\}
$
 is a bisimulation. Let us suppose that $(\rho_1\otimes_\Pi^w\rho_2,\rho_1\doublewedge_\Pi^w\rho_2)\in R$. 
\begin{itemize}
\item 
Let $\rho_1\otimes_\Pi^w\rho_2\xr{(t,x)}\rho_1'\otimes_{\Pi'}^{w'}\rho_2'$ with the transition being due to ({\sc XCall}$_1$), e.g.\
$\rho_1\pred{N}{(1,\call{m(v)})}\rho_1'$ and $\rho_2'=\rho_2$, $w'=1+_1w$ and $\Pi'=\Pi\uplus_1\Meths(v)$, 
$\Meths(v)=\{m_1',\cdots,m_j'\}$,
and recall that $\Meths(v)\cap\Meths(\rho_2)=\emptyset$.
Then, assuming $\rho_1=(\CC_1^1\|\cdots,\RR_1,\PP_1,\AA_1,S_1)$, we have that one of the following holds, for some $\texttt{x}\in\{\k,\l\}$:
\begin{align*}
&\CC_1^1=(\EE_1,E[mv'])\ \text{ and }\ (\CC_1^1,\RR_1,\PP_1,\AA_1,S_1)\tred{1}{\call{m(v)}}\\
&
\qquad(m::E::\EE_1,-,\RR_1\uplus\{m_j'\mapsto\lambda x.m_jx\mid 1\leq j\leq k\},\PP_1\cup_{\mathtt{x}}\Meths(v),\AA_1,S_1)\\
&\CC_1^1=(\EE_1,-)\ \text{ and }\ (\CC_1^1,\RR_1,\PP_1,\AA_1,S_1)\tred{1}{\call{m(v)}}\\
&\qquad\qquad\qquad\qquad\qquad
(m^1::\EE_1,mv,\RR_1,\PP_1,\AA_1\cup_{\mathtt{x}}\Meths(v),S_1)
\end{align*}
In the former case, if $\rho_2=((\EE_2,-)\|\cdots,\RR_2,\PP_2,\AA_2,S_2)$ with $\EE_1\doublewedge^{w_1}\EE_2=(E',\EE)$, we get:
\begin{align*}
&\rho_1\doublewedge^w_\Pi\rho_2=((\EE,E'[E[mv']^1])\|\cdots,\RR_1,\RR_2,\PP,\AA,S)\\
&\pred{N}{(1,\call{m(v)})}{\!\!}'\\
&(m^1\!::E'[E^1]::\EE,-)\|\cdots,\RR_1\uplus\{m_j'\mapsto\lambda x.m_jx\mid 1\leq j\leq k\},\RR_2,\PP'\!,\AA,S)
\end{align*}
with $\PP,\AA$ as in the definition of composition and $\PP'=\PP\cup_{\tt x}\Meths(v)$,
and the latter $N$-configuration equals $\rho_1'\doublewedge^{w'}_{\Pi'}\rho_2$.
The other case is treated in the same manner, and we work similarly 
for ({\sc Retn}$_1$). 
\item On the other hand, if the transition is due to 
({\sc Call}) or ({\sc Retn}) then we work as in the proof of Lemma~\ref{l:comp1}.
\item 
Suppose $\rho_1\doublewedge^w_\Pi\rho_2=(\CC_1\|\cdots,\vec \RR,\PP,\AA,S)\pred{N}{(1,\call{m(v)})}{\!\!\!\!}'$ $(\CC_1'\|\cdots,\vec\RR',\PP',\AA',S)$.
Then, assuming WLOG that $v\in\Meths$, one of the following must be the case,
for some $\texttt{x}\in\{\k,\l\}$ and $i\in\{1,2\}$:
\begin{align*}
&\CC_1=(\EE,E[m^iv'])\ \text{ and }\ (\CC_1,\vec\RR,\PP,\AA,S)\tred{1}{\call{m(v)}}{\!\!\!}'\\\
&\qquad\qquad\quad(m^i\!::E::\EE,\vec\RR\uplus_i\{m_j'\mapsto\lambda x.m_jx\mid 1\leq j\leq k\},\PP\cup_{\mathtt{x}}\Meths(v),\AA,S)\\
&\CC_1(\EE,-)\ \text{ and }\ (\CC_1,\vec\RR,\PP,\AA,S)\tred{1}{\call{m(v)}}{\!\!\!}'\ \\
&\qquad\qquad\qquad\qquad\qquad(m::\EE,M\{v/x\}^i,\vec\RR,\PP,\AA\cup_{\mathtt{x}}\Meths(v),S)
\end{align*}
We only examine the former case, as the latter one is similar, and suppose that $i=1$.
Taking
$\rho_j=(\CC_1^j\|\cdots,\RR_j,\PP_j,\AA_j,S_i)$, 
for $j=1,2$, we have that $(\CC_1^1,\CC_1^2)=((\EE_1,E'[mv'],(\EE_2,-))$,
for some $E,\EE_1,\EE_2$ such that
$\EE_1\doublewedge^{w_1}\EE_2=(E'',\EE)$ and $E=E''[E'^1]$. Moreover, taking 
$\RR_1'=\RR_1\uplus\{m_j'\mapsto\lambda x.m_jx\mid 1\leq j\leq k\}$, $\PP_1'=\PP_1\uplus_{\tt x}\{v\}$, 
$w'=1+_1w$ and $\Pi'=\Pi\uplus\Meths(v)$ (note $\Meths(v)=\{m_1',\cdots,m_k'\}$), 
\[
\rho_1\otimes^w_\Pi\rho_2\xr{(1,\call{m(v)})}
((m::E'::\EE_1,-)\|\cdots,\RR_1',\PP_1',\AA_1,S_1)\otimes^{w'}_{\Pi'}\rho_2=\rho_1'\otimes^{w'}_{\Pi'}\rho_2
\]
and $\rho_1'\doublewedge^{w'}_{\Pi'}\rho_2=(\CC_1'\|\cdots,\vec\RR',\PP',\AA',S)$ as required.
The case for return transitions is similar.
\item On the other hand, if the transition out of $\rho_1\doublewedge^w_\Pi\rho_2$ does not have a label then we work as in the proof of Lemma~\ref{l:comp1}.
\end{itemize}
Moreover, by definition of syntactic composition, $\rho_1\otimes^w_\Pi\rho_2$ is final iff 
$\rho_1\doublewedge^w_\Pi\rho_2$ is.
\end{proof}

Given an $N$-configuration $\rho$ and a history $h$, 
let us write $\rho\Downarrow h$ if $\rho\pred{N}{h}\rho'$ for some final configuration $\rho'$. Similarly if $\rho$ is of the form $(\vec C,\vec \RR,\PP,\AA,S)$. 
We have the following connections in history productions.
The next lemma is  proven in a similar fashion as Lemma~\ref{lem:eta}.

\begin{lemma}\label{lem:eta2}
For any legal $(M_1\|\cdots\|M_N,\RR_1,\RR_2,\PP,\AA,S)$ and history $h$, we have that $(M_1\|\cdots\|M_N,\RR_1,\RR_2,\PP,\AA,S)\Downarrow h$ iff $(M_1\|\cdots\|M_N,\RR_1\cup\RR_2,\PP,\AA,S)\Downarrow h$.
\end{lemma}

\begin{lemma}\label{lem:sched}
For any compatible $N$-configurations $\rho_1\asymp_\Pi^w\rho_2$ and history $h$, 
$(\rho_1\otimes_\Pi^w\rho_2)\Downarrow h$ iff:
\[
  \exists h_1,h_2,\sigma.\
\rho_1\Downarrow h_1\land\rho_2\Downarrow h_2\land h=h_1\doublewedge^\sigma_{\Pi,\Rho}h_2
\]
where $\Rho$ is computed from $\rho_1,\rho_2$ and $\Pi$ as before.
\end{lemma}
\begin{proof}
We show that,
for any compatible $N$-configurations $\rho_1\asymp_\Pi^w\rho_2$ and history suffix $s$, 
$(\rho_1\otimes_\Pi^w\rho_2)\Downarrow s$ iff:
\[
  \exists s_1,s_2,\sigma.\
\rho_1\Downarrow s_1\land\rho_2\Downarrow s_2\land s=s_1\doublewedge^\sigma_{\Pi,\Rho}s_2
\]
where $\Rho$ is computed from $\rho_1,\rho_2$ and $\Pi$ as in the beginning of this section.

The left-to-right direction follows from straightforward induction on the length of the reduction that produces $s$.
For the right-to-left direction, we do induction on the length of $\sigma$. If $\sigma=\epsilon$ then $s_1=s_2=s=\epsilon$. Otherwise, we do a case analysis on the first element of $\sigma$. We only look at the most interesting subcase, namely of $\sigma=0\sigma'$. Then, for some $m\in\Rho$:
\[
s_1=(t,\call{m(v)})s_1'\qquad
s_2=(t,\call{m(v)})s_2'
\]
By $\rho_i\Downarrow s_i$ and $\rho_1\asymp_\Pi^w\rho_2$ we have that $\rho_1\otimes^w_\Pi\rho_2\lto\rho_1'\otimes^{w'}_\Pi\rho_2$, where $w'=0+_tw$ and 
$\rho'_1\asymp_\Pi^{w'}\rho_2'$. Also, $\rho'_i\Downarrow s_i'$ and $s=s_1'\doublewedge^{\sigma'}_{\Pi,\Rho'}s_2'$ so, by IH, $(\rho_1'\otimes_\Pi^{w'}\rho_2')\Downarrow s$.
\end{proof}

We can now prove the correspondence between the traces of component libraries and those of their union.

\paragraph{Theorem~\ref{supercomp}}{\em%
Let $L_1:\Theta_1\to\Theta_2$ and  
$L_2:\Theta_1'\to\Theta_2'$ be libraries 
accessing disjoint parts of the store.
Then,
\[
\sem{L_1\cup L_2}_N
=\{ h\in\clg{H}^L\mid \exists \sigma,h_1\!\in\!\sem{L_1}_N\!, h_2\!\in\!\sem{L_2}_N\!.\,
h=h_1\doublewedge_{\Pi_0,\Rho_0}^{\sigma}h_2
\}\]
with $\Pi_0=\Theta_1\cup\Theta_2\cup\Theta_1'\cup\Theta_2'$ and
$\Rho_0=(\Theta_1\cup\Theta_1')\cap(\Theta_2\cup\Theta_2')$.}
\begin{proof}
Let us suppose $(L_i)\redL^*(\epsilon,\RR_i,S_i)$, for $i=1,2$,
with $\dom(\RR_1)\cap\dom(\RR_2)=\dom(S_1)\cap\dom(S_2)=\emptyset$.
We set:
\begin{align*}
\rho_1 &=(([],-)\|\cdots\|([],-),\RR_1,(\emptyset,\Theta_2),(\Theta_1,\emptyset),S_1)\\ \rho_2 &=(([],-)\|\cdots\|([],-),\RR_2,(\emptyset,\Theta_2'),(\Theta_1',\emptyset),S_2) \end{align*}
We pick these as the initial configurations for  $\sem{L_1}_N$ and $\sem{L_2}_N$ respectively. Then,
$(L_1\cup L_2)\redL^*(\epsilon,\RR_0,S_0)$ where 
$\RR_0=\RR_1\uplus\RR_2$ and $S_0=S_1\uplus S_2$, and
 we take 
\[
\rho_0=(([],-)\|\cdots\|([],-),\RR_0,(\emptyset,\Theta_2\cup\Theta_2'),((\Theta_1\cup\Theta_1')\setminus\Rho_0,\emptyset),S_0)
\]
as the initial $N$-configuration for $\sem{L_1\cup L_2}_N$.
On the other hand, we have
$\rho_1\doublewedge^\epsilon_{\Pi_0}\rho_2= (([],-)\|\cdots\|([],-),\RR_1,\RR_2,(\emptyset,\Theta_2\cup\Theta_2'),((\Theta_1\cup\Theta_1')\setminus\Rho_0,S_0)$. From Lemma~\ref{lem:eta2}, we have that $\rho_0\Downarrow h$ iff $\rho_1\doublewedge^\epsilon_{\Pi_0}\rho_2\Downarrow h$, for all $h$.

Pick a history $h$.
For the forward direction of the claim,  
$\rho_0\Downarrow h$ implies 
$\rho_1\doublewedge^\epsilon_{\Pi_0}\rho_2\Downarrow h$ which, from Lemma~\ref{lem:bisim2}, 
implies $\rho_1\otimes^\epsilon_{\Pi_0}\rho_2\Downarrow h$.
We now use Lemma~\ref{lem:sched} to obtain $h_1,h_2,\sigma$ such that $\rho_i\Downarrow h_i$ and $h=h_1\doublewedge^{\sigma}_{\Pi_0,\Rho_0}h_2$.
Conversely, suppose that
$h_i\in\sem{L_i}_N$ and $h=h_1\doublewedge^{\sigma}_{\Pi_0,\Rho_0}h_2$.
WLOG assume that  $(\Meths(h_1)\cup\Meths(h_2))\cap(\dom(\RR_1)\cup\dom(\RR_2))\subseteq\Pi_0$ (or we appropriately alpha-covert $\RR_1$ and $\RR_2$).
Then, $\rho_i\Downarrow h_i$, for $i=1,2$, and therefore $\rho_1\otimes^\epsilon_{\Pi_0}\rho_2\Downarrow h$ by Lemma~\ref{lem:sched}.
By Lemma~\ref{lem:bisim2} we have that $\rho_1\doublewedge^\epsilon_{\Pi_0}\rho_2\Downarrow h$, which in turn implies that $\rho_0\Downarrow h$, i.e.\ $h\in\sem{L_1\cup L_2}_N$.
\end{proof}

\cutout{
A specific application of the theorem above is for showing that the trace semantics is compositional with respect to library composition ($\comp$). 

\paragraph{Theorem~\ref{supercomp}}{\em%
For all $L'\!:\emptyset\to\Theta,\Theta_1$ and $L\!:\Theta\to\Theta'$,
$\sem{L'\comp L}_N=\{ h\setminus\Theta\mid h\in\clg{H}^L\!\land
\exists \sigma,h_2\in\sem{L}_N, h_1\in\sem{L'}_N.\,
h=h_1\!\doublewedge_{\Pi_0,\Theta}^{\sigma}h_2
\}$,
with $\Pi_0=\Theta\cup\Theta'\cup\Theta_1$.}

\begin{proof}[Proof (sketch)]
By the previous theorem, using also the fact that, for all $L:\Theta_1\to\Theta_2$ and $\Theta\subseteq\Theta_2$, $\sem{L\setminus\Theta}_N=\{h\setminus\Theta\mid h\in\sem{L}_N\}$.
\end{proof}
}


\section{Composition congruence}\label{sec:cong}

\subsection{Proof of Theorem~\ref{thm:gencomp}}\label{sec:gencomp}
\begin{proof}
Assume $L_1\genlin L_2$ and suppose $h_1\in \sem{L\cup L_1}$.
By Theorem~\ref{supercomp}, $h_1 = h'\doublewedge^\sigma_{\Pi,\Rho} h_1'$, where $h'\in\sem{L}$ and $h_1'\in\sem{L_1}$.
Because $L_1\genlin L_2$, there exists $h_2'\in \sem{L_2}$ such that $h_1'\genlin h_2'$, i.e.\ $h_1'\sat{P}{O}^\ast h_2'$.
Note that some of the rearrangements necessary to transform $h_1'$ into $h_2'$ may concern actions shared by $h_1'$ and $h'$;
their polarity will then be different in $h'$. 
Let $h''$ be obtained by applying such rearrangements to $h'$. We claim that $h'\sat{O}{P}^\ast h''$. 
Indeed, suppose that $(t',x')(t,x)_P$ are consecutive in $h_1'$, but swapped in order to obtain $h_2'$, and $(t,x)_P$ appears in $h'$ as $(t,x)_O$. 
Now, the move $(t',x')$ either appears in $h_1$, or it appears in $h'$ and gets hidden in $h_1$. In every case, let $s$ contain the moves of $h'$ that are after $(t',x')$ in the composition to $h_1$, and before $(t,x)_O$. We have that $s(t,x)_O$ is a subsequence of $h'$ and $h'\sat{O}{P}^\ast h''$ holds just if $s$ contains no moves from $t$.
But, if $s$ contained moves from $t$ then the rightmost one such would be some $(t,y)_P$. Moreover, in the composition towards $h_1$, the move would be scheduled with 1. The latter would break the conditions for trace composition as, at that point, the corresponding subsequence of $h_1'$ has as leftmost move in $t$ the P-move $(t,x)_P$. 
We can show similarly that $h'\sat{O}{P}^\ast h''$ holds in the case that the permutation in $h_1'$ is on consecutive moves $(t,x)_O(t',x')$. 
Finally, the rearrangements in $h_1'$ that do not affect moves shared with $h'$ can be treated in a simpler way: e.g.\ in the case of $(t',x')(t,x)_P$ consecutive in $h_1'$ and swapped in $h_2'$, if $(t,x)_P$ does not appear in $h'$ then we can check that $h'$ cannot contain any $t$-moves between $(t',x')$ and $(t,x)$ as the conditions for trace composition impose that only O is expected to play in that part of $h'$ (and any $t$-move would swap this polarity).
\\
Now, since $h'\in\sem{L}$,  Lemma~\ref{lem:genclosure} implies $h''\in\sem{L}$. 
Take $h_2$ to be  $h''\doublewedge^{\sigma'}_{\Pi,\Rho}  h_2'$, where $\sigma'$ is obtained from $\sigma$ following these move rearrangements. 
We then have $h_2\in\sem{L \cup L_2}$. Moreover, $h_1\genlin h_2$ thanks to $h_1'\genlin h_2'$. 
Hence, $h_2\in \sem{L\cup L_2}$ and $h_1\genlin h_2$. Thus,  $L \cup L_1 \genlin L\cup L_2$.
\end{proof}

\subsection{Proof of Theorem~\ref{thm:enccomp}}\label{sec:enccomp}
\begin{proof}
Let us consider the first sequencing case (the second one is dual), and assume that 
$L_1,L_2:\Theta\to\Theta'$ and $L:\Theta''\to\Theta$.
Assume $L_1\enclin L_2$ and suppose $h_1\in \encsem{L\comp L_1}$.
By Theorem~\ref{supercomp}, $h_1 = h'\doublewedge^\sigma_{\Pi,\Rho} h_1'$, where $h'\in\sem{L}$, $h_1'\in\sem{L_1}$ and method calls from $\Theta$ are always scheduled with 0.
The fact that O cannot switch between $\l/\k$ components in (threads of) $h_1$ implies that the same holds for $h',h_1'$, hence $h'\in\encsem{L}$ and $h_1'\in\encsem{L_1}$.
Because $L_1\enclin L_2$, there exists $h_2'\in \encsem{L_2}$ such that $h_1'\genlin h_2'$, i.e.\ $h_1'(\sat{P}{O}\cup\satsym)^\ast h_2'$.
As before, some of the rearrangements necessary to transform $h_1'$ into $h_2'$ may concern actions shared by $h_1'$ and $h'$;
we need to check that these can lead to compatible $h''\in\encsem{L}$.
Let $h''$ be obtained by applying such rearrangements to $h'$. 
We claim that $h'\sat{O}{P}^\ast h''$. 
The transpositions covered by $\sat{P}{O}$ are treated as in Lemma~\ref{thm:gencomp}.
Suppose now that $(t',x')_{P\k}(t,x)_{O\l}$ are consecutive in $h_1'$ but swapped in order to obtain $h_2'$, and $(t,x)_{O\l}$ appears in $h'$ as $(t,x)_{P\k}$. 
Now, the move $(t',x')$ cannot appear in $h'$ as it is in $L_1$'s $\k$-component ($L$ is the $\l$-component of $L_1$). Let $s$ contain the moves of $h'$ that are after $(t',x')$ in the composition to $h_1$, and before $(t,x)_{P\k}$. We 
claim that $s$ contains no moves from $t$, so $h'$ can be directly composed with $h_2'$ as far as this transposition is concerned.
Indeed, if $s$ contained moves from $t$ then, taking into account the encapsulation conditions, the leftmost one such would be some $(t,y)_{O\k}$. 
But the $\k$-component of $L$ is $L_1$, which contradicts the fact that the moves we consider are consecutive in $h_1'$.
Hence,
taking $h_2$ to be  $h''\doublewedge^{\sigma'}_{\Pi,\Rho}  h_2'$, where $\sigma'$ is obtained from $\sigma$ following the $\sat{P}{O}$ move rearrangements,
we have $h_2\in\encsem{L \comp L_2}$ and $h_1\enclin h_2$.
Thus,  $L \comp L_1 \enclin L\comp L_2$.
\\
The case of $L \uplus L_1\enclin L \uplus L_2$ is treated in a similar fashion. In this case, because of disjointness, the moves transposed in $h_1'$ do not have any counterparts in $h'$. Again, we consider 
consecutive moves
$(t',x')_{P\k}(t,x)_{O\l}$ in $h_1'$ that are swapped in order to obtain $h_2'$.
Let $s$ contain the moves of $h'$ that are after $(t',x')$ in the composition to $h_1$, and before $(t,x)$. 
As $\Theta_1,\Theta_1'$ is first-order, $(t,x)_{O\l}$ must be a return move and the $t$-move preceding it in $h_1$ must be the corresponding call. The latter is a move in $h_1'$, which therefore implies that
there can be no moves from $t$ in $s$. Similarly for the other transposition case.
\end{proof}

\subsection{Proof of Theorem~\ref{thm:relcomp}}\label{sec:relcomp}
\begin{proof}
For the first claim,
suppose $L$ is $\relation\choose\clg{G}$-closed and $L_1\rellin L_2$. Consider $h_1\in \encsem{L\comp L_1}$.
By Theorem~\ref{supercomp}, $h_1=h'\doublewedge^\sigma_{\Pi,\Rho} h_1'$, where $h'\in\encsem{L}$ and $h_1'\in\encsem{L_1}$.
Since $L_1\rellin L_2$, there exists $h_2'\in\encsem{L_2}$ such that $(h_1'\restriction \k)\genlin (h_2'\restriction \k)$
and $(\overline{h_1'}\restriction \l)\,\relation\,(\overline{h_2'}\restriction \l)$.
By the permutation-closure of $\relation$, we can pick $h_2'$ not to contain any common names with $h'$ apart from those in the common moves of $h_1'$ and $h'$.
Because $L$ is $\relation\choose\clg{G}$-closed, $h'\in\encsem{L}$ and $(h'\restriction \l) = (\overline{h_1'}\restriction \l)$ and 
$(\overline{h_1'}\restriction \l)\,\relation\,(\overline{h_2'}\restriction \l)$,
we can conclude that there exists $h''\in\encsem{L}$ such that $(h''\restriction \k) = (\overline{h_2'}\restriction \l)$ and 
$(\overline{h'}\restriction \l)\,\,\clg{G}\,\,(\overline{h''}\restriction \l)$.
Applying the corresponding rearrangements to $\sigma$, we have that $h''$ and $h_2'$ are compatible, i.e.\ $(h'' \doublewedge^{\sigma'}_{\Pi,\Rho} h_2') \in \encsem{L\comp L_2}$.
Let $h_2=h''\doublewedge^{\sigma'}_{\Pi,\Rho} h_2'$. We want to show $h_1\rellin h_2$.
To that end, it suffices to make the following observations.
\begin{compactitem}
\item  We have $(h_1\restriction \k) \genlin (h_2\restriction \k)$ because
$(h_1\restriction \k)= (h_1'\restriction \k)$, $(h_1'\restriction \k)  \genlin (h_2'\restriction \k)$ and $(h_2'\restriction \k) = (h_2\restriction \k)$.
\item We have $(\overline{h_1}\restriction \l)\,\clg{G}\, (\overline{h_2}\restriction \l)$ because
$(\overline{h_1}\restriction \l)=(\overline{h'}\restriction \l)$,
$(\overline{h'}\restriction \l)\,\clg{G}\,(\overline{h''}\restriction \l)$
and $(\overline{h''}\restriction \l) = (\overline{h_2}\restriction \l)$.
\end{compactitem}
Consequently $L\comp L_1\sqsubseteq_{\clg{G}} L\comp  L_2$.
\\
Suppose now $L_1\rellin L_2$. Consider $h_1\in\encsem{L_1\comp L}$, i.e. $h_1=h_1'\doublewedge^{\sigma}_{\Pi,\Rho} h'$, where $h_1'\in\encsem{L_1}$ and $h'\in\encsem{L}$.
Because $L_1\rellin L_2$, there exists $h_2'\in \encsem{L_2}$ such that $h_1'\rellin h_2'$, i.e. $(h_1'\restriction \k)\genlin (h_2'\restriction \k)$ and
$(\overline{h_1'}\restriction \l)\relation (\overline{h_2'}\restriction \l)$.
Define $h''$ to be $h'$ in which $(\overline{h'} \restriction \l) = (h_1'\restriction \k)$ was modified by applying the same rearrangements as those 
witnessing $(h_1'\restriction \k)\genlin (h_2'\restriction \k)$. Consequently $h'\sat{P}{O}^\ast h''$. By Lemma~\ref{lem:encclosure}, $h''\in\encsem{L}$.
Moreover, $(\overline{h''}\restriction \l)= (h_2'\restriction \k)$. Consequently, $h_2'$ and $h''$ are compatible for the corresponding $\sigma'$. Let $h_2=h_2'\doublewedge^{\sigma'}_{\Pi,\Rho} h''\in \sem{L_2\comp L}$.
Then we get:
\begin{compactitem}
\item $(h_1\restriction \k) = (h'\restriction \k) = (h''\restriction \k) = (h_2\restriction \k)$;
\item $(\overline{h_1}\restriction \l) =(\overline{h_1'}\restriction \l)$,  $(\overline{h_1'}\restriction \l)\relation (\overline{h_2'}\restriction \l)$,
$(\overline{h_2'}\restriction \l)=(\overline{h_2}\restriction \l)$.
\end{compactitem} 
Consequently $h_1\rellin h_2$ and, hence, $L_1\comp L\rellin L_2\comp L$.
\\
For the last claim, we observe that because of the type-restrictions, the elements of $\encsem{L\uplus L_i}$ are interleavings of histories from $\encsem{L}$ and $\encsem{L_i}$.
Consider now $h_1\in\encsem{L\uplus L_1}$, i.e.\ $h_1=h'\doublewedge^{\sigma}_{\Pi,\Rho} h_1'$, where $h_1'\in\encsem{L_1}$ and $h'\in\encsem{L}$, and let 
$h_2'\in \encsem{L_2}$ be such that $h_1'\rellin h_2'$. From our previous observation, we have that $h'$ can still be composed with $h_2'$, for appropriate $\sigma'$.
Thus, taking $h_2=h'\doublewedge^{\sigma'}_{\Pi,\Rho}h_2'$, we have $h_2\in\encsem{L\uplus L_2}$ and, moreover, $h_1'\rellin h_2'$ implies $h_1\rellin[\relation^+]h_2$.
\end{proof}


\end{document}
